%% file: zoom.tex
\documentclass[english,11pt,oneside,reqno,a4paper]{amsart}

\usepackage[utf8]{inputenc}
\usepackage[english]{babel}

\usepackage{amssymb}
\usepackage{amsmath}
\usepackage{amsthm}
\usepackage{amsfonts}
\usepackage{mathtools}
\usepackage{dsfont}

\usepackage{tabularx}
\usepackage{booktabs}
\usepackage{multirow}
\usepackage[table]{xcolor}
\usepackage{lmodern}
\usepackage[noadjust]{cite}
\usepackage{calc}
\usepackage{enumitem}

\usepackage{todonotes}
\usepackage{footmisc}
\usepackage{subfig}

\usepackage[colorlinks=false, pdftitle={Distance-preserving graph contractions}]{hyperref}
\hypersetup{pdfauthor={Bernstein, D\344ubel, Disser, Klimm, M\374tze, Smolny}}

\input{Figures/tikz_styles}

\usepackage{pgfplots}
\pgfdeclarelayer{background}
\pgfdeclarelayer{foreground}
\pgfsetlayers{background,main,foreground}

\usepackage{macros}

\theoremstyle{plain}
\newtheorem{thm}{Theorem}
\newtheorem{lem}{Lemma}[section]
\newtheorem{cor}[thm]{Corollary}

\makeatletter
\g@addto@macro{\endabstract}{\@setabstract}
\newcommand{\authorfootnotes}{\renewcommand\thefootnote{\@fnsymbol\c@footnote}}%
\makeatother

\begin{document}

\begin{center}

\LARGE 
Distance-preserving graph contractions
\renewcommand{\thefootnote}{\fnsymbol{footnote}}
\footnote{An extended abstract of this work has appeared in the Proceedings of the 9th Innovations in Theoretical Computer Science Conference (ITCS) 2018 \cite{DBLP:conf/innovations/BernsteinDDKMS18}.}
\par \bigskip

\normalsize
\authorfootnotes
Aaron Bernstein\textsuperscript{1}\footnote{E-Mail: \texttt{bernstei@gmail.com}},
Karl D\"aubel\textsuperscript{1}\footnote{\label{mail}E-Mail: \texttt{\{daeubel,muetze,smolny\}@math.tu-berlin.de}},
Yann Disser\textsuperscript{2}\footnote{Supported by the `Excellence Initiative' of the German Federal and State Governments and the Graduate School~CE at TU~Darmstadt. E-Mail: \texttt{disser@mathematik.tu-darmstadt.de}}, \\
Max Klimm\textsuperscript{3}\footnote{E-Mail: \texttt{max.klimm@hu-berlin.de}},
Torsten M\"utze\textsuperscript{1}\footref{mail} and
Frieder Smolny\textsuperscript{1}\footref{mail} \par \bigskip

\textsuperscript{1}Institut f\"{u}r Mathematik, TU Berlin \par
\textsuperscript{2}Department of Mathematics, Graduate School CE, TU Darmstadt \par
\textsuperscript{3}Wirtschaftswissenschaftliche Fakult\"{a}t, HU Berlin \par
\par \bigskip

\end{center}

\begin{abstract}
Compression and sparsification algorithms are frequently applied in a preprocessing step before analyzing or optimizing large networks/graphs.
In this paper we propose and study a new framework contracting edges of a graph (merging vertices into super-vertices) with the goal of preserving pairwise distances as accurately as possible.
Formally, given an edge-weighted graph, the contraction should guarantee that for any two vertices at distance $d$, the corresponding super-vertices remain at distance at least $\varphi(d)$ in the contracted graph, where $\varphi$ is a tolerance function bounding the permitted distance distortion.
We present a comprehensive picture of the algorithmic complexity of the contraction problem for affine tolerance functions $\varphi(x)=x/\alpha-\beta$, where $\alpha\geq 1$ and $\beta\geq 0$ are arbitrary real-valued parameters.
Specifically, we present polynomial-time algorithms for trees as well as hardness and inapproximability results for different graph classes, precisely separating easy and hard cases.
Further we analyze the asymptotic behavior of contractions, and find efficient algorithms to compute (non-optimal) contractions despite our hardness results.
\end{abstract}

\section{Introduction}
\label{sec:intro}

When dealing with large networks, it is often beneficial to compress or sparsify the data to manageable size before analyzing or optimizing the network directly.
To be useful, a meaningful compression should represent salient features of the original network with good approximation, while being much smaller in size.
In this paper, we focus on a compression of undirected edge-weighted graphs that approximately maintains all distances between vertices in the graph.

In this context, an extensively studied concept are \emph{spanners} (e.g.\ \cite{MR982872,MR1184695,MR2298319,MR3536579}).
Given an undirected graph $G=(V,E)$ and real numbers $\alpha \geq 1$ and $\beta \geq 0$, a subgraph $H=(V,E')$, $E'\subseteq E$, is an \emph{$(\alpha,\beta)$-spanner of~$G$} if $\dist_{H}(u,v) \leq \alpha\cdot \dist_G(u,v) + \beta$ holds for all $u,v \in V$.
While the number of edges in a spanner may be much smaller than that of the original graph, the number of vertices is the same for both, leaving further potential for compression untapped.
For illustration, consider the road network of Europe with about 50~million vertices \cite{MR3077319}, any spanner of which must again have about 50 million vertices and edges.
However, to approximately represent distances in Europe's road network one may also merge nearby vertices into super-vertices, thus achieving a much better compression of the network.
This is akin to the visual process of zooming out of a graphical representation of the map, where neighbored vertices fade into each other and edges between merged vertices vanish.
At a large enough zoom level, the entire network merges into a single vertex.

In this paper we propose and study a new framework for contracting networks that formalizes this intuitive idea and makes it applicable to general graphs.
Specifically, we study a contraction problem on graphs where a subset of edges $C \subseteq E$ is contracted.
We denote by $G/C$ the resulting simple graph obtained from $G$ by contracting the edges in $C$ and by deleting resulting loops and multiple edges, keeping only the minimum length edge between any two vertices.
For any two vertices in $G$, we compare their distance in $G$ with the distance of the corresponding super-vertices in $G/C$.

It is interesting to contrast this concept with graph spanners.
When constructing a spanner, the length of the removed edges is implicitly set to~$\infty$, resulting in an overall increase of distances.
On the other hand, a contraction implicitly sets the length of the contracted edges to zero, leading to an overall decrease of distances.
For both problems, the ultimate goal is to reduce the complexity of the network while maintaining an approximation guarantee on the distances.

The following example shows that contractions may be better suited than spanners to achieve this goal.
In a subgraph with small radius, a spanner can at best result in a spanning tree of the same order, while a contraction can reduce the whole subgraph to a single vertex, while entailing a multiplicative distance distortion of similar magnitude.
In addition, the contraction may also merge many edges entering the contracted subgraph.
Clearly, the objective here is to maximize the total number of contracted and deleted edges, as this minimizes the memory required to represent the resulting network in a computer (using e.g.\ adjacency lists).

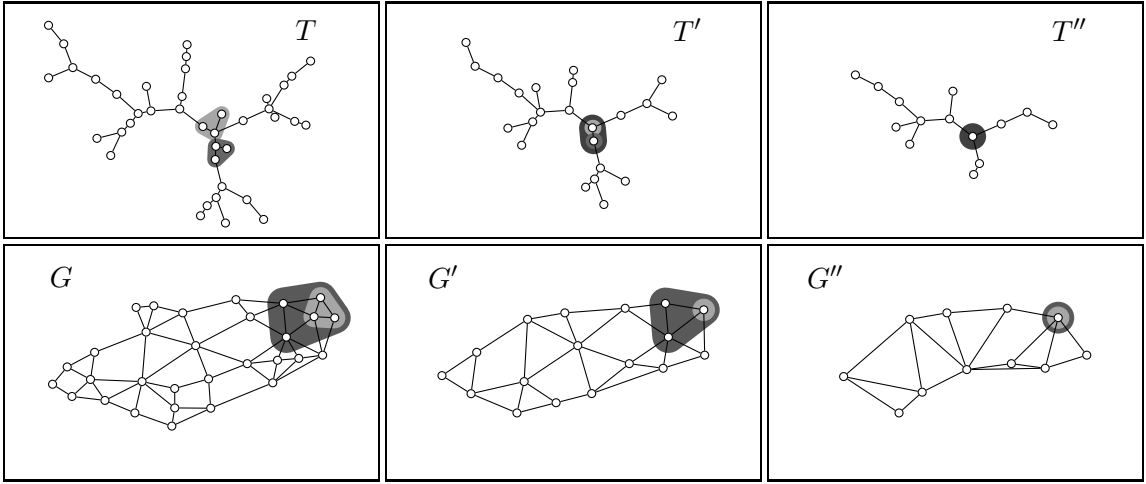
\begin{figure}
\noindent%
\centering%
\fbox{%
\begin{minipage}[c][2.9cm][c]{\dimexpr.33\linewidth-2\fboxsep-2\fboxrule}%
\centering%
\input{Figures/ex_tree1}
\end{minipage}}%
\hfill%
\fbox{%
\begin{minipage}[c][2.9cm][c]{\dimexpr.33\linewidth-2\fboxsep-2\fboxrule}%
\centering%
\input{Figures/ex_tree2}
\end{minipage}}%
\hfill%
\fbox{%
\begin{minipage}[c][2.9cm][c]{\dimexpr.33\linewidth-2\fboxsep-2\fboxrule}%
\centering%
\input{Figures/ex_tree3}
\end{minipage}}\\%
\vspace{\dimexpr.0025\linewidth}%
\fbox{%
\begin{minipage}[c][2.9cm][c]{\dimexpr.33\linewidth-2\fboxsep-2\fboxrule}%
\centering%
\input{Figures/ex_graph1}
\end{minipage}}%
\hfill%
\fbox{%
\begin{minipage}[c][2.9cm][c]{\dimexpr.33\linewidth-2\fboxsep-2\fboxrule}%
\centering%
\input{Figures/ex_graph2}
\end{minipage}}%
\hfill%
\fbox{%
\begin{minipage}[c][2.9cm][c]{\dimexpr.33\linewidth-2\fboxsep-2\fboxrule}%
\centering%
\input{Figures/ex_graph3}
\end{minipage}}
\captionsetup{width=.9\textwidth}
\caption{Top: two iterations of \Contraction{} with $\varphi(x)=4x/5 -3$ on a tree; bottom: two iterations of \Contraction{} with $\varphi(x) = 3x/4 - 3$ on a planar graph.
Distances are geometric and some contracted sets of vertices are highlighted.}
\label{fig:contraction}
\end{figure}

Given the results presented in this paper and the known results for spanners (discussed in detail below), we further believe that the combination of spanners and contractions is very powerful, promising and flexible.
As the former only increases and the latter only decreases the distances, the respective distortion guarantees provably also hold for the overall distortion.
In fact, both effects may even compensate each other.
This is true \emph{regardless} of the order in which both compression operations are applied, even when they are applied repeatedly.

In order to measure the distance distortion of the contraction, we assume a non-decreasing tolerance function $\varphi\colon \RR \to \RR$, similar to the corresponding function for spanners, see e.g.~\cite{MR2298319}.
We are interested in computing contractions that preserve distances in the following sense:
For any two vertices~$u$ and~$v$ at distance~$d$ in~$G$, the distance of the corresponding vertices in the contracted graph~$G/C$ must be at least~$\varphi(d)$.
If this condition is satisfied, we call~$C$ a \emph{$\varphi$-distance preserving contraction}, or \emph{$\varphi$-contraction} for short.
Formally, the algorithmic problem \Contraction{} considered in this paper is to compute for a given graph $G = (V,E)$ with edge lengths $\ell\colon E\to \RR_{>0}$ and a given tolerance function~$\varphi$, a $\varphi$-contraction $C \subseteq E$ such that the number of contracted and deleted edges is maximized.
We are specifically interested in the case where the tolerance function~$\varphi$ is an affine function $\varphi(x)=x/\alpha-\beta$ for real-valued parameters $\alpha \geq 1$ and $\beta \geq 0$.
We then simply write \emph{$(\alpha,\beta)$-contraction} instead of $\varphi$-contraction.
See Figure~\ref{fig:contraction} for some example instances of the problem \Contraction{}.

When considering the case of a purely multiplicative error ($\beta=0$), a slight subtlety has to be taken into account.
Specifically, for a graph with positive edge lengths it is not feasible to contract a single edge.
Therefore, we propose a slight modification of our original model:
We say that a set $C\subseteq E$ of edges of $G$ is a \emph{weak $\varphi$-distance preserving contraction}, or \emph{weak $\varphi$-contraction} for short, if it does not contract the entire graph and, for any two vertices $u$ and $v$ at distance $d$ in $G$, the distance of the corresponding vertices in $G/C$ is either zero or at least $\varphi(d)$.
We will refer to the corresponding algorithmic problem as \WContraction{}.
Put differently, in a weak contraction, the distances between different super-vertices satisfy the given distortion guarantee, but for vertices belonging to the same super-vertex, no guarantee is given.

\subsection{Our results}
\label{sec:results}

In this paper, we present a comprehensive picture of the algorithmic complexity of the described contraction problems.
Recall that we are given an input graph with edge lengths and tolerance function $\varphi$, and our goal is to compute a (weak) contraction that maximizes the total number of contracted and deleted edges.
Our main results concern affine tolerance functions $\varphi(x)=x/\alpha-\beta$ with parameters $\alpha\geq 1$ and $\beta\geq 0$.
For the reader's convenience, our results are summarized in Tables~\ref{tab:results} and~\ref{tab:results-asymp}.
Within the tables and throughout this paper, $n$ and $m$ denote the number of vertices and edges, respectively, of the input graph under consideration.

\newcommand{\myline}[1]{\raisebox{1.5mm}{\underline{\hspace{#1}}}}
\newcolumntype{P}[1]{>{\centering\arraybackslash}p{#1}}

\begin{table}[tb]
\caption{Overview of algorithmic and hardness results presented in this paper.
\label{tab:results}
}
\footnotesize
\begin{minipage}{\textwidth}
\begin{tabular*}{\linewidth}{
@{}l
@{\extracolsep{\fill}} P{2.0cm}
@{\extracolsep{0.6ex}}   P{2.0cm}
@{\extracolsep{0.6ex}}   P{2.35cm}
@{\extracolsep{0.6ex}}   P{2.35cm}
}
\toprule
Problem &
\multicolumn{4}{c}{\myline{3.9cm}\,Graph classes\,\myline{3.9cm}} \\
& 
\multicolumn{1}{c}{Path} &
\multicolumn{1}{c}{Tree} & 
\multicolumn{1}{c}{Cycle} & 
\multicolumn{1}{c}{General} \\
\Contraction{} &\\[1mm]
addit.\ {\scriptsize($\alpha\!\!=\!\!1$)}, unit lg. \rule{0pt}{1.1em}&
\multicolumn{1}{c}{\cellcolor{green!50}} &
\multicolumn{1}{c}{\cellcolor{green!50}$\cO(n)$~\scriptsize{[Th.~\ref{thm:unit-tree}]}} &
\multicolumn{1}{c}{\cellcolor{green!50}} & \multicolumn{1}{c}{\cellcolor{red!40}$m^{\frac{1}{2}\!-\!\varepsilon}$-inapx.\footnote{even for bipartite graphs and $\beta=1$} \scriptsize{[Th.~\ref{thm:inapx-bip}]}} \\

&
\multicolumn{1}{c}{\cellcolor{green!50}}&
&
\multicolumn{1}{c}{\cellcolor{green!50}}&
\multicolumn{1}{c}{\cellcolor{red!40}}\\[-2ex]

affine {\scriptsize($\alpha,\beta$)}, unit lg. \rule{0pt}{1.1em}&
\multicolumn{1}{c}{\cellcolor{green!50}$\cO(n)$ \scriptsize{[Th.~\ref{thm:path}]}} &
\multicolumn{1}{c}{\cellcolor{green!40}}&
\multicolumn{1}{c}{\cellcolor{green!50}$\cO(n)$~\scriptsize{[Th.~\ref{thm:cycle}]}} &
\multicolumn{1}{c}{\cellcolor{red!40}} \\

&
&
\multicolumn{1}{c}{\cellcolor{green!40}}&
&
\\[-2ex]

addit.\ {\scriptsize($\alpha\!\!=\!\!1$)} \rule{0pt}{1.1em}&
\multicolumn{2}{c}{\cellcolor{green!40}} &
\multicolumn{1}{c}{\cellcolor{red!30}NP-hard~\scriptsize{[Th.~\ref{thm:cycle-hard-fix}]}} &
\multicolumn{1}{c}{\cellcolor{red!50}$n^{1-\varepsilon}$-inapx.~\scriptsize{[Th.~\ref{thm:inapx-contr}]}} \\

&
\multicolumn{2}{c}{\cellcolor{green!40}} &
\multicolumn{1}{c}{\cellcolor{red!30}} &
\multicolumn{1}{c}{\cellcolor{red!50}}\\[-2ex]

affine {\scriptsize($\alpha,\beta$)} \rule{0pt}{1.1em}&
\multicolumn{2}{c}{\cellcolor{green!40}$\cO(n^3)$~\scriptsize{[Th.~\ref{thm:dp}]}} &
\multicolumn{1}{c}{\cellcolor{red!30}} &
\multicolumn{1}{c}{\cellcolor{red!50}} \\

\midrule
\WContraction{} &\\[1mm]

additive {\scriptsize($\alpha\!\!=\!\!1$)} \rule{0pt}{1.1em}&
\multicolumn{2}{c}{\cellcolor{green!30}} & 
\multicolumn{2}{l}{\cellcolor{red!30}NP-hard\footnote{also NP-hard for planar graphs with arb.\ large girth, $(\alpha,\beta)=(2,0)$, and unit lg.\ ($\ell=1$) [Th.~\ref{thm:weak-planar}].} \scriptsize{[Th.~\ref{thm:cycle-hard-fix}]}} \\
&\multicolumn{2}{c}{\cellcolor{green!30}}& \multicolumn{1}{c}{\cellcolor{red!30}}&\\[-2ex]

affine {\scriptsize($\alpha,\beta$)} \rule{0pt}{1.1em}&
\multicolumn{2}{c}{\cellcolor{green!30}$\cO(n^5)$~\scriptsize{[Th.~\ref{thm:weak-dp}]}} &
\multicolumn{1}{c}{\cellcolor{red!30}} &
\multicolumn{1}{c}{\cellcolor{red!50}$n^{1-\varepsilon}$-inapx.\footnote{even if $(\alpha,\beta)=(3/2,0)$.} \scriptsize{[Th.~\ref{thm:inapx-weak}]}} \\
\bottomrule
\end{tabular*}
\end{minipage}
\end{table}

\subsubsection*{Algorithmic results}
We develop linear time greedy algorithms for \Contraction{} with unit lengths on paths and cycles for general $\alpha$ and $\beta$, as well as on trees with $\alpha = 1$ (Theorems~\ref{thm:path},~\ref{thm:cycle} and~\ref{thm:unit-tree}).
The first two algorithms are inspired by LP rounding techniques, the latter algorithm relies on a structural characterization of optimal solutions.

We present dynamic programming algorithms solving \Contraction{} and \WContraction{} on trees in time $\cO(n^3)$ or $\cO(n^5)$, respectively (Theorems~\ref{thm:dp} and~\ref{thm:weak-dp}).
These dynamic programs compute optimal solutions on subtrees, in the latter case combining several Pareto optimal solutions in a two-dimensional parameter space (hence the larger running time).

Note that instead of maximizing the number of contracted and deleted edges, we could optimize for $\alpha$ or $\beta$ while fixing the other parameters.
The resulting problems are polynomially equivalent to our setting, via binary search over one of the parameters.

\subsubsection*{Hardness results}
We complement these algorithms by several hardness results.
First we consider the purely additive case where $\alpha=1$.
We show that here both \Contraction{} and \WContraction{} are NP-hard on cycles for any fixed $\beta>0$, by a reduction of a variant of \Partition{} (Theorem~\ref{thm:cycle-hard-fix}).
As mentioned before, both problems can be solved efficiently on graphs without cycles, and there is a linear time algorithm for \Contraction{} on cycles with unit lengths.
By reductions from \Clique{} we show that both the general as well as the unit lengths case of \Contraction{} with $\alpha=1$ are hard to approximate within factors of $n^{1-\varepsilon}$ or $m^{1/2-\varepsilon}$, respectively (Theorem~\ref{thm:inapx-contr} and Theorem~\ref{thm:inapx-bip}).

Further we consider the purely multiplicative case where $\beta=0$ (here \Contraction{} is trivial).
We show that in this case \WContraction{} is NP-hard on planar graphs with arbitrarily large girth and unit length edges by a reduction from a special case of \textsc{Planar 3SAT} (Theorem~\ref{thm:weak-planar}).
Since these graphs are locally tree-like, this result constitutes another rather sharp separation from the polynomially solvable tree case.
Furthermore, we show that the problem is hard to approximate within a factor of $n^{1-\varepsilon}$ by a reduction from \IndSet{} (Theorem~\ref{thm:inapx-weak}).

\begin{table}[tb]
\caption{Overview of asympotic bounds presented in this paper.
\label{tab:results-asymp}
}
\footnotesize
%\small
\begin{minipage}{\textwidth}
\centering
\begin{tabular}{
@{}l
@{\extracolsep{1ex}}   P{3.5cm}
@{\extracolsep{1ex}}   P{2.5cm}
@{\extracolsep{1ex}}   P{2.5cm}
}
\toprule
\Contraction{} with unit lg.\ {\scriptsize($\ell\!\!=\!\!1$)} &
\# of edges in $G/C$ &
Time & 
Reference \\[1mm]
\rowcolor{green!50}
\multicolumn{1}{@{}l}{\cellcolor{white}$(\alpha, \beta)=(2k -1, 1)$}&
\rule{0pt}{1.1em}
$n^{1+1/k}$ &
$\cO(m)$ &
[Th.~\ref{thm:asymp-mult}] \\
&
&
&\\[-2ex]
\rowcolor{green!50}
\multicolumn{1}{@{}l}{\cellcolor{white}$(\alpha, \beta)=(2\log_2 n -1, 1)$}&
\rule{0pt}{1.1em}
$2n$ &
$\cO(m)$ &
[Cor.~\ref{cor:asymp-mult}] \\
&
&
&\\[-2ex]
\rowcolor{red!50}
\multicolumn{1}{@{}l}{\cellcolor{white}$(\alpha, \beta)=(k -1, 1)$}&
\rule{0pt}{1.1em}
$\Omega(n^{1+1/k})$ &
--- &
[Th.~\ref{thm:asymp-girth}] \\
&
&
&\\[-2ex]
\rowcolor{green!50}
\multicolumn{1}{@{}l}{\cellcolor{white}$(\alpha, \beta)=(1, k)$}&
\rule{0pt}{1.1em}
$m - km/(2n)$ &
$\cO(m)$ &
[Th.~\hyperref[thm:asymp-add-first]{\ref*{thm:asymp-add}~\ref*{thm:asymp-add-first}}]\\
&
&
&\\[-2ex]
\rowcolor{green!50}
\multicolumn{1}{@{}l}{\cellcolor{white}$(\alpha, \beta)=(1, k)$}&
\rule{0pt}{1.1em}
$\cO(n^2/k)$ &
$\cO(m)$ &
[Th.~\hyperref[thm:asymp-add-high]{\ref*{thm:asymp-add}~\ref*{thm:asymp-add-high}}] \\
&
&
&\\[-2ex]
\rowcolor{red!50}
\multicolumn{1}{@{}l}{\cellcolor{white}$(\alpha, \beta)=(1, \cO(1))$}&
\rule{0pt}{1.1em}
$\Omega(n^{4/3 - o(1)})$ &
--- &
\cite{MR3536579} \\
\midrule
\Contraction{} with unit lg.\ {\scriptsize($\ell\!\!=\!\!1$)}\\ \ and min.~degree $D$&
\# of vertices in $G/C$ &
Time & 
Reference \\
\rowcolor{green!50}
\multicolumn{1}{@{}l}{\cellcolor{white}$(\alpha, \beta)=(5, 1)$}&
\rule{0pt}{1.1em}
$n/D$ &
$\cO(m)$ &
[Th.~\ref{thm:min-degree}] \\
&
&
&\\[-2ex]
\rowcolor{red!50}
\multicolumn{1}{@{}l}{\cellcolor{white}$(\alpha, \beta)=(k, 1)$}&
\rule{0pt}{1.1em}
$n/((k+1) D)$ &
--- &
[Th.~\ref{thm:asymp-degree}] \\
\bottomrule
\end{tabular}
\end{minipage}
\end{table}

\subsubsection*{Asymptotic bounds}
We now discuss our asymptotic bounds for contractions.
In this setting, we are interested in (non-optimal) contractions for graphs with unit lengths that can be computed efficiently despite the above-mentioned hardness results.
We prove that for any $k\geq 1$ any graph $G$ has a $(2k-1,1)$-contraction $C$ such that $G/C$ has at most $n^{1+1/k}$ edges, and such a contraction can be computed in time $\cO(m)$ (Theorem~\ref{thm:asymp-mult}) by successively growing clusters around center vertices.
Assuming Erd\H{o}s' girth conjecture, we show a corresponding (not tight) lower bound (Theorem~\ref{thm:asymp-girth}).

For a purely additive error, we observe two simple $(1,k)$-contractions that can be computed in $\cO(m)$ time (Theorem~\ref{thm:asymp-add}).
We show that for any even integer $0 \leq k \leq n$, the edges incident to the $k/2$ vertices of highest degrees form a $(1,k)$-contraction with objective value at least $km/(2n)$, which is asymptotically best possible for paths.
Another $(1, k)$-contraction $C$ is implicitly used by Bernstein and Chechik in their faster deterministic algorithm for dynamic shortest paths in dense graphs \cite{MR3536582}.
For any number $0 < k \leq n$, it consists of the edges incident to two vertices of degree at least $n/k$, and $G/C$ has $\cO(n^2/k)$ edges.
Both of these contractions can be computed in $\cO(m)$ time.
Further we note that the main result in~\cite{MR3536579} implies that for all $\varepsilon > 0$, any contraction $C$ such that $G/C$ has $\cO(n^{4/3 - \varepsilon})$ edges does not admit a constant additive error.

One possible advantage of contraction compared to spanners is the potentially significant reduction of \emph{vertices} as well as edges, e.g. reducing the complexity of performing algorithmic tasks in the smaller graph.
To ground this intuition, we exhibit a contraction that significantly reduces the number of vertices in any graph with minimum degree $D$ to $\cO(n/D)$ (Theorem~\ref{thm:min-degree}). We also present a lower bound (Theorem~\ref{thm:asymp-degree}) showing that we cannot guarantee $o(n/D)$ vertices, even if we allow larger approximation error.

\subsection{Comparison with previous results}

There are several models aiming to compress graphs while preserving distances.
They differ by their choice of compression operation, such as replacing the graph by a subgraph or minor, and by whether the aim is to preserve all or only certain distances.

As discussed before, graph spanners are a concept closely related to contractions, where the length of removed edges is set to $\infty$ rather than to $0$.
Our results highlight further intrinsic similarities of the two models.
Like contractions, spanners are NP-hard to compute optimally (see \cite{MR982872,MR1232326}).
While the spanner literature considers the problem of minimizing the number of remaining edges, we analyze the objective of maximizing the number of contracted edges, prohibiting a direct comparison of the respective inapproximability results.
We note however that approximation algorithms for spanner problems have been studied extensively, even though strong lower bounds are known.
For instance, computing $(2,0)$-spanners in unweighted graphs is $\Theta(\log n)$-hard to approximate (\cite{MR1291540,MR1822929}); for further references see e.g.~\cite{MR3627763}.

Despite these negative results, it is still possible to obtain powerful asymptotic guarantees in both models.
In particular, our $(2k-1, 1)$-contraction with $\cO(n^{1+1/k})$ edges for unweighted graphs has a clear analogy to the classic $(2k-1,0)$-spanner with the same number of edges~\cite{MR1184695} (note that the additive error of 1 in our result is strictly necessary, as discussed above).
There is, however, a major difference between the two results: whereas the $(2k-1,0)$-spanner can trivially be shown to be optimal assuming Erd\H{o}s' girth conjecture, applying this conjecture to the contraction model only yields a lower bound of $n^{1+1/(2k)}$ edges for a $(2k-1,1)$-contraction. Closing this gap thus remains as an interesting open problem in the contraction model, whose solution would likely yield further insight into the relationship to spanners.

Halperin and Zwick showed how an optimal $(2k-1,0)$-spanner can be constructed in linear time (see~\cite{MR2080715}). We achieve the same running time for our $(2k-1,1)$-contraction.
It is interesting to note that the clustering yielding our $(2k-1,1)$-contraction was previously used in~\cite{MR982872} to obtain a $(4k+1,0)$-spanner of the same asymptotic density.

There are also spanner results that significantly sparsify unweighted graphs at the cost of a purely additive error, as a (1,2)-spanner with $\cO(n^{3/2})$ edges~\cite{MR1681058}, or a (1,6)-spanner with $\cO(n^{4/3})$ edges~\cite{MR2298319}.
We do not know if analogous results are possible in the contraction model.
The incompressibility result in~\cite{MR3536579} mentioned above implies the same lower bound for spanners as for contractions and every other distance oracle with additive error: For every $\varepsilon > 0$ any spanner of size $\cO(n^{4/3 - \varepsilon})$ does not admit a constant additive error.
Finally, for spanners there are results that combine multiplicative and additive error, such as the $(k,k-1)$-spanner of~\cite{MR2298319}.

Gupta~\cite{MR1958411} considered the problem of approximating a tree metric on a subset of the vertices by another tree, and gave a linear time algorithm computing an $8$-approximation.
As Chan et al.~\cite{MR2304999} observed later, on complete binary trees a solution of minimum distortion is always achieved by a minor (with possibly different edge lengths) of the input tree, so this seems to be the first investigation of contractions that approximate graph distances.
Krauthgamer et al.~\cite{MR3158782} considered an extension to general graphs, studying the size of minors preserving all distances between a given terminal set of fixed size.
Cheung et al.~\cite{DBLP:conf/icalp/CheungGH16} introduced a multiplicative distortion to this model.
As here no two terminals may be merged, these approaches cannot compress a graph at all if every vertex is a terminal.

The \emph{pairwise preservers} due to Coppersmith et al.~\cite{MR2257273} combine spanners with the aim of preserving only terminal distances.
Given a graph $G$ and a set of $k$ terminal pairs, a pairwise preserver is a spanning subgraph inducing exactly the same terminal distances as~$G$. Coppersmith et al.~\cite{MR2257273} proved that for every undirected weighted graph there exists a pairwise preserver of size $\cO(n + n^{1/2}k)$. Furthermore, they showed that every directed weighted graph has a pairwise preserver of size $\cO(nk^{1/2})$.
For the special case of undirected unweighted graphs, Bodwin et al.~\cite{MR3478437} showed the existence of a pairwise preserver with $\cO(n^{2/3}k^{2/3} + nk^{1/3})$ edges.
Recently, Bodwin~\cite{MR3627768} proved that any directed weighted graph has a pairwise preserver of size $\cO(n + n^{2/3}k)$.

\subsection{Further related work}

The preservation of graph properties other than distances has been studied as well.
Biedl et al.~\cite{MR1844744} considered contractions in capacitated networks with the goal of maintaining the maximum flow in the network.
Here an edge $e$ is called \emph{useless}, if for every capacity function there is a maximum flow not using $e$.
Biedl et al.\ showed that finding all useless edges is NP-complete, but solvable in $\cO(n^2)$ time on certain planar graphs.
For undirected networks, Misio{\l}ek et al.~\cite{MR2190897} gave an algorithm finding all useless edges in $\cO(n + m)$ time.
Toivonen et al.~\cite{DBLP:conf/icdm/ZhouMT10} considered a more general model aiming to maintain the quality of paths with respect to any given function, e.g., distance or capacity.
They investigated strategies of removing edges, without decreasing the quality of the best path between any pair of vertices.

Graph simplification problems have also been studied in several other contexts, and we conclude this section by mentioning two such examples:
H{\"u}bler et al.~\cite{DBLP:conf/icdm/HublerKBG08} studied a problem related to graph mining, examining how to choose an induced subgraph with a given number of vertices and with similar topological properties as the input graph.
Numerous papers investigate, directly or as a tool, sparsifiers that preserve the effective resistance between certain or all pairs of vertices, see e.g.~\cite{MR3017573,MR3441994,MR3631020,MR3678224,DBLP:conf/focs/ChuGPSSW18}.

\subsection{Outline of this paper}

In Section~\ref{sec:prelim} we introduce important definitions and notations that will be used throughout this paper.
In Section~\ref{sec:greedy} we present our three greedy algorithms for solving \Contraction{} with unit lengths on paths, cycles and trees (the latter result requires $\alpha=1$).
In Section~\ref{sec:trees} we discuss efficient dynamic programming algorithms for \Contraction{} and \WContraction{} on trees.
Sections~\ref{sec:hard-add} and \ref{sec:hard-mult} are devoted to our hardness results, focussing on the cases of purely additive and multiplicative error, respectively.
In Section~\ref{sec:asymp} we present our asymptotic results on contractions.

\section{Preliminaries}
\label{sec:prelim}

Throughout this paper we consider simple undirected graphs~$G$ (without parallel edges or loops).
We let~$V(G)$ and~$E(G)$ denote the vertex and edge set of~$G$, respectively, and we define $n(G):=|V(G)|$ and $m(G):=|E(G)|$.
If the context is clear, we simply write $V$, $E$, $n$ and~$m$.
We also use the notation $[n]:=\{1,2,\ldots,n\}$.
We assume that~$G$ is connected, otherwise the contraction problem can be solved independently for each connected component.
Edge lengths are given by a function $\ell\colon E\to\RR_{>0}$.
The \emph{distance} $\dist_\ell(u,v)$ between two vertices~$u$ and~$v$ is the length of a shortest path between~$u$ and~$v$ in~$G$ with respect to~$\ell$.

Given a subset of edges $C\subseteq E$, we denote the resulting simple graph obtained from $G$ by contracting the edges in $C$, deleting resulting loops and keeping only the minimum length edge between any two vertices by $G/C$.
We denote the number of deleted loops and multi-edges by $\Delta(C)$ (thus $m(G/C)=m(G)-|C|-\Delta(C)$).
Instead of contracting a set $C\subseteq E$ of edges in $G$, setting their edge lengths to zero has the same effect on the distances in the resulting graph.
This is somewhat cleaner conceptually, so we will often adopt this viewpoint.
Specifically, we let $\ell_C$ be the new length function that assigns 0 to every edge in $C$, and that is equal to the original edge lengths $\ell$ on the edges $E\setminus C$.

A \emph{tolerance function} is a non-decreasing function $\varphi\colon\RR\to \RR$.
Roughly speaking, this function describes by how much the distance between two vertices may drop when contracting edges (i.e., setting edge lengths to zero).
Formally, given a graph $G$ with edge lengths $\ell$ and a tolerance function $\varphi$, we say that a subset of edges $C\subseteq E$ is a \emph{$\varphi$-distance preserving contraction} or \emph{$\varphi$-contraction} for short, if
\begin{equation}
\label{eq:contr-cond}
\dist_{\ell_C}(u,v)\geq \varphi(\dist_\ell(u,v))
\end{equation}
holds for any two vertices~$u$ and~$v$ in~$G$.
Similarly, we say that $C$ is a \emph{weak $\varphi$-distance preserving contraction} or \emph{weak $\varphi$-contraction} for short, if any two vertices $u$ and $v$ satisfy relation \eqref{eq:contr-cond} or the relation $\dist_{\ell_C}(u,v)=0$, and if the graph $(V,C)$ is disconnected (equivalently, if $G/C$ is not a single vertex).
The last condition prevents solutions $C\subseteq E$ for which the graph is contracted to a single vertex.
If $\varphi(x)=x/\alpha-\beta$, then we simply write (weak) $(\alpha,\beta)$-contraction instead of (weak) $\varphi$-contraction.

An \emph{instance} of the problem \Contraction{} or \WContraction{} is a triple $(G,\ell,\varphi)$, where~$G$ is the underlying graph, $\ell$ the length function and~$\varphi$ the tolerance function, and the objective is to find a (weak) $\varphi$-distance preserving contraction $C\subseteq E$, such that
\begin{equation}
\label{eq:objective}
\Phi(C) := |C|+\Delta(C)=m(G)-m(G/C)
\end{equation}
is maximized.
This quantity equals the number of edges we save when going from $G$ to $G/C$.
Note that on trees we have $\Phi(C)=|C|$ for any (weak) contraction $C$, whereas on general graphs we have $\Phi(C)\geq |C|$.

\begin{problem}{\wContraction{}}
Input: & A graph $G = (V,E)$ with edge lengths $\ell\colon E \to \RR_{>0}$ and a non-decreasing function $\varphi\colon \RR \to \RR$. \\
Output: & A (weak) $\varphi$-distance preserving contraction $C \subseteq E$ maximizing $\Phi(C)$. \\
\end{problem}

In this context we sometimes refer to a set of edges that forms a (weak) contraction as a \emph{feasible} solution, and to a (weak) contraction of maximum value $\Phi(C)$ as an \emph{optimal} solution.

We begin by proving that our contraction model behaves nicely when contracting edges in phases, i.e., the total error is simply the error accumulated over the contraction phases (but not more).
To state this result we denote the composition of tolerance functions~$\varphi$ and~$\psi$ as $(\psi \circ \varphi)(x):=\psi(\varphi(x))$.

\begin{thm}
\label{thm:contract-comp}
Let~$C$ be a (weak) $\varphi$-contraction for~$G$, and let~$C'$ be a (weak) $\psi$-contraction for~$G/C$.
Then $C\cup C'$ is a (weak) $(\psi \circ \varphi)$-contraction for~$G$.
\end{thm}

\begin{proof}
We only prove the statement for contractions~$\varphi$ and~$\psi$.
The proof for weak contractions works analogously.
Let~$\ell$ denote the edge lengths of~$G$ and consider a pair of vertices $u,v\in V(G)$.
Then we have $\dist_{\ell_{C\cup C'}}(u,v) \geq \psi(\dist_{\ell_C}(u,v))$ by the definition of~$C'$ and $\dist_{\ell_C}(u,v)\geq \varphi(\dist_\ell(u,v))$ by the definition of~$C$.
Combining these inequalities and using that~$\psi$ is non-decreasing we obtain $\dist_{\ell_{C\cup C'}}(u,v)\geq \psi(\varphi(\dist_\ell(u,v)))$, as desired.
\end{proof}

Note that Theorem~\ref{thm:contract-comp} only concerns the \emph{feasibility} of repeated contractions, but not about their \emph{optimality} when searching for contractions of maximum cardinality.
With respect to solution quality, contracting in phases may be arbitrarily bad:
Consider a star with~$k$ unit length edges and additive tolerance functions $\varphi(x)=\psi(x)=x-1$.
An optimum $(\psi \circ \varphi)$-contraction contains all~$k$ edges, whereas finding an optimal $\varphi$-contraction~$C$ and then an optimal $\psi$-contraction of~$G/C$ allows contracting only one edge in each phase, leading to a $(\psi \circ \varphi)$-contraction of value 2.

\section{Greedy algorithms}
\label{sec:greedy}

In this section we consider three special cases of the problem \Contraction{} with affine tolerance function $\varphi(x)=x/\alpha-\beta$.
We obtain simple greedy algorithms computing maximum size $\varphi$-contractions in $\cO(n)$ time on paths and cycles with unit lengths, and on trees with unit lengths and $\alpha=1$.

\subsection{Paths with unit length edges}
\label{sec:paths}

In this section we consider the special case of contracting a path~$P_n$ with~$n-1$ unit length edges~$\ell=1$ and the tolerance function $\varphi(x)=x/\alpha-\beta$.
In this case optimal solutions have a very special structure, which leads to a straightforward greedy algorithm running in linear time.
Recall that as a path is a tree, our objective functions satisfies $\Phi(C)=|C|$ for any contraction $C$.

Observe that a solution $C\subseteq E(P_n)$ for the instance $(P_n,\ell,\varphi)$ of the problem \Contraction{} is feasible, if and only if every subpath $P'\subseteq P_n$ satisfies the condition
\begin{equation}
\label{eq:path-cond}
|E(P')\cap C| \leq (1-1/\alpha)|E(P')|+\beta .
\end{equation}
This observation leads to the following natural greedy algorithm $\Greedy(P_n,\alpha,\beta)$:
The algorithm considers the edges $e_1,e_2,\ldots,e_{n-1}$ of~$P_n$ as they are encountered when starting from one of the two end vertices of~$P_n$.
It iteratively constructs a solution~$C$ for the subpath on the first~$i$ edges $e_1,e_2,\ldots,e_i$ for $i=1,2,\ldots,n-1$, by initializing $C:=\emptyset$, and by adding the edge~$e_i$ to~$C$ if and only if the condition $|C|+1\leq (1-1/\alpha)i+\beta$ is satisfied (so after adding~$e_i$ to~$C$, \eqref{eq:path-cond} is still satisfied).

\begin{thm}
\label{thm:path}
Let~$P_n$ be a path with unit length edges~$\ell=1$ and consider the tolerance function $\varphi(x)=x/\alpha-\beta$, $\alpha,\beta\geq 1$.
The set of edges computed by the algorithm $\Greedy(P_n,\alpha,\beta)$ is an optimal solution for the instance $(P_n,\ell,\varphi)$ of the problem \Contraction{}, and it is computed in time $\cO(n)$.
\end{thm}

\begin{proof}
Let $C\subseteq E(P_n)$ be the set of edges computed by the algorithm $\Greedy(P_n,\alpha,\beta)$.
Clearly, we have $|C|=\lfloor(1-1/\alpha)|E(P_n)|+\beta\rfloor$, and this is optimal according to \eqref{eq:path-cond}.
However, it remains to show that~$C$ is feasible.
For $1\leq i\leq j\leq n-1=|E(P_n)|$ we let~$P_{i,j}$ denote the subpath of~$P_n$ formed by the edges $e_i,e_{i+1},\ldots,e_j$.
By the definition of our algorithm we know that $|E(P_{1,i})\cap C|=\lfloor (1-1/\alpha)i+\beta\rfloor$, from which we obtain that
\begin{align*}
|E(P_{i,j})\cap C| &= |E(P_{1,j})\cap C| - |E(P_{1,i-1})\cap C| \\
&= \lfloor (1-1/\alpha)j+\beta \rfloor - \lfloor (1-1/\alpha)(i-1)+\beta \rfloor \\
&\leq (1-1/\alpha)j+\beta - ((1-1/\alpha)(i-1)+\beta - 1) \\
&= (1-1/\alpha)(j-i+1)+1 \\
&\leq (1-1/\alpha)|E(P_{i,j})|+\beta ,
\end{align*}
where we used the assumption $\beta\geq 1$ in the last step.
Using \eqref{eq:path-cond} it thus follows that~$C$ is feasible.
\end{proof}

\subsection{Cycles with unit length edges}
\label{sec:cycles}

In this section we consider the special case of contracting a cycle~$C_n$ with~$n$ vertices and unit length edges~$\ell=1$ and the tolerance function $\varphi(x)=x/\alpha-\beta$, $\alpha\geq 1$, $\beta\geq 0$.
For this case we present a greedy algorithm running in linear time.
The main purpose of this result is to clearly separate the polynomially solvable cases of \Contraction{} from the NP-hard cases, and the case of a cycle with unit length edges precisely forms this boundary on the polynomially solvable side.
Recall in this context that we can solve \Contraction{} in polynomial time on any tree (this will be proved in Section~\ref{sec:dp} below), and that \Contraction{} is NP-hard already on a cycle for $\alpha=1$ (with arbitrary edge lengths; we will show this in Section~\ref{sec:cycles-hard} below).

We first argue that on a cycle it is equivalent to maximize the number of contracted edges $|C|$ or to maximize our objective function $\Phi(C)$ defined in \eqref{eq:objective}.
This is because the set of pairs $(|C|,\Phi(C))$ for all feasible contractions $C$ in a cycle $G=C_n$ is given by $\{(1,1),(2,2),\ldots,(n-3,n-3),(n-2,n-1),(n-1,n),(n,n)\}$, so it forms a monotone function, implying that maximizing either one of the two quantities is equivalent.
Based on this argument, for the rest of this section we consider maximizing the number $|C|$ of contracted edges.

Observe that a solution $C\subseteq E(C_n)$ ($C_n$ is the cycle we want to contract, and~$C$ is the set of edges to be contracted) for the instance $(C_n,\ell,\varphi)$ of the problem \Contraction{} is feasible, if and only if every subpath $P\subseteq C_n$ of length $d:=|E(P)|\in\{1,2,\ldots,n-1\}$ satisfies the condition
\begin{equation}
\label{eq:cycle-cond}
|E(P)\cap C| \leq \lfloor d - \min\{d,n-d\}/\alpha + \beta \rfloor .
\end{equation}
Rounding down on the right-hand side of \eqref{eq:cycle-cond} is justified because $|E(P)\cap C|$ is always an integer.

Defining
\begin{subequations}
\label{eq:uniform-solution}
\begin{align}
\lambda' &:= \min_{d\in\{1,2,\ldots,n-1\}} \frac{\lfloor d - \min\{d,n-d\}/\alpha + \beta \rfloor}{d} , \label{eq:uniform-sol1} \\
\lambda  &:= \min\{1,\lambda'\} , \label{eq:uniform-sol2} 
\end{align}
\end{subequations}
we obtain from \eqref{eq:cycle-cond} that $\lambda\in[0,1]$ is the maximal amount by which we can contract each edge in a uniform \emph{fractional} solution.
Inspired by the rounding technique from \cite{MR589671}, we turn this fractional solution into an integer optimal solution, yielding the following greedy algorithm $\Greedy(C_n,\alpha,\beta)$:
The algorithm considers the edges $e_1,e_2,\ldots,e_n$ of~$C_n$ as they are encountered when walking around the cycle.
It iteratively constructs a solution~$C$ by initializing $C:=\emptyset$ and by adding the edge~$e_i$ to~$C$ if and only if $\lfloor \lambda i\rfloor - \lfloor \lambda (i-1)\rfloor = 1$ for all $i=1,2,\ldots,n$ (since $\lambda\in[0,1]$, this difference is always either 0 or 1).
Note that we contract all edges of~$C_n$ if and only if~$\lambda=1$.

\begin{thm}
\label{thm:cycle}
Let~$C_n$ be a cycle with unit length edges~$\ell=1$ and consider the tolerance function $\varphi(x) = x/\alpha - \beta$, $\alpha \geq 1$, $\beta \geq 0$.
The set of edges computed by the algorithm $\Greedy(C_n,\alpha,\beta)$ is an optimal solution for the instance $(C_n,\ell,\varphi)$ of the problem \Contraction{}, and it is computed in time $\cO(n)$.
\end{thm}

The next lemma shows that the contraction computed by our algorithm has the maximum size.

\begin{lem}
\label{lem:lambda-prop}
For any feasible solution $C \subseteq E(C_n)$ we have $|C| \leq \lfloor \lambda n \rfloor$ with~$\lambda$ defined in \eqref{eq:uniform-solution}.
\end{lem}

\begin{proof}
If $\lambda=1$ this inequality is trivial.
So let us assume that $\lambda=\lambda'<1$ and that the minimum in \eqref{eq:uniform-sol1} is attained for some $d\in\{1,2,\ldots,n-1\}$.
Starting at some vertex~$u$ of the cycle, we walk along the cycle and cover it with~$n$ consecutive paths $P_1,P_2,\ldots,P_n$ of length~$d$ each ($P_{i+1}$ starts where~$P_i$ ends).
The sum of the lengths of the paths is~$n d$, so this process ends at the starting vertex~$u$, and each edge of the cycle and each edge of~$C$ is covered exactly~$d$ times.
We therefore obtain
\begin{equation*}
|C| = \frac{1}{d} \sum_{i=1}^n |E(P_i)\cap C|
\leBy{eq:cycle-cond} \sum_{i=1}^n \frac{\lfloor d - \min\{d,n-d\}/\alpha + \beta\rfloor}{d} \eqBy{eq:uniform-solution} \lambda n .
\end{equation*}
As~$|C|$ must be integral this inequality yields the desired bound $|C|\leq \lfloor\lambda n\rfloor$.
\end{proof}

With Lemma~\ref{lem:lambda-prop} in hand, we are now ready to prove Theorem~\ref{thm:cycle}.

\begin{proof}[Proof of Theorem~\ref{thm:cycle}]
In this proof we will use that for any two real numbers~$x$ and~$y$ we have
\begin{subequations}
\begin{align}
\lfloor x\rfloor + \lfloor y\rfloor &\leq \lfloor x+y\rfloor , \label{eq:floor-sum} \\
\lfloor x \rfloor - \lfloor y \rfloor &\leq \lceil x-y \rceil . \label{eq:floor-diff}
\end{align}
\end{subequations}

Let $C\subseteq E(C_n)$ be the set of edges computed by the algorithm $\Greedy(C_n,\alpha,\beta)$.
Clearly, we have $|C|=\sum_{i=1}^n (\lfloor \lambda i\rfloor - \lfloor \lambda (i-1)\rfloor) = \lfloor \lambda n \rfloor$, which is optimal by Lemma~\ref{lem:lambda-prop}.
However, it remains to show that~$C$ is feasible.
We consider a path~$P$ of length $d:=|E(P)|\in\{1,2,\ldots,n-1\}$ on the edges $e_k,e_{k+1},\ldots,e_{k+d-1}$ (indices are considered cyclically modulo~$n$, so $e_{n+i}=e_i$).
We distinguish two cases:
If $k+d-1\leq n$, we have
\begin{equation}
\label{eq:PcapC1}
|E(P)\cap C| = \sum_{i=k}^{k+d-1} (\lfloor \lambda i \rfloor - \lfloor \lambda(i-1)\rfloor)
= \lfloor \lambda (k+d-1) \rfloor - \lfloor \lambda(k-1) \rfloor \leBy{eq:floor-diff} \lceil \lambda d\rceil .
\end{equation}
If $k+d-1>n$, we obtain
\begin{align}
|E(P)\cap C| &= \sum_{i=k}^n (\lfloor \lambda i\rfloor - \lfloor \lambda(i-1) \rfloor)
+ \sum_{i=1}^{d-n+k-1} (\lfloor \lambda i\rfloor - \lfloor \lambda(i-1) \rfloor) \notag \\
&= \lfloor \lambda n \rfloor - \lfloor \lambda(k-1)\rfloor + \lfloor \lambda (d-n+k-1) \rfloor \notag \\
&\leBy{eq:floor-sum} \lfloor \lambda (d+k-1) \rfloor - \lfloor \lambda(k-1)\rfloor
\leBy{eq:floor-diff} \lceil \lambda d\rceil . \label{eq:PcapC2}
\end{align}
Applying \eqref{eq:uniform-solution} and using that $\lceil \lfloor x\rfloor\rceil=\lfloor x\rfloor$ shows that the right-hand sides of \eqref{eq:PcapC1} and \eqref{eq:PcapC2} can both be bounded from above by $\lfloor d - \min\{d,n-d\}/\alpha + \beta \rfloor$, proving that~$C$ is indeed feasible by \eqref{eq:cycle-cond}.
\end{proof}

\subsection{Trees with unit length edges and additive error}
\label{sec:trees-unit-length}

In this section we consider the special case of contracting a tree~$T$ with unit length edges $\ell=1$ and the tolerance function $\varphi(x)=x-\beta$ (purely additive error; we can assume w.l.o.g.\ that~$\beta$ is an integer).
Note that in this setting the objective function defined in \eqref{eq:objective} satisfies $\Phi(C)=|C|$ for any contraction $C$.
It turns out that in this case, optimal solutions have a very special structure that can be exploited to compute them in linear time.
Specifically, an optimal solution is obtained by taking all edges of~$T$ which have the property that only short paths start from one of its end vertices.
Formally, for the tree~$T$ and $d\in\NN_{\geq 0}$, we let $L(T,d)$ denote the set of all edges~$e$ of~$T$ which have one end vertex~$v$ such that all paths that start at~$v$ and do not contain~$e$ have length at most~$d-1$ (together with~$e$ these paths have length at most~$d$).
E.g., we have $L(T,0)=\emptyset$, and the set $L(T,1)$ are all the edges incident to a leaf (see Figure~\ref{fig:L}).

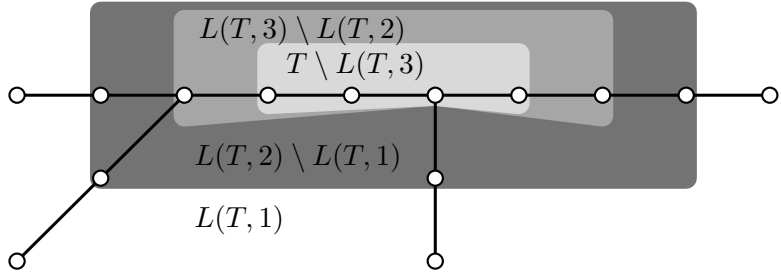
\begin{figure}[ht]
\centering
\input{Figures/leafs}
\caption{Illustration of the sets of edges $L(T,d)$, $d=1,2,3$.
The set $L(T,3)$ is an optimal solution of the problem \Contraction{} with $\beta=6$.}
\label{fig:L}
\end{figure}

Clearly, the set $L(T,d)$ can be computed in linear time by repeatedly removing all leaves of~$T$ in~$d$ rounds.
This is a variant of the well-known linear time algorithm to compute the so-called center of a tree (see \cite[Section~15.11]{DBLP:books/daglib/0022194}).

\begin{thm}
\label{thm:unit-tree}
Let~$T$ be a tree with unit length edges~$\ell=1$ and consider the tolerance function $\varphi(x)=x-\beta$, $\beta\in\NN_{\geq 0}$.
If~$\beta$ is even, the set of edges $L(T,d)$ with $d:=\lfloor\beta/2\rfloor$ is an optimal solution for the instance $(T,\ell,\varphi)$ of the problem \Contraction{}.
If~$\beta$ is odd, $L(T,d)\cup \{e\}$, $e\in E\setminus L(T,d)$, is an optimal solution.
These solutions can be computed in time $\cO(n)$.
\end{thm}

\begin{proof}
We define $C:=L(T,d)$ if~$\beta$ is even and $C:=L(T,d)\cup\{e\}$, for some $e\in E\setminus L(T,d)$, if~$\beta$ is odd.
We first argue that~$C$ is a feasible solution.
To see this note that for the given tolerance function we only need to verify that the path~$P$ between any two \emph{leaves} $u,v$ of~$T$ contains at most~$\beta$ edges.
Consider all the edges of~$P$ for which both end vertices have distance at least~$d$ from both~$u$ and~$v$.
None of those edges is in $L(T,d)$ by its definition.
It follows that $|P\cap L(T,d)|\leq 2d=2\lfloor\beta/2\rfloor$ and therefore $|P\cap C|\leq \beta$.

To prove that~$C$ is a solution of maximum size we argue by induction over~$\beta$.
The claim is trivially true for $\beta=0$ and $\beta=1$ (in these cases $|C|=0$ and $|C|=1$, respectively).
So let~$D$ be an arbitrary feasible solution of the instance $(T,\ell,\beta)$ of the problem \Contraction{} for some $\beta\geq 2$.
We need to show that $|C|\geq |D|$.
To this end we let $V^*\subseteq V(T)$ denote the set of leaves of~$T$ and we define $E^*:=L(T,1)$.
Moreover, we define $T^*:=T\setminus V^*$ and $C^*:=C\setminus E^*$.
By induction, $C^*=L(T^*,d-1)$ is an optimal solution for the instance $(T^*,\ell,\beta-2)$.

We first consider the case that $E^*\setminus D=\emptyset$ or $D\setminus E^*=\emptyset$ (this is equivalent to $E^*\subseteq D$ or $D\subseteq E^*$).
In this case we define $D^*:=D\setminus E^*$, and observe that~$D^*$ is a feasible solution for the instance $(T^*,\ell,\beta-2)$.
It follows that $|C^*|\geq |D^*|$, implying that $|C|=|C^*|+|E^*|\geq |D^*|+|E^*|\geq |D|$, as claimed.

\begin{figure}[ht]
\centering
\input{Figures/exchange}
\caption{Illustration of the proof of Theorem~\ref{thm:unit-tree}.}
\label{fig:exchange}
\end{figure}
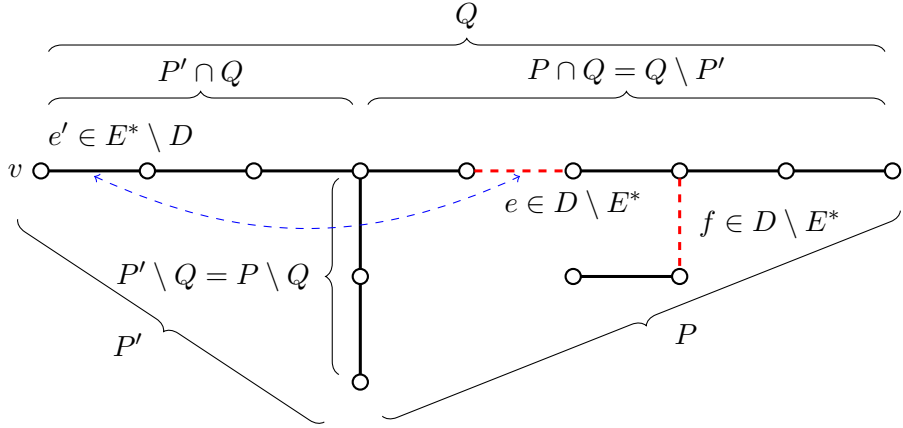

It remains to consider the case that both sets $E^*\setminus D$ and $D\setminus E^*$ are nonempty, so there is an edge $e'\in E^*\setminus D$ and an edge $f\in D\setminus E^*$.
We denote the leaf incident to~$e'$ by~$v$.
We will now remove an edge $e\in D\setminus E^*$ from~$D$ and add~$e'$ instead to obtain another feasible solution~$D'$ satisfying $|D|=|D'|$.
Repeating this exchange argument and applying the reasoning from the first case then proves the lemma.
The edge $e\in D\setminus E^*$ to be removed from~$D$ is obtained by considering the path that connects~$v$ and~$f$ in~$T$ and that contains~$f$, and by choosing the first edge from~$D$ (or equivalently, from $D\setminus E^*$) that is encountered when following this path from~$v$ to~$f$.
It may happen that $e=f$ is the first such edge we encounter.
To complete the proof of the lemma it remains to show that $D'=D\setminus\{e\}\cup\{e'\}$ is feasible.
To prove this we only need to check paths which start in~$v$ and contain~$e'$ but not~$e$.
Let~$P'$ be such a path, let~$Q$ be any path that also starts in~$v$ but does contain~$e$, and consider the path $P:=(P'\setminus Q)\cup (Q\setminus P')$ (see Figure~\ref{fig:exchange}).
Here and in the following we slightly abuse notation and interpret these set unions/differences/intersections in terms of the edge sets of the graphs.
As~$D$ is feasible and as $P\cap Q$ contains~$e$, the number of edges in~$D$ or~$D'$ on $P'\setminus Q=P\setminus Q$ is at most $\beta-1$.
By the choice of~$e$, the number of edges of~$D'$ on $P'\cap Q$ is 1 (the only edge of~$D'$ on this path is~$e'$).
As $P'=(P'\setminus Q)\cup (P'\cap Q)$, we obtain that the number of edges from~$D'$ on~$P'$ is at most $\beta-1+1=\beta$, as desired.
This completes the proof.
\end{proof}

\section{Dynamic programs for general trees}
\label{sec:trees}

In this section we describe dynamic programming algorithms for the problems \Contraction{} and \WContraction{} on trees with general edge lengths and affine tolerance functions.
Recall that on trees our objective function satisfies $\Phi(C)=|C|$ for any contraction~$C$.

\subsection{\Contraction{} on trees}
\label{sec:dp}

In this section we describe a dynamic programming algorithm for the problem of computing an optimal contraction of a tree~$T$ with arbitrary edge lengths $\ell\colon E\to \RR_{>0}$ and an affine tolerance function $\varphi(x)=x/\alpha-\beta$, $\alpha\geq 1$, $\beta\geq 0$, generalizing the solution for the special case presented at the beginning of the previous section.
The goal is to prove the following result.

\begin{thm}
\label{thm:dp}
Let~$T$ be a tree with edge lengths $\ell\colon E\to\RR_{>0}$ and consider the tolerance function $\varphi(x)=x/\alpha-\beta$, $\alpha\geq 1$, $\beta\geq 0$.
An optimal solution for the instance $(T,\ell,\varphi)$ of the problem \Contraction{} can be computed by dynamic programming in time $\cO(n^3)$.
\end{thm}

Observe that a solution $C\subseteq E$ is feasible if and only if for any two vertices~$u$ and~$v$ of~$T$ we have $\load_{C,\alpha}(u,v)\leq \beta$, where the \emph{load between~$u$ and~$v$} is defined as
\begin{subequations}
\label{eq:load}
\begin{equation}
\label{eq:load-path}
\load_{C,\alpha}(u,v):=\dist_\ell(u,v)/\alpha-\dist_{\ell_C}(u,v)
\end{equation}
(recall \eqref{eq:contr-cond}).
For any vertex~$v$ of~$T$ we also define the \emph{load of~$T$ at~$v$} as
\begin{equation}
\label{eq:load-tree}
\load_{C,\alpha}(T,v):=\max\{\load_{C,\alpha}(u,v): u \in V(T) \} .
\end{equation}
\end{subequations}
Note that $\load_{C,\alpha}(T,v)\geq 0$, as we have $\load_{C,\alpha}(v,v)=0$.
The next lemma states a criterion when feasible solutions of subtrees can be combined to a feasible solution of the entire tree.
The definitions \eqref{eq:load-path}, \eqref{eq:load-tree} and the lemma are illustrated in Figure~\ref{fig:load}.

\begin{lem}
\label{lem:tree-partition}
Consider a partition of~$T$ into two subtrees~$T_1$ and~$T_2$ that only have a vertex~$v\in V$ in common.
Then $C\subseteq E$ is a feasible solution for the instance $(T,\ell,\varphi)$ of the problem \Contraction{} if and only if the following two conditions hold: $C\cap E(T_1)$ and $C\cap E(T_2)$ are feasible solutions for the instances $(T_1,\ell,\varphi)$ and $(T_2,\ell,\varphi)$ respectively; and we have $\load_{C,\alpha}(T_1,v)+\load_{C,\alpha}(T_2,v)\leq \beta$.
\end{lem}

\begin{proof}
Observe that the path between two vertices $u\in T_1$ and $w\in T_2$ contains the vertex~$v$, so we obtain $\load_{C,\alpha}(u,w)=\load_{C,\alpha}(u,v)+\load_{C,\alpha}(v,w)$ from \eqref{eq:load-path}.
Using \eqref{eq:load-tree} it follows that the condition $\load_{C,\alpha}(u,w)\leq \beta$ holding for all such pairs of vertices $u,w$ is equivalent to $\load_{C,\alpha}(T_1,v)+\load_{C,\alpha}(T_2,v)\leq \beta$.
\end{proof}

We will use this lemma to formulate our dynamic programming algorithm.
The idea is to compute optimal solution for subtrees and combining them to an optimal solution for the entire tree.

To describe the algorithm we introduce a few definitions.
An \emph{ordered rooted tree} is a rooted tree with a specified left-to-right ordering for the children of each vertex.
Given the tree~$T$, we can pick an arbitrary vertex as the root, and for each descendant of the root an arbitrary left-to-right ordering of its children, yielding an ordered rooted tree (different roots and orderings yield different ordered rooted trees, but any one of them is good for our purposes).
We slightly abuse notation in the following and use~$T$ to denote this ordered rooted tree.
All trees considered in the rest of this section are ordered and rooted.
For any vertex~$v$ of~$T$, we let~$T_v$ denote the subtree of~$T$ rooted at~$v$, and we use~$c(v)$ to denote the number of children of~$v$.
If $u_1,u_2,\ldots,u_{c(v)}$ are the children of~$v$ (in the specified ordering), we write $T_{v,i}$, $i\in\{1,\ldots, c(v)\}$, for the subtree of~$T$ that contains~$v$,~$u_i$ and all the descendants of~$u_i$.
We also define $T_{v,0}:=\{v\}$.
Furthermore, we define $T_{v,i}^+:=\bigcup_{0\leq j\leq i} T_{v,i}$, so we have $T_v=T_{v,c(v)}^+$.
These definitions are illustrated in Figure~\ref{fig:load}.

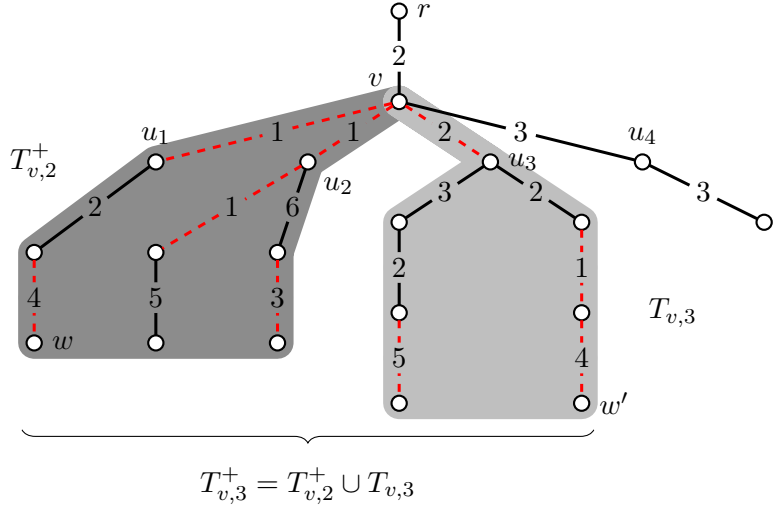
\begin{figure}[ht]
\centering
\input{Figures/load}
\caption{Illustration of the definitions of Section~\ref{sec:dp} for $(\alpha,\beta)=(2,4)$.
For the set~$C$ of dashed edges we obtain $\load_{C,\alpha}(T_{v,2}^+,v)=\load_{C,\alpha}(v,w) = 7/2-2 = 3/2$ and $\load_{C,\alpha}(T_{v,3},v) = \load_{C,2}(v,w') = 9/2-2 = 5/2$.
We therefore have $\load_{C,\alpha}(T_{v,2}^+,v)+\load_{C,\alpha}(T_{v,3},v)=4\leq \beta$, so by Lemma~\ref{lem:tree-partition} the set~$C$ is a feasible $(\alpha,\beta)$-contraction in $T_{v,3}^+$.}
\label{fig:load}
\end{figure}

Using these definitions it follows straightforwardly from \eqref{eq:load-path} and \eqref{eq:load-tree} that for any set of edges $C\subseteq E(T_{u_i})$ we have
\begin{subequations}
\label{eq:load-add-edge}
\begin{align}
\load_{C\cup\{\{v,u_i\}\},\alpha}(T_{v,i},v) &= \load_{C,\alpha}(T_{u_i},u_i)+\ell(v,u_i)/\alpha , \label{eq:load-add-edge-yes} \\
\load_{C,\alpha}(T_{v,i},v) &= \max\{ \load_{C,\alpha}(T_{u_i},u_i)-(1-1/\alpha)\ell(v,u_i), 0 \} . \label{eq:load-add-edge-no}
\end{align}
\end{subequations}
Note that the load increases if the edge $\{v,u_i\}$ is added to $C$ (see \eqref{eq:load-add-edge-yes}), and it decreases otherwise (see \eqref{eq:load-add-edge-no}).
Moreover, for any set of edges $C\subseteq T_{v,i}^+$ and any $i=1,2,\ldots,c(v)$ we obtain from those definitions that
\begin{equation}
\label{eq:load-join-trees}
\load_{C,\alpha}(T_{v,i}^+,v) = \max \{\load_{C,\alpha}(T_{v,i-1}^+,v), \load_{C,\alpha}(T_{v,i},v) \} .
\end{equation}
These rules allow us to compute the load of all subtrees of~$T$ in a bottom-up fashion.
Our dynamic program maintains the minimum load of all subtrees of~$T$ in three-dimensional matrices~$L$ and~$L^+$.
We begin defining these matrices in an abstract way, and then establish several recursive relations which directly translate into a dynamic program.
Specifically, for $v\in V$, $i\in\{0,1,\ldots,c(v)\}$ and $s\in\{0,1,\ldots,m\}$ (recall that $m=|E|$) we define
\begin{equation}
\label{eq:def-L}
L(v,i,s) := \min \{\load_{C,\alpha}(T_{v,i},v): \text{$C$ is a feasible solution of $(T_{v,i},\ell,\varphi)$ of size $|C|=s$}\} .
\end{equation}
If there is no feasible solution of the required size, we have $L(v,i,s)=\infty$.
The entries of $L^+(v,i,s)$ are defined analogously to \eqref{eq:def-L} by considering the load of~$T_{v,i}^+$ instead of~$T_{v,i}$.
In words, the entries~$L(v,i,s)$ and~$L^+(v,i,s)$ describe feasible solutions~$C$ of size~$s$ of the instances $(T_{v,i},\ell,\varphi)$ or $(T_{v,i}^+,\ell,\varphi)$, respectively, of the problem \Contraction{} for which the load at the vertex~$v$ is as small as possible (the matrices contain the minimum achievable load, not the corresponding set of edges).

\begin{lem}
\label{lem:L-Lp-rec}
Let~$v$ be a vertex of~$T$ and let $u_1,u_2,\ldots,u_{c(v)}$ be the children of~$v$.
Then the matrices~$L$ and~$L^+$ defined in and directly after \eqref{eq:def-L} satisfy the relations
\begin{subequations}
\begin{align}
L(v,i,0) &= L^+(v,i,0) = 0 \quad \text{for all } i\in\{0,1,\ldots,c(v)\} , \label{eq:dp-init-L1} \\
L(v,0,s) &= L^+(v,0,s) = \infty \quad \text{for all } s\in\{1,2,\ldots,m\} , \label{eq:dp-init-L2} \\
L(v,i,s) &= \begin{cases}
     \mu    & \text{if } \mu \leq \beta , \\
     \infty & \text{otherwise}, 
  \end{cases} \label{eq:dp-update-L}
\end{align}
where $\mu:=\min \big\{ L^+(u_i,c(u_i),s-1)+\ell(v,u_i)/\alpha, \, \max\{L^+(u_i,c(u_i),s)-(1-1/\alpha)\ell(v,u_i), 0\}\big\}$ for all $i\in\{1,2,\ldots,c(v)\}$ and $s\in\{1,2,\ldots,m\}$.

Moreover, we have
\begin{align}
L^+(v,i,s) &= \min \big\{\max \{L^+(v,i-1,t),L(v,i,s-t)\} : t\in\{0,1,\ldots,s\} \text{ and } \notag \\
 & \hspace{3cm} L^+(v,i-1,t)+L(v,i,s-t)\leq \beta \big\} , \label{eq:dp-update-Lp1}
\end{align}
\end{subequations}
for all $i\in\{1,2,\ldots,c(v)\}$ and $s\in\{1,2,\ldots,m\}$.
\end{lem}

The most interesting of these recursive relations are of course \eqref{eq:dp-update-L} and \eqref{eq:dp-update-Lp1}.
The relation \eqref{eq:dp-update-L} captures the two possibilities of either adding the edge $\{v,u_i\}$ or not adding it to a partial solution in the tree~$T_{u_i,c(u_i)}^+=T_{u_i}$ to obtain a solution for the tree~$T_{v,i}$ (recall \eqref{eq:load-add-edge}).
The relation \eqref{eq:dp-update-Lp1}, on the other hand, describes how to distribute~$s$ contraction edges in~$T_{v,i}^+$ among the two subtrees~$T_{v,i-1}^+$ and~$T_{v,i}$ ($t$ is the number of edges contracted in the first tree, and~$s-t$ the number of edges in the second tree, respectively).

\begin{proof}
The relations \eqref{eq:dp-init-L1} and \eqref{eq:dp-init-L2} follow immediately from the definitions of the trees $T_{v,i}$ and~$T_{v,i}^+$ and from \eqref{eq:def-L}.
The relation \eqref{eq:dp-update-L} follows from \eqref{eq:load-add-edge} and \eqref{eq:def-L}.
The relation \eqref{eq:dp-update-Lp1} follows from \eqref{eq:load-join-trees} and \eqref{eq:def-L} with the help of Lemma~\ref{lem:tree-partition}.
\end{proof}

We are now ready to prove Theorem~\ref{thm:dp}.

\begin{proof}[Proof of Theorem~\ref{thm:dp}]
Given the instance $(T,\ell,\varphi)$, we fix an arbitrary root~$r$ of~$T$ and an arbitrary ordering of the children of each vertex, making~$T$ an ordered rooted tree.
We then compute the entries of the matrices~$L$ and~$L^+$ using Lemma~\ref{lem:L-Lp-rec}.
We first initialize various entries using \eqref{eq:dp-init-L1} and \eqref{eq:dp-init-L2}, and compute the remaining entries in a bottom-up fashion moving upwards from the leaves to the root.
Specifically, at a vertex~$v$ with children $u_1,u_2,\ldots,u_{c(v)}$ for which all the entries of~$L$ and~$L^+$ have already been computed, we first compute $L(v,i,s)$ for all $i\in\{1,2,\ldots,c(v)\}$ and $s\in\{1,2,\ldots,m\}$ using \eqref{eq:dp-update-L}, and then $L^+(v,i,s)$ for all $i\in\{1,2,\ldots,c(v)\}$ and $s\in\{1,2,\ldots,m\}$ using \eqref{eq:dp-update-Lp1}.

Let~$s^*$ be the largest~$s$ such that $L^+(r,c(r),s)\leq \beta$.
From \eqref{eq:def-L} we obtain that~$s^*$ is the size of an optimal solution of the instance $(T,\ell,\varphi)$.
The corresponding set of edges $C\subseteq E$ can be obtained by keeping track of the arguments for which the minima and maxima in \eqref{eq:dp-update-L} and \eqref{eq:dp-update-Lp1} are attained in each step.

Clearly,~$L$ and~$L^+$ both have $\cO(n^2)$ entries, and computing each entry takes time $\cO(n)$, so the running time of our dynamic program is $\cO(n^3)$.
\end{proof}

\subsection{\WContraction{} on trees}
\label{sec:weak-dp}

In this section we consider the problem of computing weak contractions for a tree~$T$ with affine tolerance function $\varphi(x)=x/\alpha-\beta$.
Here, our main result is a dynamic programming algorithm that builds on the algorithmic ideas presented in Section~\ref{sec:dp}.

\begin{thm}
\label{thm:weak-dp}
Let~$T$ be a tree with edge lengths $\ell\colon E\to\RR_{>0}$ and consider the tolerance function $\varphi(x)=x/\alpha-\beta$, $\alpha\geq 1$, $\beta\geq 0$.
An optimal solution for the instance $(T,\ell,\varphi)$ of the problem \WContraction{} can be computed by dynamic programming in time $\cO(n^5)$.
\end{thm}

In this setting we need to specifically keep track of pairs of vertices whose distance remains positive when contracting a set of edges $C\subseteq E$ (i.e., not all edges in between these vertices are contracted).
To this end we extend the definitions \eqref{eq:load} as follows:
For any vertex~$v$ of~$T$ we define the \emph{weak load of~$T$ at~$v$} as
\begin{equation}
\label{eq:wload}
\wload_{C,\alpha}(T,v) := \max \{\load_{C,\alpha}(u,v): u\in V(T) \text{ and } \dist_{\ell_C}(u,v)>0\} .
\end{equation}
Note that in the maximization we have to consider all vertices~$u$ such that at least one edge on the path from~$u$ to~$v$ is not in~$C$.
This definition together with \eqref{eq:load-tree} yields $\wload_{C,\alpha}(T,v)\leq \load_{C,\alpha}(T,v)$.
In contrast to the load, the weak load may be negative.
In particular, $\wload_{C,\alpha}(T,v)=-\infty$ if and only if $C=E$.

The following lemma is the counterpart to Lemma~\ref{lem:tree-partition} for weak contractions.
It describes how to combine feasible solutions on subtrees to a feasible solution of the entire tree.
There is one important subtlety here:
While the notion of a weak contraction forbids contracting all edges of~$T$, we clearly have to allow this for partial solutions on subtrees of~$T$ (as long as some other edge not in the subtree is is not contracted, this might still yield a feasible solution).

\begin{lem}
\label{lem:weak-tree-partition}
Consider a partition of~$T$ into two subtrees~$T_1$ and~$T_2$ that only have a vertex~$v\in V$ in common.
Then $C\subsetneq E$ is a feasible solution for the instance $(T,\ell,\varphi)$ of the problem \WContraction{} if and only if the following two conditions hold: For~$i=1,2$, either~$C$ contains every edge of~$T_i$ or $C\cap E(T_i)$ is a feasible solution for the instance $(T_i,\ell,\varphi)$ of \WContraction{}; and we have
\begin{equation}
\label{eq:weak-tree-cond}
\load_{C,\alpha}(T_1,v)+\wload_{C,\alpha}(T_2,v)\leq \beta \quad \text{and} \quad \wload_{C,\alpha}(T_1,v)+\load_{C,\alpha}(T_2,v)\leq \beta .
\end{equation}
\end{lem}

\begin{proof}
Let $C\subsetneq E$.
For the rest of the proof we omit the subscripts~$C$ and~$\alpha$ and simply write $\load_{C,\alpha}=\load$ and $\wload_{C,\alpha}=\wload$.

We first assume that~$C$ is a feasible solution for the instance $(T,\ell,\varphi)$ of the problem \WContraction{}.
I.e., any two vertices $u,w$ of~$T$ with $\dist_{\ell_C}(u,w)>0$ satisfy the condition $\load(u,w)\leq \beta$.
This is true in particular for all pairs of vertices $u,w\in T_i$, $i=1,2$, implying that either $C\supseteq T_i$ or $C\cap T_i\subsetneq T_i$ is a feasible solution for the instance $(T_i,\ell,\varphi)$.
If $\wload(T_2,v)=-\infty$, the claimed inequality $\load(T_1,v)+\wload(T_2,v)\leq \beta$ is trivially satisfied.
So suppose that $\wload(T_2,v)$ is a finite number, and let $u\in T_1$ and $w\in T_2$ be such that $\load(u,v)=\load(T_1,v)$, and $\dist_{\ell_C}(v,w)>0$ as well as $\load(v,w)=\wload(T_2,v)$.
Then we also have $\dist_{\ell_C}(u,w)>0$, so we know that $\load(u,w)\leq \beta$ by the assumption that~$C$ is feasible for $(T,\ell,\varphi)$.
Combining this last inequality with the relation $\load(u,w)=\load(u,v)+\load(v,w)=\load(T_1,v)+\wload(T_2,v)$ proves that the right hand side of the equation is at most~$\beta$, as claimed.
The proof of the second inequality $\wload(T_1,v)+\load(T_2,v)\leq \beta$ works symmetrically.
This proves one direction of the equivalence.

To prove the reverse direction, we now assume that either $C\supseteq T_i$ or $C\cap T_i\subsetneq T_i$ is a feasible solution for the instance $(T_i,\ell,\varphi)$ for $i=1,2$, and that $\load(T_1,v)+\wload(T_2,v)\leq \beta$ and $\wload(T_1,v)+\load(T_2,v)\leq \beta$.
To show that~$C$ is a feasible solution for the instance $(T,\ell,\varphi)$, let $u\in T_1$ and $w\in T_2$ be such that $\dist_{\ell_C}(u,w)>0$.
It follows that $\dist_{\ell_C}(u,v)>0$ or $\dist_{\ell_C}(v,w)>0$.
We first consider the case that $\dist_{\ell_C}(u,v)>0$.
By the definitions \eqref{eq:load} and \eqref{eq:wload} we have $\load(u,v)\leq \wload(T_1,v)$, and also $\load(v,w)\leq \load(T_2,v)$, yielding $\load(u,w)=\load(u,v)+\load(v,w)\leq \wload(T_1,v)+\load(T_2,v)\leq \beta$ (the last inequality holds by assumption).
This proves that $\load(u,w)\leq \beta$, as desired.
The proof of the other case $\dist_{\ell_C}(v,w)>0$ works symmetrically.
This completes the proof of the lemma.
\end{proof}

As in Section~\ref{sec:dp}, we view~$T$ as an ordered rooted tree, and consider its subtrees $T_v$, $T_{v,i}$ and~$T_{v,i}^+$ for all $v\in V$ and $i\in\{0,1,\ldots,c(v)\}$ (recall the definitions given after Lemma~\ref{lem:tree-partition}).
Let us briefly highlight the differences between Lemmas~\ref{lem:tree-partition} and \ref{lem:weak-tree-partition}.
The dynamic programming algorithm presented in Section~\ref{sec:dp} exploits the fact that the optimal way to contract exactly $|C|=s$ edges in a subtree~$T_v$ of~$T$ rooted at a particular vertex~$v$ is to contract a set of edges that minimizes $\load_{C,\alpha}(T_v,v)$.
This is possible as the optimality condition in Lemma~\ref{lem:tree-partition} only depends on this parameter.
Here the situation is more complicated, as Lemma~\ref{lem:weak-tree-partition} also considers $\wload_{C,\alpha}(T_v,v)$.
Figure~\ref{fig:wload} illustrates that it is not sufficient to minimize only one of these parameters.

\begin{figure}[ht]
\subfloat[]{
\input{Figures/wload_1}}
\hspace{2cm}
\subfloat[]{
\input{Figures/wload_2}}
\hspace{2cm}
\subfloat[]{
\input{Figures/wload_3}}
\hspace{2cm}
\subfloat[]{
\input{Figures/wload_4}}
\caption{Example of the behavior of the parameters~$\load()$ and~$\wload()$ for $(\alpha,\beta)=(2,1/2)$.
Consider the tree~$T$ in (c) and (d) with length functions that differ only in the value they assign to the topmost edge.
Parts (a) and (b) of the figure show the subtree~$T_v$ of~$T$ and two different subsets of edges~$C$ and~$C'$ of~$T_v$, respectively, drawn by dashed edges.
We have $\load_{C,\alpha}(T_v,v)=\wload_{C,\alpha}(T_v,v) = 7/2-3 = 1/2$, while $\load_{C',\alpha}(T_v,v) = 7/2$ and $\wload_{C,\alpha}(T_v,v) = 7/2-4 = -1/2$, so no solution dominates the other.
For the length function in (c),~$C$ is feasible and optimal, but~$C'$ is not feasible.
For the length function in (d), on the other hand,~$C'$ is feasible and can be extended to an optimal solution of size 3, but~$C$ can not be extended.}
\label{fig:wload}
\end{figure}
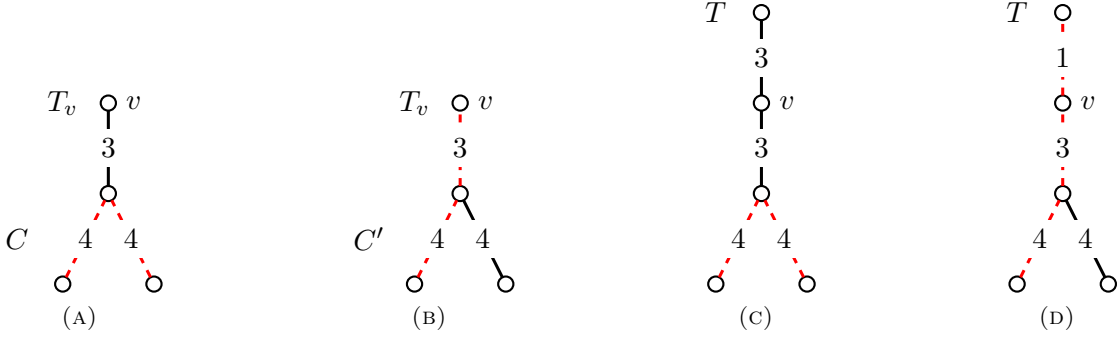

Consequently, we keep track of an entire Pareto front of non-dominated partial solutions (see Figure~\ref{fig:pareto}).
Formally, we define the \emph{set $F(T_v,s)$ of feasible partial solutions of size~$s$} as the family of all sets $C\subseteq E(T_v)$ with $|C|=s$ such that either $C=E(T_v)$ or~$C$ is a feasible solution for the instance $(T_v,\ell,\varphi)$ of \WContraction{}.
For two sets $C,C'\in F(T_v,s)$ we say that \emph{$C$ dominates~$C'$ at~$v$} if $\load_{C,\alpha}(T_v,v) \leq \load_{C',\alpha}(T_v,v)$ and $\wload_{C,\alpha}(T_v,v) \leq \wload_{C',\alpha}(T_v,v)$, and we define the \emph{Pareto front} $P(T_v,v,s)$ as a minimal family of sets $C\in F(T_v,s)$ such that no set $C'\in F(T_v,s)$ dominates~$C$ at~$v$.
Note that the domination relation is reflexive, so there may be several different such minimal families, all with the same pairs of load and weak load values, and any choice among them is equally good for us.
This definition is illustrated in Figure~\ref{fig:pareto}.

The following crucial lemma asserts that the number of points on the Pareto front, i.e., the size of the family $P(T_v,v,s)$ is at most~$n+1$.
This property is essential for our dynamic programming approach, and it does not follow immediately from the definition of $P(T_v,v,s)$, as the set of feasible solutions $F(T_v,s)$ is typically of exponential size.

\begin{lem}
\label{lem:pareto}
For any $C\subseteq E(T_v)$, we have $\load_{C,\alpha}(T_v,v) \in \Lambda(T_v,v):=\{\dist_\ell(u,v)/\alpha : u\in V(T_v)\}$ or $\load_{C,\alpha}(T_v,v)=\wload_{C,\alpha}(T_v,v)$.
Consequently, the Pareto front $P(T_v,v,s)$ has size at most~$n+1$.
\end{lem}

\begin{proof}
By the definitions \eqref{eq:load} and \eqref{eq:wload} we have $\wload_{C,\alpha}(T_v,v)\leq \load_{C,\alpha}(T_v,v)$ for all $C\subseteq E(T_v)$.
Now let $C\subseteq E(T_v)$ be such that $\wload_{C,\alpha}(T_v,v)<\load_{C,\alpha}(T_v,v)$.
Again by the previously mentioned definitions this implies that $\load_{C,\alpha}(T_v,v)=\load_{C,\alpha}(u,v)=\dist_\ell(u,v)/\alpha$ for some $u\in V(T_v)$, which is indeed an element of the set $\Lambda(T_v,v)$.
Consequently, the Pareto front $P(T_v,v,s)$ consists of at most one set $C\in F(T_v,s)$ with $\wload_{C,\alpha}(T_v,v)=\load_{C,\alpha}(T_v,v)$ and at most one set $C\in F(T_v,s)$ with $\wload_{C,\alpha}(T_v,v)<\load_{C,\alpha}(T_v,v)$ for each number in $\Lambda(T_v,v)$.
Using that $|\Lambda(T_v,v)|\leq |V(T_v)|\leq n$ it follows that $|P(T_v,v,s)|\leq n+1$.
\end{proof}

\begin{figure}[ht]
\centering
\input{Figures/pareto_1}
\input{Figures/pareto_2}
\caption{The right hand side of the figure shows the set of all pairs of $(\load,\wload)$ values for all feasible solutions of size $s=3$ for the rooted tree $(T_v,v)$ shown on the left for the parameter values $(\alpha,\beta)=(2,13/2)$.
The points in $P(T_v,v,s)$ are highlighted with small circles.
The vertical dashed lines represent all values in the set $\Lambda(T_v,v)=\{0,1/2,3/2,2,3,7/2,4\}$.
Moreover, we have $\lambda^*(T_v,v,s)=1$.}
\label{fig:pareto}
\end{figure}
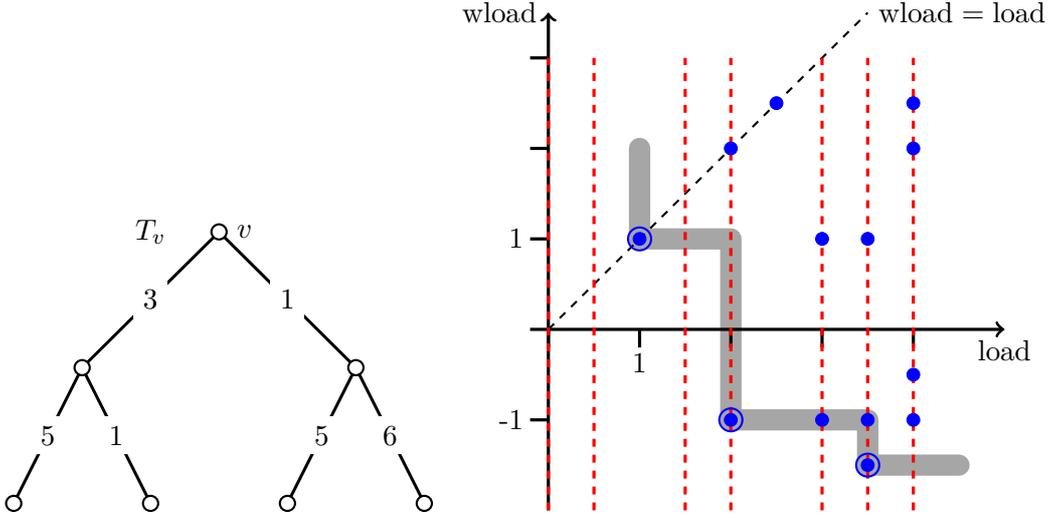

By Lemma~\ref{lem:pareto} the load values of all points on the Pareto front with $\wload()<\load()$ are in the set $\Lambda(T_v,v)$.
There might also be one point with $\wload()=\load()$ on the Pareto front (as in the example shown in Figure~\ref{fig:pareto}), and this load value might not be an element of $\Lambda(T_v,v)$.
We extend the set $\Lambda(T_v,v)$ accordingly by defining for $s\in\{0,1,\ldots,m\}$ (recall that $m=|E|$)
\begin{subequations}
\label{eq:lambda-Lambda*}
\begin{align}
\lambda^*(T_v,v,s) &:= \min \{\load_{C,\alpha}(T_v,v) : C \in F(T_v,s)\} , \label{eq:lambda*} \\
\Lambda^*(T_v,v,s) &:= \Lambda(T_v,v) \cup \{\lambda^*(T_v,v,s)\} . \label{eq:Lambda*}
\end{align}
\end{subequations}
If the set $F(T_v,s)$ is empty, we have $\lambda^*(T_v,v,s)=\infty$.

We now describe recursive relations for the weak load that are analogous to \eqref{eq:load-add-edge} and \eqref{eq:load-join-trees} for the load.
It follows straightforwardly from \eqref{eq:load} and \eqref{eq:wload} that for any vertex~$v$ of~$T$ and its children $u_i$, $i=1,2,\ldots,c(v)$, and for any set of edges $C\subseteq E(T_{u_i})$ we have
\begin{subequations}
\label{eq:wload-add-edge}
\begin{align}
\wload_{C\cup\{\{v,u_i\}\},\alpha}(T_{v,i},v) &= \wload_{C,\alpha}(T_{u_i},u_i)+\ell(v,u_i)/\alpha , \label{eq:wload-add-edge-yes} \\
\wload_{C,\alpha}(T_{v,i},v) &= \load_{C,\alpha}(T_{u_i},u_i) - (1-1/\alpha)\ell(v,u_i) . \label{eq:wload-add-edge-no}
\end{align}
\end{subequations}
Note that the weak load increases if the edge $\{v,u_i\}$ is added (see \eqref{eq:wload-add-edge-yes}).
On the other hand, if the edge $\{v,u_i\}$ is not added, it may decrease or increase (the right hand side of \eqref{eq:wload-add-edge-no} refers to the load, \emph{not} to the weak load).
Moreover, for any set of edges $C\subseteq T_{v,i}^+$ and any $i=1,2,\ldots,c(v)$ the definition \eqref{eq:wload} readily implies
\begin{equation}
\label{eq:wload-join-trees}
\wload_{C,\alpha}(T_{v,i}^+,v) = \max \{\wload_{C,\alpha}(T_{v,i-1}^+,v), \wload_{C,\alpha}(T_{v,i},v) \} .
\end{equation}
These rules together with the corresponding relations \eqref{eq:load-add-edge} and \eqref{eq:load-join-trees} allow us to compute the weak load and the load of all Pareto optimal partial solutions in a bottom-up fashion, similar to the approach taken in Section~\ref{sec:dp}.
Before it was sufficient to compute one optimal partial solution for every subtree~$T_{v,i}$ and~$T_{v,i}^+$, $i\in\{1,2,\ldots,c(v)\}$, and every possible size~$s$ of the contracted set of edges, but now our dynamic program keeps track of the entire Pareto fronts $P(T_{v,i},v,s)$ and $P(T_{v,i}^+,v,s)$.
We store the corresponding pairs of load and weak load values on the Pareto front in separate four-dimensional matrices $W$, $W^+$, $L$ and~$L^+$ (the entries of~$W$ and~$W^+$ are certain weak load values, and the entries of~$L$ and~$L^+$ are the corresponding load values).
We begin defining these matrices in an abstract way, and then establish several recursive relations which directly translate into a dynamic programming algorithm.
Specifically, for $v\in V$, $i\in\{0,1,\ldots,c(v)\}$, $s\in\{0,1,\ldots,m\}$ and $\lambda\in \Lambda(T_{v,i},v)$ with $\Lambda(T_{v,i},v)$ as in Lemma~\ref{lem:pareto} we define
\begin{subequations}
\label{eq:weak-def-LW}
\begin{align}
W(v,i,s,\lambda) &:= \min \{\wload_{C,\alpha}(T_{v,i},v): C \in F(T_{v,i},s) \text{ with } \load_{C,\alpha}(T_{v,i},v) \leq \lambda\} , \label{eq:weak-def-W} \\
L(v,i,s,\lambda) &:= \min \{\load_{C,\alpha}(T_{v,i},v) : C \in F(T_{v,i},s), \wload_{C,\alpha}(T_{v,i},v) = W(v,i,s,\lambda)\} . \label{eq:weak-def-L}
\end{align}
\end{subequations}
If there is no set~$C$ satisfying these requirements, we have $W(v,i,s,\lambda)=L(v,i,s,\lambda)=\infty$.
The entries of $W^+(v,i,s,\lambda)$ and $L^+(v,i,s,\lambda)$ are defined analogously to \eqref{eq:weak-def-LW} by considering the tree~$T_{v,i}^+$ instead of~$T_{v,i}$ (in particular, in this case we have $\lambda\in\Lambda(T_{v,i}^+,v)$).

The definitions of $W(v,i,s,\lambda)$ and $L(v,i,s,\lambda)$ given in \eqref{eq:weak-def-LW} extend straightforwardly to the value $\lambda=\lambda^*(T_{v,i},v,s)$ defined in \eqref{eq:lambda*}.
Similarly, the definitions of $W^+(v,i,s,\lambda)$ and $L^+(v,i,s,\lambda)$ from before extend to the value $\lambda=\lambda^*(T_{v,i}^+,v,s)$.
It is easy to see that we have in fact
\begin{equation}
\label{eq:L-lambda*}
L(v,i,s,\lambda^*(T_{v,i},v,s)) = \lambda^*(T_{v,i},v,s)
\end{equation}
(an analogous relation holds for the entries of~$L^+$).

The recursive relations satisfied by the matrices $W$, $L$, $W^+$ and~$L^+$ defined before are captured by the following two lemmas.
The initialization steps and the recursive computation of~$W$ and~$L$ are treated in Lemma~\ref{lem:W-L-rec1}.
The recursive computation of~$W^+$ and~$L^+$ is somewhat more technical, and is treated separately in Lemma~\ref{lem:W-L-rec2}.

\begin{lem}
\label{lem:W-L-rec1}
Let~$v$ be a vertex of~$T$ and let $u_1,u_2,\ldots,u_{c(v)}$ be the children of~$v$.
Then the matrices $W$, $W^+$, $L$ and~$L^+$ defined in and directly after \eqref{eq:weak-def-LW} satisfy the relations
\begin{subequations}
\begin{align}
W(v,0,0,0) &= W^+(v,0,0,0) = -\infty , \label{eq:weak-dp-init-W1} \\
W(v,i,0,\lambda) &= -(1-1/\alpha)\ell(v,u_i) \quad \text{for all } i\in\{1,2,\ldots,c(v)\} \text{ and } \lambda\in\Lambda^*(T_{v,i},v,0) , \label{eq:weak-dp-init-W2} \\
L(v,i,0,\lambda) &= 0 \quad \text{for all } i\in\{0,1,\ldots,c(v)\} \text{ and } \lambda\in\Lambda^*(T_{v,i},v,0) , \label{eq:weak-dp-init-L1} \\
L^+(v,i,0,\lambda) &= 0 \quad \text{for all } i\in\{0,1,\ldots,c(v)\} \text{ and } \lambda\in\Lambda^*(T_{v,i}^+,v,0) , \label{eq:weak-dp-init-L2} \\
W(v,0,s,\lambda) &= L(v,0,s,\lambda) = \infty \quad \text{for all } s\in\{1,2,\ldots,m\} \text{ and } \lambda\in\Lambda^*(T_{v,0},v,s) , \label{eq:weak-dp-init-WL1} \\
W^+(v,0,s,\lambda) &= L^+(v,0,s,\lambda) = \infty \quad \text{for all } s\in\{1,2,\ldots,m\} \text{ and } \lambda\in\Lambda^*(T_{v,0}^+,v,s) . \label{eq:weak-dp-init-WL2}
\end{align}
Furthermore, we have
\begin{align}
W(v,i,s,\lambda) &= \begin{cases}
          \mu & \text{if } \lambda = 0 , \\
          \min\{\mu,\nu\} & \text{if } \lambda > 0 \text{ and } \mu\leq \min\{\beta,\lambda\}  , \\
          \nu & \text{if } \lambda > 0 \text{ and } \mu > \min\{\beta,\lambda\} \text{ and } \nu \leq \beta , \\
          \infty & \text{otherwise} ,
          \end{cases} \label{eq:weak-dp-update-W1} \\
L(v,i,s,\lambda) &= \begin{cases}
          \max\{\mu,0\} & \text{if } W(v,i,s,\lambda) = \mu , \\
          L^+(u_i,c(u_i),s-1,\lambda-\ell(v,u_i)/\alpha)+\ell(v,u_i)/\alpha & \text{otherwise}  ,
          \end{cases} \label{eq:weak-dp-update-L1}
\end{align}
with $\mu:=\lambda^*(T_{u_i},u_i,s) - (1-1/\alpha)\ell(v,u_i)$ and $\nu:=W^+(u_i,c(u_i),s-1,\lambda-\ell(v,u_i)/\alpha)+\ell(v,u_i)/\alpha$ for all $i\in\{1,2,\ldots,c(v)\}$, $s\in\{1,2,\ldots,m\}$ and $\lambda\in \Lambda(T_{v,i},v)$.

Finally, we have
\begin{align}
L(v,i,s,\lambda^*) &= \lambda^* = \begin{cases}
            \min\{\lambda+\ell(v,u_i)/\alpha,\max \{0, \mu\}\} & \text{if } \mu \leq \beta , \\
            \lambda+\ell(v,u_i)/\alpha & \text{otherwise} ,
            \end{cases} \label{eq:weak-dp-update-L2} \\
W(v,i,s,\lambda^*) &= \begin{cases}
            \rho & \text{if } L(v,i,s,\lambda^*) = \lambda+\ell(v,u_i)/\alpha , \\
            \mu & \text{otherwise} ,
            \end{cases} \label{eq:weak-dp-update-W2}
\end{align}
\end{subequations}
where $\lambda \in \Lambda^*(T_{u_i},u_i,s-1)$ is minimal such that $\rho := W^+(u_i,c(u_i),s-1,\lambda)+\ell(v,u_i)/\alpha \leq \beta$, if such a value~$\lambda$ exists, and $\lambda := \rho := \infty$ otherwise, for all $i\in\{1,2,\ldots,c(v)\}$, $s\in\{1,2,\ldots,m\}$ and $\lambda^*=\lambda^*(T_{v,i},v,s)$.
\end{lem}

Note that the relations \eqref{eq:weak-dp-init-W1}--\eqref{eq:weak-dp-init-WL2} are the initialization steps, and the relations \eqref{eq:weak-dp-update-W1}--\eqref{eq:weak-dp-update-W2} capture the two possibilities of either adding or not adding the edge $\{v,u_i\}$ to a partial solution in the tree $T_{u_i,c(u_i)}^+=T_{u_i}$ to obtain a solution for the tree~$T_{v,i}$ (recall \eqref{eq:load-add-edge} and \eqref{eq:wload-add-edge}).

We only refer to well-defined entries of~$W^+$ and~$L^+$ in~\eqref{eq:weak-dp-update-L1} and in the definition of~$\nu$, as $\lambda-\ell(v,u_i)/\alpha \in \Lambda(T_{u_i}, u_i)$ holds for every $\lambda \in \Lambda(T_{v,i},v) \setminus \{0\}$.
Note that we either have $\nu \leq \lambda$ or $\nu = \infty$, while~$\mu$ may also take a value in the open interval $(\lambda,\infty)$.

\begin{proof}
The relations~\eqref{eq:weak-dp-init-W1}--\eqref{eq:weak-dp-init-WL2} follow immediately from the definitions of the trees~$T_{v,i}$ and~$T_{v,i}^+$ and the definitions of the respective matrices given in~\eqref{eq:weak-def-LW} and afterwards.
The relations~\eqref{eq:weak-dp-update-W1} and~\eqref{eq:weak-dp-update-L2} follow from~\eqref{eq:wload-add-edge} and the definitions of~$W$ and~$L$, respectively:
Consider a partial solution $C\in F(T_{v_i},s)$.
If $\load_{C,\alpha}(T_{v,i},v) = 0$, then~$C$ does not contain the edge $\{v,u_i\}$, so we have $W(v,i,s,0)=\mu$.
The other cases of~\eqref{eq:weak-dp-update-W1} as well as \eqref{eq:weak-dp-update-L2} are implied by the following observation:
If $\wload_{C,\alpha}(T_{v,i},v) \leq \lambda$ and $\{v, u_i\} \notin C$, then by~\eqref{eq:wload-add-edge} we have $\mu \leq \beta$ and $\mu \leq \lambda$.

The relation~\eqref{eq:weak-dp-update-L1} is closely related to~\eqref{eq:weak-dp-update-W1}.
If $\mu \neq \nu$, then~\eqref{eq:weak-dp-update-L1} follows immediately from~\eqref{eq:weak-dp-update-W1} and the definitions of~$W$ and~$L$.
If $\mu = \nu \leq \min\{\beta,\lambda\}$, then both a partial solution containing the edge $\{v,u_i\}$ as well as one missing this edge minimize the weak load.
As the weak load is bounded from above by the load, we get $L(v,i,s,\lambda) = W(v,i,s,\lambda) = \mu$ in this case.
This implies~\eqref{eq:weak-dp-update-L1}.
An analogous argument yields~\eqref{eq:weak-dp-update-W2}.
\end{proof}

The following lemma describes the recursive relations satisfied by the entries of~$W^+$ and~$L^+$.
Specifically, the lemma describes how to distribute~$s$ contraction edges in~$T_{v,i}^+$ among the two subtrees~$T_{v,i-1}^+$ and~$T_{v,i}$ ($t$ is the number of edges contracted in the first tree, and~$s-t$ the number of edges in the second tree, respectively).
To compute a single point on the Pareto front $P(T_{v,i}^+,v,s)$, we need to consider all points on the Pareto fronts $P(T_{v,i-1}^+,v,t)$ and $P(T_{v,i},v,s-t)$.

\begin{lem}
\label{lem:W-L-rec2}
Let~$v$ be a vertex of~$T$, and let $s \in\{0,1,\dots,m\}$ and $i\in\{1,2,\dots,c(v)\}$ be fixed throughout this lemma.
For $t\in \{0,1,\ldots,s\}$ we let $\Pi(t)$ denote the set of all pairs $(\lambda_1,\lambda_2)$ with $\lambda_1\in \Lambda^*(T_{v,i-1}^+,v,t)$ and $\lambda_2\in \Lambda^*(T_{v,i},v,s-t)$ such that $W^+(v,i-1,t,\lambda_1)+L(v,i,s-t,\lambda_2)\leq \beta$ and $L^+(v,i-1,t,\lambda_1)+W(v,i,s-t,\lambda_2)\leq \beta$.
For $t\in \{0,1,\ldots,s\}$ and $\lambda\in \Lambda(T_{v,i}^+,v)$ we let $\Pi(t,\lambda)\subseteq\Pi(t)$ denote the set of all pairs $(\lambda_1,\lambda_2)\in\Pi(t)$ satisfying $\max\{L^+(v,i-1,t,\lambda_1),L(v,i,s-t,\lambda_2)\}\leq \lambda$. \\
For all $\lambda\in\Lambda(T_{v,i}^+,v)$, defining
\begin{subequations}
\begin{align}
W(t) &:= \min\big\{\max\{W^+(v,i-1,t,\lambda_1),W(v,i,s-t,\lambda_2)\} : (\lambda_1,\lambda_2)\in \Pi(t,\lambda) \big\} , \label{eq:Wt} \\
L(t) &:= \min\big\{\max\{L^+(v,i-1,t,\lambda_1),L(v,i,s-t,\lambda_2)\} : (\lambda_1,\lambda_2) \text{ minimizes } \eqref{eq:Wt} \big\} , \label{eq:Lt}
\end{align}
we have
\begin{align}
W^+(v,i,s,\lambda) &= \min\{W(t):t\in\{0,1,\ldots,s\}\} , \label{eq:weak-dp-update-W3} \\
L^+(v,i,s,\lambda) &= \min\{L(t):t \text{ minimizes } \eqref{eq:weak-dp-update-W3} \} . \label{eq:weak-dp-update-L3}
\end{align}
For $\lambda^*=\lambda^*(T_{v,i}^+,v,s)$, defining
\begin{align}
L^*(t) &:= \min\big\{\max\{L^+(v,i-1,t,\lambda_1),L(v,i,s-t,\lambda_2)\} : (\lambda_1,\lambda_2)\in \Pi(t) \big\} , \label{eq:Lt*} \\
W^*(t) &:= \min\big\{\max\{W^+(v,i-1,t,\lambda_1),W(v,i,s-t,\lambda_2)\} : (\lambda_1,\lambda_2) \text{ minimizes } \eqref{eq:Lt*} \big\} , \label{eq:Wt*}
\end{align}
we have
\begin{align}
L^+(v,i,s,\lambda^*) &= \lambda^* = \min\{L^*(t):t\in\{0,1,\ldots,s\}\} , \label{eq:weak-dp-update-L4} \\
W^+(v,i,s,\lambda^*) &= \min\{W^*(t):t \text{ minimizes } \eqref{eq:weak-dp-update-L4} \} . \label{eq:weak-dp-update-W4}
\end{align}
\end{subequations}
\end{lem}

\begin{proof}
The relation \eqref{eq:weak-dp-update-W3} follows by combining the definitions \eqref{eq:weak-def-W} and \eqref{eq:Wt} with the relations \eqref{eq:load-join-trees}, \eqref{eq:wload-join-trees} and the condition \eqref{eq:weak-tree-cond} from Lemma~\ref{lem:weak-tree-partition}.
The argument for \eqref{eq:weak-dp-update-L3} is analogous, using the definitions \eqref{eq:weak-def-L} and \eqref{eq:Lt} instead of \eqref{eq:weak-def-W} and \eqref{eq:Wt}.

The relation \eqref{eq:weak-dp-update-L4} follows by combining the definitions \eqref{eq:lambda*} and \eqref{eq:Lt*} (recall also \eqref{eq:L-lambda*}) with the relations \eqref{eq:load-join-trees}, \eqref{eq:wload-join-trees} and the condition \eqref{eq:weak-tree-cond} from Lemma~\ref{lem:weak-tree-partition}.
The argument for \eqref{eq:weak-dp-update-W4} is analogous, using the definitions \eqref{eq:weak-def-W} and \eqref{eq:Wt*} instead of \eqref{eq:lambda*} and \eqref{eq:Lt*}.
\end{proof}

We can trivially compute the quantities $W(t)$, $L(t)$, $W^*(t)$ and~$L^*(t)$ as defined in Lemma~\ref{lem:W-L-rec2} in time $\cO(n^2)$ (using that $|\Pi(t)|=\cO(n^2)$ and $|\Pi(t,\lambda)|=\cO(n^2)$ by Lemma~\ref{lem:pareto}).
The following lemma shows how to do the same computation in time $\cO(n)$, so that the entries $W^+(v,i,s,\lambda)$ and $L^+(v,i,s,\lambda)$ can be computed via \eqref{eq:weak-dp-update-W3}, \eqref{eq:weak-dp-update-L3}, \eqref{eq:weak-dp-update-L4} and \eqref{eq:weak-dp-update-W4} in time $\cO(n^2)$ (instead of the trivial bound $\cO(n^3)$).

\begin{lem}
\label{lem:wt-lt-comp}
If the numbers in the sets $\Lambda^*(T_{v,i-1}^+,v,t)$ and $\Lambda^*(T_{v,i},v,s-t)$ are sorted increasingly, the quantities $W(t)$, $L(t)$, $W^*(t)$ and~$L^*(t)$ defined in Lemma~\ref{lem:W-L-rec2} can be computed in time $\cO(n)$.
Consequently, $W^+(v,i,s,\lambda)$ and $L^+(v,i,s,\lambda)$ can be computed for all $s\in\{0,1,\ldots,m\}$ and all $\lambda\in\Lambda^*(T_{v,i}^+,v,s)$ in time $\cO(n^2)$.
\end{lem}

\begin{proof}
We define the sequence~$P_1$ of all pairs of finite numbers $(L^+(v,i-1,t,\lambda),W^+(v,i-1,t,\lambda))$ for all $\lambda\in\Lambda^*(T_{v,i-1}^+,v,t)$ in increasing order of $\lambda$-values.
Similarly, we define the sequence~$P_2$ of all pairs of finite numbers $(L(v,i,s-t,\lambda),W(v,i,s-t,\lambda))$ for all $\lambda\in\Lambda^*(T_{v,i},v,s-t)$ in increasing order of $\lambda$-values.
By Lemma~\ref{lem:pareto} each of these lists has size $\cO(n)$.
Note that these sequences correspond to the Pareto fronts $P(T_{v,i-1}^+,v,t)$ and $P(T_{v,i},v,s-t)$, respectively.
Some pairs of points may appear multiple times consecutively in~$P_1$ and~$P_2$, and in a preprocessing step we eliminate these duplicates in time $\cO(n)$.
We know that after this preprocessing step, the first entries in the simplified lists~$P_1$ and~$P_2$ are strictly increasing, and the second entries are strictly decreasing (recall Figure~\ref{fig:pareto}).

We first argue how to compute~$W(t)$ and~$L(t)$.
We begin discarding all pairs from each list whose first entry ($L^+$ or~$L$, respectively) is strictly greater than~$\lambda$ in time $\cO(n)$.
We then process the remaining lists~$P_1$ and~$P_2$ beginning at the last entries $(L_j^+,W_j^+)$ and $(L_k,W_k)$ (with smallest~$W^+$ or~$W$-values, respectively) in two phases.

In the first phase we compute~$W(t)$ as follows:
If $L_j^+ +W_k>\beta$, we discard the last element of~$P_1$ by decreasing~$j$ by 1 (by our sorting of the lists we know that $L_j^+ +W_{k'}>\beta$ for all $k'\leq k$).
If $W_j^+ +L_k>\beta$, we discard the last element of~$P_2$ by decreasing~$k$ by 1 (by our sorting of the lists we know that $W_{j'}^++L_k>\beta$ for all $j'\leq j$).
Once $L_j^+ +W_k\leq \beta$ and $W_j^+ +L_k\leq \beta$ for the first time, we have found $W(t)=\max\{W_j^+,W_k\}$. If this never happens we know that $W(t)=\infty$.
This computation is correct by the definition of $\Pi(t,\lambda)$ in Lemma~\ref{lem:W-L-rec2} and by \eqref{eq:Wt}, and it takes time $\cO(n)$.

In the second phase we compute~$L(t)$ as follows:
If $W(t)=\infty$, we know that $L(t)=\infty$, too.
Otherwise we distinguish two cases:
If $W_j^+\geq W_k$, we decrease~$k$ further as long as both inequalities $W_j^+\geq W_k$ and $L_j^+ +W_k\leq \beta$ are still satisfied (so that they still hold for the final~$k$).
If $W_j^+\leq W_k$, we decrease~$j$ further as long as both inequalities $W_j^+\leq W_k$ and $W_j^+ +L_k\leq \beta$ are still satisfied (so that they still hold for the final~$j$).
In the end we set $L(t)=\max\{L_j^+,L_k\}$.
Note that in the first case, the third constraint $W_j^+ +L_k\leq \beta$ remains valid by the monotonicity $L_{k'}\leq L_k$ for all $k'\leq k$, and in the second case, the third constraint $L_j^+ +W_k\leq \beta$ remains valid by the monotonicity $L_{j'}^+\leq L_j^+$ for all $j'\leq j$.
Therefore, the correctness of the computation of~$L(t)$ follows from \eqref{eq:Lt}.

The procedure to compute~$W^*(t)$ and~$L^*(t)$ processes~$P_1$ and~$P_2$ (as obtained from the preprocessing step explained in the beginning) starting at the first entries $(L_j^+,W_j^+)$,~$j=1$, and $(L_k,W_k)$,~$k=1$, in two phases very similarly to before.
We omit the details here.
\end{proof}

We are now ready to prove Theorem~\ref{thm:weak-dp}.

\begin{proof}[Proof of Theorem~\ref{thm:weak-dp}]
Given the instance $(T,\ell,\varphi)$, we fix an arbitrary root~$r$ of~$T$ and an arbitrary ordering of the children of each vertex, making~$T$ an ordered rooted tree.

We begin precomputing and sorting all of the sets $\Lambda(T_{v,i},v)$ and $\Lambda(T_{v,i}^+,v)$, $v\in V$, $i\in\{0,1,\ldots,c(v)\}$, and we maintain them as sorted lists throughout the algorithm.
This takes time $\cO(n^2\log n)$ in total (recall Lemma~\ref{lem:pareto}).

We then compute the entries of the matrices $W$, $L$, $W^+$ and~$L^+$ using Lemmas~\ref{lem:W-L-rec1} and \ref{lem:wt-lt-comp}.
We first initialize various entries using \eqref{eq:weak-dp-init-W1}--\eqref{eq:weak-dp-init-WL2}, and compute the remaining entries in a bottom-up fashion moving upwards from the leaves to the root.
Specifically, at a vertex~$v$ with children $u_1,u_2,\ldots,u_{c(v)}$ for which all the entries of $W$, $L$, $W^+$ and~$L^+$ have already been computed, we first compute $W(v,i,s,\lambda)$ and then $L(v,i,s,\lambda)$ for all $i\in\{1,2,\ldots,c(v)\}$, $s\in\{1,2,\ldots,m\}$ and $\lambda \in \Lambda(T_{v,i},v)$ using \eqref{eq:weak-dp-update-W1} and \eqref{eq:weak-dp-update-L1}, then we compute $L(v,i,s,\lambda^*(T_{v,i},v,s))=\lambda^*(T_{v,i},v,s)$ and $W(v,i,s,\lambda^*(T_{v,i},v,s))$ for all $i\in\{1,2,\ldots,c(v)\}$ and $s\in\{1,2,\ldots,m\}$ using \eqref{eq:weak-dp-update-L2} and \eqref{eq:weak-dp-update-W2}.
We obtain sorted lists containing the numbers in $\Lambda^*(T_{v,i},v,s)$ by inserting $\lambda^*(T_{v,i},v,s)$ at the correct position into the precomputed list $\Lambda(T_{v,i},v)$.
Next, we compute $W^+(v,i,s,\lambda)$ and then $L^+(v,i,s,\lambda)$ for all $i\in\{1,2,\ldots,c(v)\}$, $s\in\{1,2,\ldots,m\}$ and $\lambda\in \Lambda(T_{v,i}^+,v)$ using \eqref{eq:weak-dp-update-W3} and \eqref{eq:weak-dp-update-L3}, and then we compute $L^+(v,i,s,\lambda^*(T_{v,i}^+,v,s))=\lambda^*(T_{v,i}^+,v,s)$ and $W^+(v,i,s,\lambda^*(T_{v,i}^+,v,s))$ for all $i\in\{1,2,\ldots,c(v)\}$ and $s\in\{1,2,\ldots,m\}$ using \eqref{eq:weak-dp-update-L4} and \eqref{eq:weak-dp-update-W4}.
We obtain sorted lists containing the numbers in $\Lambda^*(T_{v,i}^+,v,s)$ by inserting $\lambda^*(T_{v,i}^+,v,s)$ at the correct position into the precomputed list $\Lambda(T_{v,i}^+,v)$.

Let~$s^*$ be the largest~$s$ such that $W^+(r,c(r),s,\lambda^*(T,r,s))$ is finite.
From \eqref{eq:weak-def-LW} we obtain that~$s^*$ is the size of an optimal solution of the instance $(T,\ell,\varphi)$.
The corresponding set of edges $C\subseteq E$ can be obtained by keeping track of the arguments for which the minima and maxima in \eqref{eq:weak-dp-update-W1}--\eqref{eq:weak-dp-update-W2} and \eqref{eq:Wt}--\eqref{eq:weak-dp-update-W4} are attained in each step.

Each of the matrices $W$, $L$, $W^+$ and~$L^+$ has $\cO(n^3)$ entries (recall Lemma~\ref{lem:pareto}).
Computing an entry of~$W$ or~$L$ takes $\cO(n)$ time by Lemma~\ref{lem:W-L-rec1}, while computing an entry of~$W^+$ or~$L^+$ can be achieved in time $\cO(n^2)$ by Lemma~\ref{lem:wt-lt-comp}, so the runnning time of our dynamic program is $\cO(n^5)$.
\end{proof}

\section{Hardness for additive tolerance functions}
\label{sec:hard-add}

In this section we prove that the problems \Contraction{} and \WContraction{} for the tolerance function $\varphi(x)=x-\beta$ (purely additive error) are hard already on cycles (Section~\ref{sec:cycles-hard} below).
We then prove that \Contraction{} with the same tolerance function is hard to approximate for general graphs and for bipartite graphs (Section~\ref{sec:inapx-contraction}).

\subsection{Hardness of \Contraction{} and \WContraction{}}
\label{sec:cycles-hard}

Recall that we can compute optimal (weak) $(\alpha,\beta)$-contractions in polynomial time on trees (this was shown in Section~\ref{sec:dp}), and have a linear time algorithm for \Contraction{} on cycles with unit length edges (this was shown in Section~\ref{sec:cycles}).
We now show that the problem with $\alpha=1$ is NP-hard on cycles with arbitrary edge lengths.

\begin{thm}
\label{thm:cycle-hard-fix}
For any fixed $\beta>0$, the problems \Contraction{} and \WContraction{} with tolerance function $\varphi(x)=x-\beta$, $\beta\geq 0$, are NP-hard on cycles.
\end{thm}

Theorem~\ref{thm:cycle-hard-fix} (where~$\beta$ is not part of the input) follows immediately from Theorem~\ref{thm:cycle-hard} below (where~$\beta$ is part of the input).
The reason is that an instance with $\alpha=1$ does not change when multiplying all edge lengths and $\beta$ by some constant.

\begin{thm}
\label{thm:cycle-hard}
The problems \Contraction{} and \WContraction{} with tolerance function $\varphi(x)=x-\beta$, $\beta\geq 0$, are NP-hard on cycles.
\end{thm}

The rest of this section is devoted to proving Theorem~\ref{thm:cycle-hard}.

For our proof we will use the following variant of the well-known problem \Partition{}, referred to as \CPartition{}.
To state the problem we say that a set of positive rational numbers $\{a_1,a_2,\ldots,a_n\}$ is \emph{close to 1}, if $\sum_{i=1}^n a_i = n$ and $\varepsilon:=\sum_{i=1}^n |a_i-1|<1/5$.

\begin{problem}{\CPartition{}}
Input: & A set of positive rational numbers $\{a_1,a_2,\ldots,a_n\}$ that is close to 1. \\
Output: & `Yes' if there is a subset $I \subseteq [n]$ such that $\sum_{i \in I} a_i = \sum_{i\in [n]\setminus I} a_i$, `No' otherwise. \\
\end{problem}

Note that for a `Yes'-instance of this problem, the solution $I\subseteq [n]$ must have size~$n/2$, so $|I|=|[n]\setminus I|=\sum_{i \in I} a_i=\sum_{i\in [n]\setminus I} a_i=n/2$.
In particular, this implies that~$n$ is even.

In the classical problem \Partition{}, the input set is not constrained to be close to 1.
\Partition{} was shown to be NP-complete already in Karp's seminal paper~\cite{MR0378476}.
The fact that \CPartition{} is also NP-complete follows from a straightforward rescaling argument.

\begin{lem}
\label{lem:cpartition}
\CPartition{} is NP-complete.
\end{lem}

\begin{proof}
Given an instance $\{a_1,a_2,\ldots,a_n\}$ of \Partition{}, we first add~$n$ additional zeroes $a_{n+1}=a_{n+2}=\cdots=a_{2n}=0$ to the instance (by this we ensure that a partition with equal sums is transformed into one where both partition classes have the same number~$n$ of summands).
We then linearly transform all the~$a_i$ according to $a_i':=(a_i+C)/D$, where~$C$ and~$D$ are sufficiently large constants so that the transformed values~$a_i'$ are close to 1.
The transformed set of numbers has even cardinality~$2n$, is close to 1, and it admits a partition into two sets of size~$n$ with equal sum if and only if the original instance allows a partition into two sets with equal sum.
\end{proof}

\begin{proof}[Proof of Theorem~\ref{thm:cycle-hard}]
We first focus on the problem \Contraction{}.
We reduce \CPartition{}, which is NP-complete by Lemma~\ref{lem:cpartition}, to the problem \Contraction{} on a cycle with tolerance function $\varphi(x)=x-\beta$, $\beta\geq 0$.

Let $\cI=\{a_1,a_2,\ldots,a_n\}$ be an instance of \CPartition{} such that $a_1\geq a_2 \geq \cdots \geq a_n$.
This ensures that all~$a_i$ that are bigger than 1 appear before all~$a_i$ that are smaller than 1, which is the only property of the ordering that we exploit in the proof later on.
The instance of \Contraction{} we construct is on the cycle $C_{2n+4}$ with~$2n+4$ edges.
We label the vertices of the cycle by walking around the cycle as follows:
The first~$n+1$ vertices are labelled $u_0,u_1,\ldots,u_n$, then there are two special vertices $v_1$, $v_2$, and the remaining~$n+1$ vertices are labelled $w_0,w_1,\ldots,w_n$, see Figure~\ref{fig:cycle-hard}.
We denote the subpath $(u_0,\ldots,u_n)$ as~$P_u$, and the subpath $(w_0,\ldots,w_n)$ by~$P_w$.

We now define $\varepsilon:=\sum_{i=1}^n |1-a_i|<1/5$, $\beta:=n/2+2\varepsilon$ and $\beta':=\beta+1 > \beta$, and the length function~$\ell$ on the cycle edges by setting $\ell(u_{i-1},u_i):=a_i$ and $\ell(w_{i-1},w_i)=2-a_i$ for all $i\in [n]$, and by $\ell(u_n,v_1)=\ell(v_2,w_1):=\varepsilon$, $\ell(v_1,v_2):=\beta'$, and $\ell(w_n,u_0):=\beta'+2\varepsilon$ (see Figure~\ref{fig:cycle-hard}).

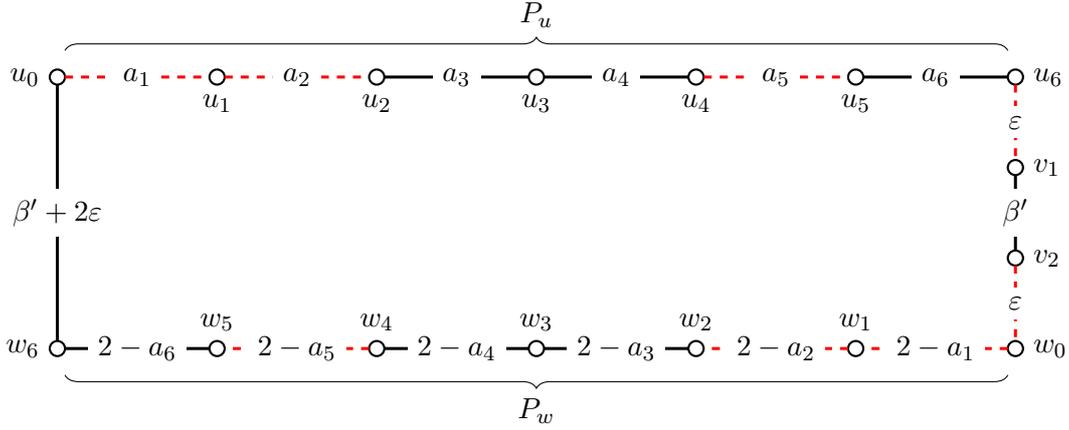
\begin{figure}[ht]
\centering
\input{Figures/cycle}
\caption{Reduction from \CPartition{} to \Contraction{} for $n=6$ discussed in the proof of Theorem~\ref{thm:cycle-hard}.
The dashed edges in the figure represent the set~$C(I)$ for $I=\{1,2,5\}$.}
\label{fig:cycle-hard}
\end{figure}

Now consider the instance $\cJ := (C_{2n+4}, \ell, \varphi)$ with $\varphi(x) = x-\beta$ of the problem \Contraction{}.
Observe that no $\varphi$-contraction may contain an edge $\{u,v\}$ of length greater than~$\beta$ (in particular, no feasible solution may contain one of the edges of length~$\beta'$ or $\beta'+2\varepsilon$).
Furthermore any (weak) $\varphi$-contraction $C$ on this graph satisfies $\Phi(C) = |C|$.

We will show that~$\cJ$ has an optimal solution of cardinality (and thus of value)~$n+2$ if and only if~$\cI$ is a `Yes'-instance.
In particular, we will see that any feasible solution of~$\cJ$ of size~$n+2$ contains the two edges of length~$\varepsilon$ and exactly~$n/2$ edges with length~$a_i$,~$i\in I$, from~$P_u$ and the corresponding edges with length~$2-a_i$,~$i\in I$, from~$P_w$.
Such solutions correspond to subsets of~$[n]$ in the following natural way:
For any subset $I\subseteq [n]$ of size~$n/2$ we let~$C(I)$ be the subset of edges of the cycle~$C_{2n+4}$ consisting of the two edges of length~$\varepsilon$ and of all edges $\{u_{i-1},u_i\}$ and $\{w_{i-1},w_i\}$ (of length~$a_i$ or~$2-a_i$, respectively) for all $i\in I$.
Thus we will show that~$C(I)$ is an optimal solution of the instance~$\cJ$ of \Contraction{} if and only if $\sum_{i \in I} a_i = \sum_{i\in[n]\setminus I} a_i = n/2$, i.e.,~$\cI$ is a `Yes'-instance of \CPartition{}.

Both directions of this equivalence are captured and proved as Claim~2 and 4 below.
Claims~1 and 3 are auxiliary statements used in the proofs of these two main claims.

For any path~$P$ on the cycle we let~$\ell(P)$ denote the sum of~$\ell(e)$ over all edges~$e$ of~$P$.
For all $i\in [n]$ we denote by $P_i^\sqsupset$ and $P_i^\sqsubset$ the path on the cycle between the vertices~$u_i$ and~$w_i$ that contains and that does not contain the edge $\{v_1,v_2\}$, respectively (in Figure~\ref{fig:cycle-hard}, these are the right and left segment of the cycle).

\underline{Claim~1:}
For all~$i\in [n]$, the number $\ell(P_i^\sqsupset)$ lies in the interval $[n+\beta'+\varepsilon,n+\beta'+2\varepsilon]$ and the number $\ell(P_i^\sqsubset)$ lies in the interval $[n+\beta'+2\varepsilon,n+\beta'+3\varepsilon]$.
In particular, we have $\dist_\ell(u_i,w_i)=\min\{\ell(P_i^\sqsupset),\ell(P_i^\sqsubset)\}=\ell(P_i^\sqsupset)$ and the difference $\ell(P_i^\sqsubset)-\ell(P_i^\sqsupset)$ lies in the interval $[0,2\varepsilon]$.

\underline{Proof of Claim~1:}
Note that the condition $\sum_{i=1}^n a_i=n$ implies that
\begin{equation}
\label{eq:ai-split}
\varepsilon=2\sum_{a_i:a_i\geq 1} (a_i-1)=2\sum_{a_i:a_i<1}(1-a_i) .
\end{equation}
By our assumption $a_1\geq a_2\geq\cdots \geq a_n$, the numbers $\ell(P_i^\sqsupset)$ form a unimodal sequence for $i=0,1,\ldots,n$ that is maximized for $i=0$ and $i=n$, proving that $\ell(P_i^\sqsupset)\leq n+\beta'+2\varepsilon$ (note that $\ell(P_u)=\ell(P_w)=n$).
By \eqref{eq:ai-split} the minimum of this unimodal sequence is at most~$\varepsilon$ smaller than the maximum.
This proves the first part of the claim.
As $\ell(P_i^\sqsupset)+\ell(P_i^\sqsubset)=2(n+\beta'+2\varepsilon)$, we obtain the second part of the claim.
The last part of the claim is an immediate consequence of the first two.
\qedclaim

\underline{Claim~2:}
If $I \subseteq [n]$ is a solution of the instance~$\cI$ of \CPartition{} such that $\sum_{i \in I} a_i = \sum_{i\in[n]\setminus I} a_i = n/2$, then~$C(I)$ is a $(1,\beta)$-contraction.

\underline{Proof of Claim~2:}
It suffices to prove that there is no pair of vertices whose distance decreases by more than~$\beta$ when contracting the edges in~$C(I)$.

We start by verifying this for the pairs $u_i,w_i$ for $i \in [n]$.
We first consider the path $P_i^\sqsupset$ between~$u_i$ and~$w_i$.
Observe that $\sum_{e\in C(I)\cap P_i^\sqsupset}\ell(e)$ lies in the interval $[n/2+\varepsilon,n/2+2\varepsilon]=[\beta-\varepsilon,\beta]$.
Similarly to before, this follows from the observation that by the assumption $a_1\geq a_2\geq\cdots \geq a_n$ those sums form a unimodal sequence for $i=0,1,\ldots,n$ that is maximized for~$i=0$ and~$i=n$, and by using \eqref{eq:ai-split} (recall also that $|I|=n/2$).
Consequently, we have
\begin{equation}
\label{eq:ell-Piright}
\ell_{C(I)}(P_i^\sqsupset)\geq \ell(P_i^\sqsupset)-\beta .
\end{equation}
Since $\sum_{e\in C(I)} \ell(e)=n+2\varepsilon=2\beta-2\varepsilon$, we obtain that $\sum_{e\in C(I)\cap P_i^\sqsubset}\ell(e)$ lies in the interval $[\beta-2\varepsilon,\beta-\varepsilon]$, yielding
\begin{equation}
\label{eq:ell-Pileft}
\ell_{C(I)}(P_i^\sqsubset)\geq \ell(P_i^\sqsubset)-(\beta-\varepsilon)\geq \ell(P_i^\sqsubset)-\beta .
\end{equation}
Combining \eqref{eq:ell-Piright} and \eqref{eq:ell-Pileft} proves that
\begin{equation}
\label{eq:dist-ui-wi}
\dist_{\ell_{C(I)}}(u_i,w_i) \geq \dist_\ell(u_i,w_i)-\beta .
\end{equation}
Now consider two vertices~$u_i$ and~$w_j$,~$j<i$ (the case~$j>i$ can be treated analogously).
Let $P_{i,j}^\sqsupset$ and $P_{i,j}^\sqsubset$ be the path on the cycle between the vertices~$u_i$ and~$w_j$ that contains and that does not contain the edge $\{v_1,v_2\}$, respectively.
Using that $P_{i,j}^\sqsupset\subseteq P_i^\sqsupset$ we obtain
\begin{equation}
\label{eq:ell-Pijsmile}
\ell_{C(I)}(P_{i,j}^\sqsupset)\geq \ell(P_{i,j}^\sqsupset)-\beta
\end{equation}
from \eqref{eq:ell-Piright}.

We know that $a_i\leq 1+1/5\leq 8/5$ and consequently
\begin{equation}
\label{eq:2mai}
2-a_i\geq 2/5\geq 2\varepsilon
\end{equation}
by the assumption that the input $\{a_1,a_2,\ldots,a_n\}$ of the instance~$\cI$ is close to 1 (there is plenty of leeway in all those inequalities).
Furthermore, we have
\begin{equation}
\label{eq:ell-Pij}
\dist_\ell(u_i,w_j)\leq \ell(P_i^\sqsupset)-(2-a_i)\leBy{eq:2mai} \ell(P_i^\sqsupset)-2\varepsilon \leq \min\{\ell(P_i^\sqsupset),\ell(P_i^\sqsubset)\}=\dist_\ell(u_i,w_i) ,
\end{equation}
where the second-to-last inequality follows from Claim~1.

Combining those observations yields
\begin{equation}
\label{eq:ell-Pijfrown}
\ell_{C(I)}(P_{i,j}^\sqsubset) \geq \dist_{\ell_{C(I)}}(u_i,w_i) \geBy{eq:dist-ui-wi} \dist_\ell(u_i,w_i)-\beta \geBy{eq:ell-Pij} \dist_\ell(u_i,w_j)-\beta .
\end{equation}
Combining \eqref{eq:ell-Pijsmile} and \eqref{eq:ell-Pijfrown} proves that
\begin{equation}
\label{eq:dist-ui-wj}
\dist_{\ell_{C(I)}}(u_i,w_j) \geq \dist_\ell(u_i,w_j)-\beta .
\end{equation}

From \eqref{eq:ell-Pijsmile} and \eqref{eq:ell-Pijfrown} we can derive analogous relations for the remaining cases where we need to consider the distance between a vertex~$u_i$,~$i\in[n]$, and a vertex $w\in\{v_1,v_2,u_0,u_1,\allowbreak \ldots, u_{i-1},u_{i+1},\ldots,u_n\}$, between a vertex $w_i$, $i\in[n]$, and a vertex $u\in\{v_1,v_2,w_0,w_1,\ldots,\allowbreak w_{i-1},w_{i+1},\ldots,w_n\}$, and between the vertices~$v_1$ and~$v_2$.
This completes the proof of Claim~2.
\qedclaim

\underline{Claim~3:}
Every $(1,\beta)$-contraction~$C$ contains at most~$n/2$ edges in $(P_u \cup P_w) \cap P_i^\sqsupset$ for all $i\in [n]$ and at most~$n/2$ edges in $(P_u \cup P_w) \cap P_i^\sqsubset$ for all $i\in [n]$.

\underline{Proof of Claim~3:}
Note that for any $I\subseteq [n]$ and $k\in\{0,1,\ldots,n\}$ we have $\sum_{i\in I:i>k} a_i+\sum_{i\in I:i\leq k}(2-a_i)\geq |I|-\varepsilon$ by the definition of $\varepsilon$.
Consequently, assuming for the sake of contradiction that~$C$ contains strictly more than~$n/2$ edges in $(P_u \cup P_w) \cap P_i^\sqsupset$, we have $\ell_C(P_i^\sqsupset)-\ell(P_i^\sqsupset)\geq n/2+1-\varepsilon$.
Similarly, assuming that~$C$ contains strictly more than~$n/2$ edges in $(P_u \cup P_w) \cap P_i^\sqsubset$ yields $\ell_C(P_i^\sqsubset)-\ell(P_i^\sqsubset)\geq n/2+1-\varepsilon$.
By Claim~1 the difference $\ell(P_i^\sqsubset)-\ell(P_i^\sqsupset)$ lies in the interval $[0,2\varepsilon]$, so in both cases we obtain
\begin{equation*}
\dist_\ell(u_i,w_i) - \dist_{\ell_C}(u_i,w_i) \geq \frac{n}{2}+1-\varepsilon -2\varepsilon > \frac{n}{2}+2\varepsilon = \beta ,
\end{equation*}
where we used that $\varepsilon<1/5$ in the second-to-last step.
This contradicts the fact that~$C$ is a $(1,\beta)$-contraction, proving Claim~3.
\qedclaim

\underline{Claim~4:}
Let~$C$ be a feasible solution of the instance~$\cJ$ of \Contraction{}.
Then we have $|C|\leq n+2$, and if $|C|=n+2$, we have $C=C(I)$ for some set $I \subseteq [n]$ with $\sum_{i \in I} a_i = \sum_{i\in[n]\setminus I} a_i = n/2$.

\underline{Proof of Claim~4:}
As~$C$ does not contain any of the edges of length~$\beta'$ or $\beta'+2\varepsilon$, we have $|C|\leq n+2$ by Claim~3 (the +2 comes from the two edges of length~$\varepsilon$ that may be contained in~$C$).
Suppose now that $|C|=n+2$.
Applying Claim~3 again shows that~$C$ must contain both edges of length~$\varepsilon$, and that it contains the edge $\{u_{i-1},u_i\}$ if and only if it contains the edge $\{w_{i-1},w_i\}$, for all $i\in[n]$.
Defining $I:=\{i \in [n]: \{u_{i-1},u_i\} \in C\}$ we have $|I|=n/2$ and $C=C(I)$.

By Claim~1 we have $\dist_\ell(u_0,w_0)=\ell(P_0^\sqsupset)$ and $\dist_\ell(u_n,w_n)=\ell(P_n^\sqsupset)$.
As~$C$ is a $(1,\beta)$-contraction containing the two edges of length~$\varepsilon$ we thus obtain $\sum_{i\in I} a_i = \sum_{e\in C\cap P_u} \ell(e) \leq \beta-2\varepsilon = n/2$.
Similarly, we have $\sum_{i\in [n]\setminus I} a_i = \sum_{i\in I} (2-a_i)\allowbreak = \sum_{e\in C\cap P_w} \ell(e)\allowbreak \leq \beta-2\varepsilon = n/2$.
As $\sum_{i\in[n]}a_i=n$, these two inequalities must be tight, yielding $\sum_{i\in I} a_i = \sum_{i\in[n]\setminus I} a_i = n/2$.
\qedclaim

Combining Claims~2 and 4 proves the statement of the theorem for the problem \Contraction{}.

We now focus on the problem \WContraction{}.
The hardness result follows immediately from the following claim.

\underline{Claim~5:}
For $n \geq 5$, any feasible weak $(1,\beta)$-contraction $C$ on the instance $\cJ$ is also a feasible $(1,\beta)$-contraction.

\underline{Proof of Claim~5:}
Suppose for the sake of contradiction that $C$ is not a feasible $(1,\beta)$-contraction.
This means there are vertices $a,b$ such that $\dist_{\ell_C}(a,b) = 0$ and $\dist_\ell(a,b)>\beta$, i.e., $a$ and $b$ lie on a (maximal) subpath $Q$ formed by edges from $C$ on the cycle.
Let $u$ be one end vertex of $Q$, and let $x$ be the neighbour of $u$ not on $Q$.
Let $v$ be the last vertex on $Q$ when traversed starting at $u$, such that the length of the $x$-$v$-path $P$ containing $u$ is at most $\beta+\ell(x, u)$, and let $y$ be the next vertex on $Q$ when traversed starting at $u$.
Such a vertex $y$ exists as $\ell(Q) > \beta$, and the $x$-$y$-path $P'$ containing $u$ has length strictly greater than $\beta+\ell(x, u)$.

We have $\dist_{\ell_C}(x, y) > 0$, as $C$ does not contract the entire cycle.
By~\eqref{eq:contr-cond}, we have $\dist_{\ell_C}(x, y) \geq \dist_\ell(x, y)-\beta$.
As $\dist_{\ell_C}(x, y) \leq \ell(x, u)$, we get $\dist_\ell(x, y) \leq \beta+\ell(x, u)$.
As we saw before, the $x$-$y$-path $P'$ has length strictly greater than $\beta+\ell(x, u)$, thus the $x$-$y$-path $P''$ not containing $u$ must have length at most $\beta+\ell(x, u)$.
As the entire cycle has length $2n+2\beta'+4\varepsilon = 3n+2+8\varepsilon$ and can be partitioned into $P, P''$ and the edge $\{v,y\}$, we get
\begin{align*}
%\begin{equation*}
3n+2+8\varepsilon &= \ell(P)+\ell(P'') + \ell(v, y) \\
&\leq 2(\beta + \ell(x, u)) + \ell(v, y) %\\ 
%&
\leq 5\beta+3+4\varepsilon = 5n/2+3+14\varepsilon,
\end{align*}
%\end{equation*}
where the second inequality holds as the two longest edges of the cycle have length $\beta'+2\varepsilon=\beta+1+2\varepsilon$ and $\beta'=\beta+1$, respectively.
From this chain of inequalities we obtain $n \leq 2+12\varepsilon < 4+2/5$, contradicting the assumption $n\geq 5$.
\end{proof}

The reader might be tempted to `simplify' the previous reduction proof by omitting the four special edges of length $\varepsilon$, $\beta'$ and $\beta'+2\varepsilon$ and by setting $\beta:=n/2$ instead.
However, this would invalidate Claim~2 (specifically, the estimate \eqref{eq:ell-Pileft} would not always hold).

\subsection{Inapproximability of \Contraction{}}
\label{sec:inapx-contraction}

We are able to extend the before-mentioned hardness result for \Contraction{} as follows:

\begin{thm}
\label{thm:inapx-contr}
For any fixed $\beta>0$ and $\varepsilon > 0$, it is NP-hard to approximate the problem \Contraction{} with tolerance function $\varphi(x)=x-\beta$, $\beta\geq 0$, to within a factor of $n^{1-\varepsilon}$.
\end{thm}

For the following theorem the additive error is fixed to $\beta=1$.

\begin{thm}
\label{thm:inapx-bip}
For any $\varepsilon > 0$, it is NP-hard to approximate the problem \Contraction{} with tolerance function $\varphi(x)=x-1$ on bipartite graphs with unit length edges $\ell=1$ to within a factor of $m^{1/2-\varepsilon}$.
\end{thm}

Our reductions are based on the inapproximability of the well-known \Clique{} problem.
Recall that a \emph{clique} in a graph $G$ is a complete subgraph of $G$.

\begin{problem}{\Clique{}}
Input: & A graph $G$. \\
Output: & A clique in $G$ of maximum size. \\
\end{problem}

It was shown in \cite{MR2403018} that for any $\varepsilon > 0$, it is NP-hard to approximate \Clique{} to within a factor of $n^{1-\varepsilon}$.

The following lemma will be used in our proofs.
It shows that for $(1,\beta)$-contractions the feasibility condition \eqref{eq:contr-cond} needs not be checked for all pairs of vertices $u$ and $v$, but only for those satisfying certain extra conditions.

\begin{lem}
\label{lem:add-contr-cond}
A set of edges $C\subseteq E$ is a $(1,\beta)$-contraction if and only if all pairs of vertices $u,v\in V$ with the property that every shortest path with respect to $\ell_C$ between $u$ and $v$ starts and ends with an edge from $C$ satisfy condition \eqref{eq:contr-cond}.
\end{lem}

\begin{proof}
Suppose for the sake of contradiction that all pairs of vertices $u,v\in V$ as in the lemma satisfy condition \eqref{eq:contr-cond} and that $C$ is \emph{not} a $(1,\beta)$-contraction.
Then there is a pair of vertices $u,v \in V$ violating \eqref{eq:contr-cond} and a shortest path $P$ with respect to $\ell_C$ between $u$ and $v$ that does not start or end with an edge from $C$.
We choose $u$ and $v$ such that $\dist_{\ell_C}(u,v)$ is minimal, and we may assume that the first edge $\{u,w\}$ of $P$ is not contained in $C$, so $\dist_\ell(u,v) - \dist_{\ell_C}(u,v) = \dist_\ell(w,v) - \dist_{\ell_C}(w,v)$.
By our choice of $u$ and $v$, the vertices $w$ and $v$ satisfy \eqref{eq:contr-cond}, i.e., the right-hand side of this equation is bounded by $\beta$, a contradiction.
\end{proof}

\begin{proof}[Proof of Theorem~\ref{thm:inapx-contr}]
Let $\beta,\varepsilon>0$ be fixed and let $G=(V,E)$ be an instance of \Clique{}.

We define a graph $H=H(G)$ as follows, see Figure~\ref{fig:inapx-reduc1}:
The vertex set of $H$ is given by $(V\times\{1,2\})\cup \{s\}$, i.e., we create two copies of each original vertex and add a special vertex $s$.
The edge set of $H$ is given by $\{\{(u,1),(v,1)\}:\{u,v\}\in E\}$ plus the edges $\{(v,1),(v,2)\}$ and $\{s,(v,2)\}$ for all $v\in V$.
The first set of edges are simply the original edges of $G$ on the first copies of the vertices, the second set is a perfect matching between the two copies of the vertex set, and the third set of edges connects the special vertex~$s$ to all vertices of the second copy of the vertex set.
The length function $\ell$ on the edges of $H$ is set to $2\beta+2$, $\beta$ or $\beta+1$ for those three sets of edges, respectively.

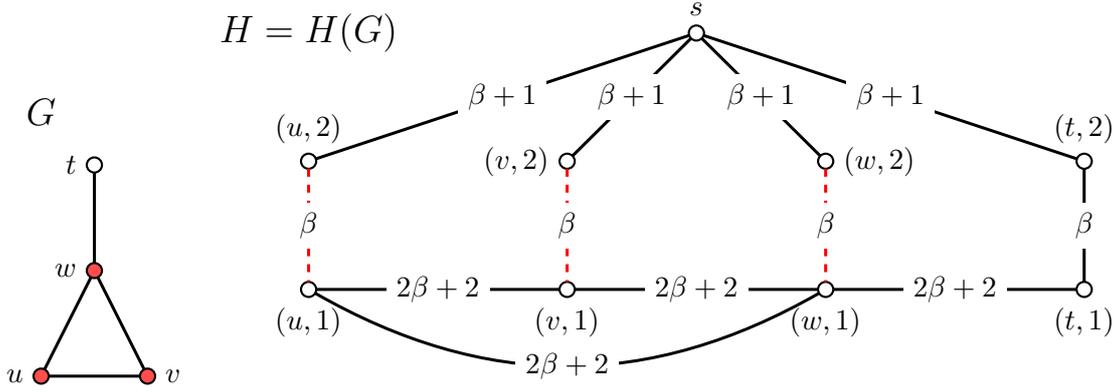
\begin{figure}[ht]
\centering
\input{Figures/lollipop_g_1}
\input{Figures/lollipop_h_1}
\caption{An instance $G$ of \Clique{} (left) and the corresponding instance $\cI=(H,\ell,\varphi)$ (right) of \Contraction{} constructed in the proof of Theorem~\ref{thm:inapx-contr}.
The dashed edges form the set $C(U)$ for $U=\{u,v,w\}$.}
\label{fig:inapx-reduc1}
\end{figure}

Now consider the instance $\cI:=(H,\ell,\varphi)$ of the problem \Contraction{} with the tolerance function $\varphi(x)=x-\beta$.
Clearly, any $(1,\beta)$-contraction $C$ in $H$ can contain only edges of the form  $\{(u,1),(u,2)\}$ for some $u\in V$.
As $H$ does not contain two edges between two different connected components of $(V,C)$, our objective function defined in \eqref{eq:objective} satisfies $\Phi(C)=|C|$ for any feasible solution $C$ of $\cI$.
We will show that it allows a feasible solution with $k$ edges (and thus of value $k$) if and only if $G$ has a clique with $k$ vertices.
Formally, for $U \subseteq V$ we define $C(U) := \{\{(u,1),(u,2)\} : u \in U\}$ (see Figure~\ref{fig:inapx-reduc1}).
We proceed to show that $U$ induces a clique in $G$ if and only if $C(U)$ is a $(1,\beta)$-contraction in $H=H(G)$.

Note that for any two vertices $u,v \in U$ we have
\begin{align*}
\dist_{\ell_{C(U)}}((u,1),(v,1)) &= 2\beta+2 =
\begin{cases}
\dist_\ell((u,1),(v,1)) & \text{if } \{u,v\} \in E,\\
\dist_\ell((u,1),(v,1))-2\beta & \text{otherwise},
\end{cases} \\
\dist_{\ell_{C(U)}}((u,1),(v,2)) &= 2\beta+2 = \dist_\ell((u,1),(v,2)) - \beta, \\
\dist_{\ell_{C(U)}}((u,2),(v,2)) &= 2\beta+2 = \dist_\ell((u,2),(v,2)) , \\
\dist_{\ell_{C(U)}}((u,1),(u,2)) &= 0 = \dist_\ell((u,1),(u,2)) - \beta.
\end{align*}
These relations together with Lemma~\ref{lem:add-contr-cond} show that $C(U)$ is a $(1,\beta)$-contraction in $H$ if and only if $U$ is a clique in $G$.

As $n(H)$ differs from $n(G)$ only by a constant factor, an $n^{1-\varepsilon}$-approximation algorithm for \Contraction{} would yield an $n^{1-\varepsilon'}$-approximation algorithm for \Clique{} via this reduction.
Together with the before-mentioned inapproximability of \Clique{} \cite{MR2403018} this proves the theorem.
\end{proof}

The rest of this section is devoted to proving Theorem~\ref{thm:inapx-bip}, so we now focus on $(1,1)$-contractions in bipartite graphs with unit length edges $\ell=1$.
The next lemma characterizes the structure of contractions in this setting.

\begin{lem}
\label{lem:bip-contr}
Let $G=(V, E)$ be a bipartite graph with unit edge lengths $\ell=1$ and let $C \subseteq E$ be a set of edges.

\begin{enumerate}[label=(\roman*),leftmargin=7mm]
\item
If $C$ is a $(1,1)$-contraction, then $C$ is a matching.

\item
If $C = \{e,f\}$ with edges $e = \{u_1,u_2\}, f = \{v_1,v_2\} \in E$, then $C$ is a $(1,1)$-contraction if and only if $\dist_\ell(u_1,v_1) = \dist_\ell(u_2,v_2)$ and $\dist_\ell(u_1,v_2) = \dist_\ell(u_2,v_1)$.

\item
$C$ is a $(1,1)$-contraction if and only if all two-element subsets of $C$ are.
\end{enumerate}
\end{lem}

\begin{proof}
\begin{enumerate}[label=(\roman*),leftmargin=7mm]
\item
Suppose for the sake of contradiction that $C$ contains a path $(u,v,w)$ on two edges.
As $G$ is bipartite, it has no triangles, so $\dist_\ell(u,w)=2$ and $\dist_{\ell_C}(u,w)=0$, a contradiction to the assumption that $C$ is a $(1,1)$-contraction.

\item
For the edges $e=\{u_1,u_2\}$ and $f=\{v_1,v_2\}$ we define $d_{i,j}:=\dist_\ell(u_i,v_j)$ for $i,j\in\{1,2\}$.

Let $C=\{e,f\}$ be a $(1,1)$-contraction.
Both $d_{1,1}$ and $d_{2,2}$ must have the same parity (as $G$ is bipartite), so if $d_{1,1}<d_{2,2}$, the difference between them is exactly~2.
However, this would mean that $\dist_{\ell_C}(u_2,v_2)=d_{1,1}=d_{2,2}-2=\dist_\ell(u_2,v_2)-2$, a contradiction to the assumption that $C$ is a $(1,1)$-contraction.
Repeating the same argument with $d_{1,1}$ and $d_{2,2}$ interchanged shows that $d_{1,1}=d_{2,2}$.
An analogous argument shows that $d_{1,2}=d_{2,1}$.

Now suppose that $d_{1,1}=d_{2,2}$ and $d_{1,2}=d_{2,1}$.
From these conditions it follows that for all $i,j\in\{1,2\}$ every path between $u_i$ and $v_j$ that contains both edges $e$ and $f$ has length at least $d_{i,j}+2$ with respect to $\ell$.
Consequently, we have $\dist_{\ell_C}(u_i,v_j)\geq \dist_\ell(u_i,v_j)-1$ for $C=\{e,f\}$.
By Lemma~\ref{lem:add-contr-cond}, $C$ is a $(1,1)$-contraction.

\item
One direction of the equivalence is obvious, so we only need to prove the other direction.
So we assume that all two-element subsets of $C$ are $(1,1)$-contractions, and we need to prove that $C$ is a $(1,1)$-contraction.
The argument is a straightforward generalization of the argument for (ii) from before.
Let $P$ be a path that contains exactly $k$ edges from $C$, and that starts and ends with an edge from $C$.
Let $e_1,e_2,\ldots,e_k$ be those edges and $u_{1,1},u_{1,2},u_{2,1},u_{2,2},\ldots,u_{k,1},u_{k,2}$ their end vertices as they are encountered when traversing $P$ (so $u_{1,1}$ and $u_{k,2}$ are the end vertices of $P$).
For all $i=1,2,\ldots,\lfloor k/2\rfloor$ the pair of edges $e_{2i-1}$ and $e_{2i}$ and their end vertices satisfy the distance conditions from~(ii).
From these conditions it follows that the subpath of $P$ between $u_{2i-1,1}$ and $u_{2i,2}$ has length at least $\dist_\ell(u_{2i-1,1},u_{2i,2})+2$.
So overall the length of $P$ is at least $\dist_\ell(u_{1,1},u_{k,2})+2\lfloor k/2\rfloor\geq \dist_\ell(u_{1,1},u_{k,2})+(k-1)$.
Consequently, we have $\dist_{\ell_C}(u_{1,1},u_{k,2})\geq \dist_\ell(u_{1,1},u_{k,2})-1$.
By Lemma~\ref{lem:add-contr-cond}, $C$ is a $(1,1)$-contraction.
\end{enumerate}
\end{proof}

With Lemma~\ref{lem:bip-contr} in hand, we are now ready to prove Theorem~\ref{thm:inapx-bip}.

\begin{proof}[Proof of Theorem~\ref{thm:inapx-bip}]
Let $\varepsilon>0$ be fixed and let $G=(V,E)$ be an instance of \Clique{}.
We construct a bipartite graph $H=H(G)$ as follows, see Figure~\ref{fig:inapx-reduc-bip}:
For every vertex $v \in V$, the graph $H$ contains two vertices $(v,1)$ and $(v,2)$ and the edge $f_v:=\{(v,1),(v,2)\}$.
For every edge $e=\{u,v\} \in E$, we add a vertex $x_e$ and the edges $f_{e,u}:=\{x_e,(u,1)\}$ and $f_{e,v}:=\{x_e,(v,1)\}$ to $H$.
Furthermore, we add a new special vertex~$s$ to $H$ and all the edges $\{s,(v,2)\}$, $v\in V$, and $\{s,x_e\}$, $e\in E$.
It is easy to check that the graph~$H$ defined in this way is bipartite.

All edges of $H$ receive unit lengths ($\ell = 1$) and we consider the instance $\cI=(H,\ell,\varphi)$ of the problem \Contraction{} with the tolerance function $\varphi(x)=x-1$.

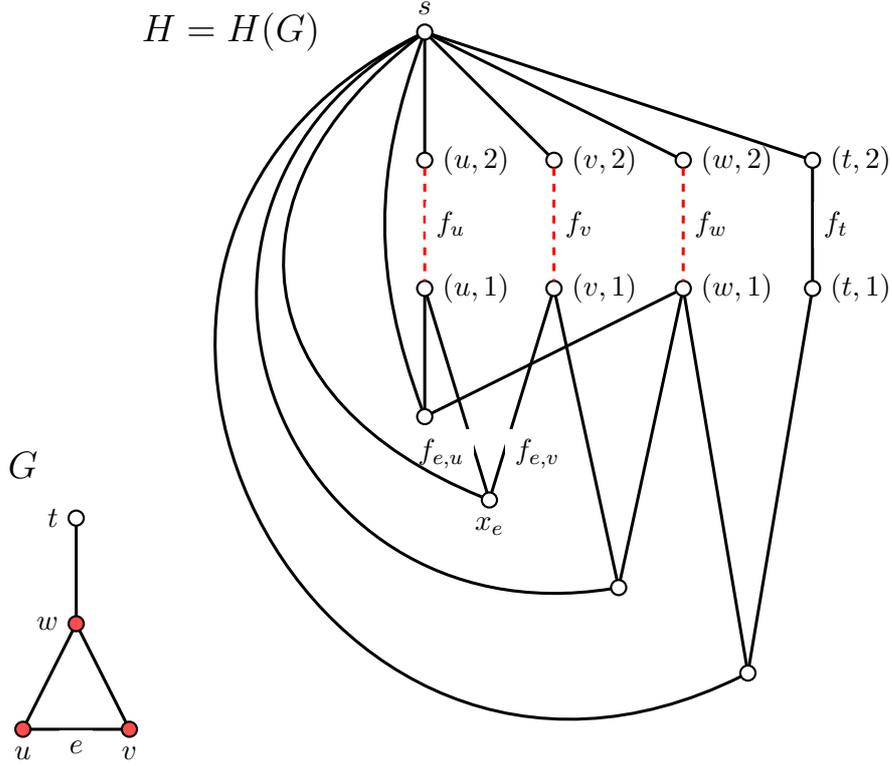
\begin{figure}[ht]
\centering
\input{Figures/lollipop_g_2}
\input{Figures/lollipop_h_2}
\caption{An instance $G$ of \Clique{} (left) and the corresponding instance $\cI=(H,\ell,\varphi)$ (right) of \Contraction{} constructed in the proof of Theorem~\ref{thm:inapx-bip}.
The dashed edges form the set $C(U)$ for $U=\{u,v,w\}$.}
\label{fig:inapx-reduc-bip}
\end{figure}

For any set of vertices $U \subseteq V$ we define $C(U) := \{f_u : u \in U\}$ (see Figure~\ref{fig:inapx-reduc-bip}).

\underline{Claim~1:}
If $U \subseteq V$ is a clique in $G$, then $C(U)$ is a $(1,1)$-contraction in $H$ and $\Phi(C(U)) = |U|$.

\underline{Proof of Claim~1:}
Let $U$ be a a set of vertices in $G$ that form a clique, and let $u,v \in U$ be two vertices from this clique.
Then we have $\dist_\ell((u,1),(v,1))=\dist_\ell((u,2),(v,2))=2$ and $\dist_\ell((u,1),(v,2))=\dist_\ell((u,2),(v,1))=3$, so Lemma~\ref{lem:bip-contr}~(ii) implies that $C(\{u,v\})$ is a $(1,1)$-contraction in $H$.
Repeating this argument for every pair of vertices from $U$ and applying Lemma~\ref{lem:bip-contr}~(iii) yields that $C(U)$ is a $(1,1)$-contraction in $H$.
As there are never two edges in $H$ between any two connected components of the graph $(V,C(U))$, we have $\Phi(C(U)) = |C(U)| = |U|$.
\qedclaim

For any set of edges $C \subseteq E(H)$, we let $U(C)$ be the set of vertices $v \in V$ for which $(v,1)$ is incident to an edge in $C$.

\underline{Claim~2:}
If $C\subseteq E(H)$ is a $(1,1)$-contraction, then $C$ is a matching in $H$ and $U(C)$ is a clique in $G$ of size at least $\Phi(C)-3$.

\underline{Proof of Claim~2:}
$C$ is a matching by Lemma~\ref{lem:bip-contr}~(i).

Let $u,v \in U(C)$.
We will show that $e=\{u,v\}\in E$ by applying Lemma~\ref{lem:bip-contr}~(ii) to the two edges in $C$ incident to $(u,1)$ and $(v,1)$.
To prove that $e\in E$ it suffices to show that $\dist_\ell((u,1),(v,1))=2$.

Let us first consider the case that $f_u,f_v \in C$.
As $\dist_\ell((u,2),(v,2)) = 2$ (the shortest path between those vertices goes via~$s$), Lemma~\ref{lem:bip-contr}~(ii) implies that $\dist_\ell((u,1),(v,1))=2$.
We now consider the case that there is an edge $e'\in E\setminus\{e\}$ with $f_u, f_{e',v} \in C$.
We then have $\dist_\ell((u,2),x_{e'}) = 2$ (via~$s$), so Lemma~\ref{lem:bip-contr}~(ii) yields $\dist_\ell((u,1),(v,1))=2$.
Finally, we consider the case that there are two edges $e',e''\in E\setminus \{e\}$ with $f_{e',u}, f_{e'',v} \in C$.
We then have $\dist_\ell(x_{e'},x_{e''}) = 2$ (via~$s$), again implying that $\dist_\ell((u,1),(v,1))=2$.
This proves that indeed $e\in E$, so $U(C)$ forms a clique in $G$.

Every edge in $H$ is either incident to $s$ or to a vertex of the form $(v,1)$, $v\in V$.
Since at most one of the edges incident to $s$ can be in $C$, the definition of $U(C)$ shows that the size of $U(C)$ is either $|C|-1$ or $|C|$.
Therefore, to finish the proof of Claim 2, it suffices to show that $\Phi(C) \leq |C|+2$.
If $C$ contains no two edges that are connected by more than one edge in $H$, then we have $\Phi(C)=|C|$.
Otherwise we consider two such edges $f$ and $g$ from $C$.
It is easy to check that either $f$ or $g$ must be incident to $s$, so suppose that the edge $f$ contains $s$.
We first consider the case that $f=\{s,x_e\}$ for some edge $e=\{u,v\}\in E$.
In this case it follows that $g=\{(u,1),(u,2)\}$ or $g=\{(v,1),(v,2)\}$, so we have $\Phi(C)=|C|+2$.
Now consider the case that $f=\{s,(u,2)\}$ for some vertex $u \in V$.
In this case it follows that $g=\{(u,1),x_e\}$ for exactly one edge $e \in E$ incident to $u$ in $G$, showing that $\Phi(C)=|C|+2$.
In all three cases we have $\Phi(C)\leq |C|+2$, as claimed.
\qedclaim

Combining Claims~1 and 2 will allow us to prove the following claim:

\underline{Claim~3:}
If there is an $n^{1/2-\varepsilon}$-approximation algorithm for \Contraction{}, then there is an $n^{1-\varepsilon/2}$-approximation algorithm for \Clique{}.

\underline{Proof of Claim~3:}
Suppose for the sake of contradiction that such an approximation algorithm for \Contraction{} exists.
We use it to compute a clique in a given instance $G$ of \Clique{} as follows:
We construct $\cI=(H(G),\ell,\varphi)$ and compute a solution $C$ of \Contraction{} for this instance, and we define the clique $U(C)$ as before (recall Claim~2).
If $U(C)\neq \emptyset$, we return $U(C)$, otherwise we return any vertex from $G$.
We denote the clique computed in this fashion by $U$.

We may assume that $n(G) \geq 16^{1/\varepsilon}$, in particular $n(H) \geq 16^{1/\varepsilon}$.
It follows that
\begin{equation}
\label{eq:nHnG}
n(H) = 1+2n(G)+m(G) \leq 1+2n(G)+\binom{n(G)}{2} \leq n(G)^2.
\end{equation}
By assumption we know that 
\begin{equation}
\label{eq:PhiC*}
\Phi(C)\cdot n(H)^{1/2-\varepsilon} \geq \Phi(C^*),
\end{equation}
where $C^*$ is an optimal solution of $\cI$.
In particular, $\Phi(C)$ is positive.

Combining these observations we get
\begin{align*}
|U|\cdot n(G)^{1-\varepsilon/2} &\geBy{eq:nHnG} |U|\cdot n(H)^{1/2-\varepsilon/2} \\
&\geq \max\{\Phi(C)-3,1\} \cdot n(H)^{1/2-\varepsilon/2} \\
&= \big(\max\{\Phi(C)-3,1\}\cdot n(H)^{\varepsilon/2}\big) \cdot n(H)^{1/2-\varepsilon} \\
&\geq \Phi(C) \cdot n(H)^{1/2-\varepsilon}  \\
&\geBy{eq:PhiC*} \Phi(C^*)\\
&\geq \omega(G),
\end{align*}
where the second inequality holds because of Claim~2, and the last inequality involving the clique number $\omega(G)$ holds because of Claim~1.
\qedclaim

As $m(H)=\Theta(n(H))$, Claim~3 implies the theorem (using the inapproximability of \Clique{} proved in \cite{MR2403018}).
\end{proof}

\section{Hardness for multiplicative tolerance function}
\label{sec:hard-mult}

By Theorem~\ref{thm:cycle-hard-fix}, the problem \WContraction{} with purely additive tolerance function $\varphi(x)=x-\beta$ is NP-hard on cycles.
In this section we prove the hardness and inapproximability of this problem also in the case of a purely multiplicative tolerance function $\varphi(x)=x/\alpha$, $\alpha\geq 1$.
Recall that the problem \Contraction{} is trivial for this tolerance function (we may not contract any edges).

\subsection{Hardness of planar \WContraction{}}
\label{sec:planar}

To state the main result of this section recall that the \emph{girth} of a graph $G$ is defined as the minimum length of a cycle in $G$.

\begin{thm}
\label{thm:weak-planar}
For any $g\geq 2$, the problem \WContraction{} with tolerance function $\varphi(x)=x/2$, is NP-hard for planar graphs with girth at least $3g$ and unit length edges $\ell=1$.
\end{thm}

Theorem~\ref{thm:weak-planar} implies that \WContraction{} is hard for a general multiplicative tolerance function $\varphi(x)=x/\alpha$, $\alpha\geq 1$, but it leaves open the question whether this is true also for other fixed values of $\alpha$ other than 2 (when $\alpha$ is not part of the input).
The arguments given in this section for $\alpha=2$ carry over straightforwardly to any fixed value $2\leq \alpha<3$, but not to 3 or larger values (for $\alpha<2$ and unit length edges the problem is trivial).

We first characterize the set of feasible solutions in this special case.

\begin{lem}
\label{lem:weak-planar}
Let $G=(V,E)$ be a graph with girth at least 6 and unit length edges $\ell=1$, and consider the tolerance function $\varphi(x)=x/2$.
Furthermore, let $C\subseteq E$ be a set of edges such that $(V,C)$ is disconnected.
Then $C$ is a weak $(2,0)$-contraction if and only if for any two edges $e,f\in C$ either $e$ and $f$ are incident and both contain a degree-1 vertex, or any path containing $e$ and $f$ also contains at least two edges not in $C$.
\end{lem}

Recall that the assumption that $(V,C)$ is disconnected prevents solutions $C\subseteq E$ for which the contracted graph $G/C$ is a single vertex.
Note that Lemma~\ref{lem:weak-planar} does not require $G$ to be planar.

\begin{proof}
To prove the equivalence, we need the following auxiliary claim:

\underline{Claim:}
If $C$ is a weak $(2,0)$-contraction, then every component of $(V,C)$ that is not a single edge is a star with the property that each of its vertices except the center of the star has degree~1 in $G$.

\underline{Proof of Claim:}
Let $M$ be a component of $(V,C)$ with more than one edge.
Clearly, there must be an edge $\{u,v\}$ with vertices $u\notin V(M)$ and $v\in V(M)$.
If $M$ contains a path $P$ on two edges starting at $v$ and ending at some vertex $w$, then $\dist_\ell(u,w)=3$ and $\dist_{\ell_C}(u,w)=1$, a contradiction to the assumption that $C$ is a weak $(2,0)$-contraction (note that $P\cup \{u,v\}$ is the shortest path between $u$ and $w$, as the girth of $G$ is at least 6).
Thus the edges of $M$ must form a star centered at $v$.
By the same argument, no vertex outside $M$ can be connected to any vertex of $M$ other than $v$.
This proves the claim.
\qedclaim

We first assume that $C$ is a weak $(2,0)$-contraction, and we need to show that any two edges $e,f\in C$ satisfy the conditions of the lemma.
If $e$ and $f$ are incident, the statement follows from the auxiliary claim from before.
If $e$ and $f$ are not incident, we consider an inclusion-minimal path $P$ containing both $e$ and $f$.
We let $u$ and $v$ be the end vertices of $P$, $u'$ the other end vertex of $e$, and $v'$ the other end vertex of $f$ ($u'$ and $v'$ are the vertices at distance 1 from the ends of the path).
If the distance between $u'$ and $v'$ was only 1, we have $\dist_\ell(u,v)=3$ and $\dist_{\ell_C}(u,v)=1$ (here we need again the assumption that the girth is at least 6), a contradiction to the assumption that $C$ is a weak $(2,0)$-contraction.
Therefore at least two edges lie between $u'$ and $v'$.
The auxiliary claim from before implies that no two incident edges on $P$ between $u$ and $v$ are contained in $C$, therefore $P$ must contain at least two edges not in $C$.
This proves one direction of the equivalence.

To prove the other direction, we now assume that any two edges $e,f$ satisfy the conditions of the lemma, and we need to show that $C$ is a weak $(2,0)$-contraction.
Consider any two vertices $u$ and $v$ with $\dist_{\ell_C}(u,v)>0$, and any path between $u$ and $v$.
As no inner vertex of $P$ is a leaf, we know that between any two consecutive edges from $C$ on $P$ there are at least 2 edges not in $C$.
This proves that $\dist_{\ell_C}(u,v)\geq \dist_\ell(u,v)/2$, as desired.

This completes the proof of the lemma.
\end{proof}

For a given propositional formula $F$ in conjunctive normal form (CNF) the bipartite \emph{variable-clause graph} $\Gamma(F)$ is defined as follows:
The two partition classes of $\Gamma(F)$ are given by the sets of variables and clauses of $F$, and there is an edge between a variable $x$ and a clause $c$ if $x$ appears in $c$.
If $c$ contains $x$ as a positive or negative literal, we call the corresponding edge of $\Gamma(F)$ a \emph{positive or negative edge}, respectively.
A planar drawing of $\Gamma(F)$, where positive and negative edges appear in cyclically contiguous intervals around every variable vertex, is called \emph{contiguous}.

We call a $k$-CNF formula \emph{regular}, if every clause contains exactly $k$ literals, no clause contains a literal twice, every variable appears at least once as a positive literal and at least once as a negative literal in the formula.

Consider now the following variant of \textsc{3SAT}.

\begin{problem}{\CPSAT{}}
Input: & A regular 3-CNF formula $F$ and a contiguous planar drawing of $\Gamma(F)$. \\
Output: & `Yes', if $F$ has a satisfying assignment, `No' otherwise. \\
\end{problem}

\begin{lem}
\label{lem:cpsat}
\CPSAT{} is NP-complete.
\end{lem}

\begin{proof}
The more general variant of \CPSAT{} not requiring $F$ to be regular was shown to be NP-complete in \cite{DBLP:journals/ijcga/BergK12}.
We now show how to reduce this generalization to \CPSAT{}, which will prove the lemma.
Given a (not necessarily regular) 3-CNF formula $F$ we first eliminate all variables appearing only as negative or only as positive literals and all clauses containing exactly one literal, as well as multiple appearances of literals in the same clause.
This yields a formula $F'$ in which all clauses have two or three literals, no clause contains a literal twice, and every variable appears at least once as a positive literal and at least once as a negative literal in $F'$.
Moreover, since $\Gamma(F')$ is a subgraph of $\Gamma(F)$, we also obtain a contiguous planar drawing of $\Gamma(F')$.
As a last step we eliminate clauses $c$ with two literals by introducing a new variable $x$ for each of them and replacing $c$ by the equivalent formula $(c \lor x) \land (c \lor \ol{x})$.
It is easy to check that the resulting formula $F''$ is regular and equisatisfiable to $F$, and to obtain a contiguous planar drawing of $\Gamma(F'')$, see Figure~\ref{fig:exact-3}.

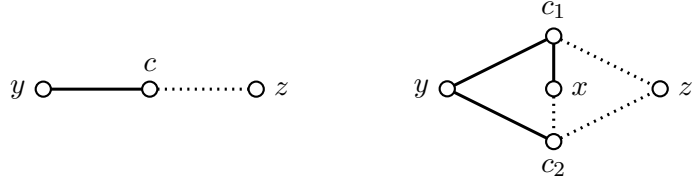
\begin{figure}[ht]
\centering
\input{Figures/clause_1}
\hspace{1cm}
\input{Figures/clause_2}
\caption{The clause $c = (y \lor \ol{z})$ is replaced by $c_1 \land c_2 = (c \lor x) \land (c \lor \ol{x})$, preserving the contiguous planar drawing of the variable-clause graph.
The positive and negative edges are drawn as solid and dotted edges, respectively.}
\label{fig:exact-3}
\end{figure}
\end{proof}

\begin{proof}[Proof of Theorem~\ref{thm:weak-planar}]
We first present the proof for the case $g=2$, and then sketch how to generalize it for larger values of $g$.

We reduce \CPSAT{} to \WContraction{}.
Consider an instance $F$ of \CPSAT{} with variables $x_1,x_2, \ldots, x_n$ and clauses $c_1,c_2,\linebreak \ldots, c_m$.

Given the formula $F$, we construct from it a graph $G=G(F)$ as follows, see Figures~\ref{fig:gadgets} and~\ref{fig:sat}.
For every variable $x_i$, $i\in[n]$, we add a variable gadget $H(x_i)$ as shown on the left hand side of Figure~\ref{fig:gadgets} to the graph $G$.
The vertices $u_i$ and $\ol{u}_i$ will be used later to connect this gadget to other parts of the graph.
The idea of the variable gadget is that an optimal solution of our instance of \WContraction{} should contain either the four edges $T_i:=\{t_i,\ol{t}_i,t_i',t_i''\}$ or the four edges $F_i:=\{f_i,\ol{f}_i,f_i',f_i''\}$, corresponding to setting $x_i$ to \ttrue{} or \tfalse{}, respectively.

For every clause $c_j$, $j\in [m]$, we add a clause gadget $H(c_j)$ (a star with three edges) as shown on the right hand side of Figure~\ref{fig:gadgets} to the graph $G$.
The vertices $v_j^1$, $v_j^2$ and $v_j^3$ will be used later to connect this gadget to other parts of the graph.
The idea of the clause gadget is that a feasible solution contains at most one of these three edges, and if it does contain one of them, this restricts the choice we have inside the respective neighbouring variable gadget.

\begin{figure}[ht]
\centering
\input{Figures/var_gadget}
\input{Figures/clause_gadget}
\caption{Variable gadget (left) and clause gadget (right) used in the proof of Theorem~\ref{thm:weak-planar} for $g=2$.
The vertices used for connecting those gadgets to the rest of the graph are marked in black.}
\label{fig:gadgets}
\end{figure}
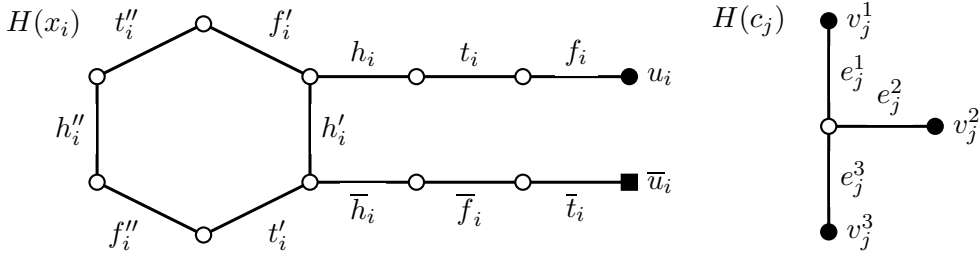

We connect the variable and clause gadgets in $G$ as follows (see Figure~\ref{fig:sat}):
For every $j\in[m]$ and $k\in[3]$, if the $k$-th literal in the clause $c_j$ is $x_i$, we add an edge connecting $u_i$ to $v_j^k$, and if the $k$-th literal in the clause $c_j$ is $\ol{x_i}$, we add an edge connecting $\ol{u}_i$ to $v_j^k$.
We refer to the edges added to $G$ in this step as \emph{connection edges}.

\begin{figure}[ht]
\centering
\input{Figures/var_clause_graph} \\
\input{Figures/sat}
\caption{The graphs $\Gamma(F)$ (top) and $G=G(F)$ (bottom) constructed in the proof of Theorem~\ref{thm:weak-planar} for $g=2$ and the formula $F=c_1\land c_2\land c_3\land c_4\land c_5$ with clauses $c_1 := (\ol{x_1} \lor x_3\lor x_4)$, $c_2 := (\ol{x_1} \lor x_3\lor \ol{x_4})$, $c_3 := (\ol{x_1} \lor x_2 \lor \ol{x_3})$, $c_4 := (x_1 \lor \ol{x_2}\lor x_5)$ and $c_5 := (x_1 \lor \ol{x_2}\lor \ol{x_5})$.
The positive and negative edges of $\Gamma(F)$ are drawn as solid and dotted edges, respectively.
The set of dashed edges in $G(F)$ is the set $C(\tau)$ for the variable assignment $\tau$ that sets all variables to \ttrue{}.}
\label{fig:sat}
\end{figure}
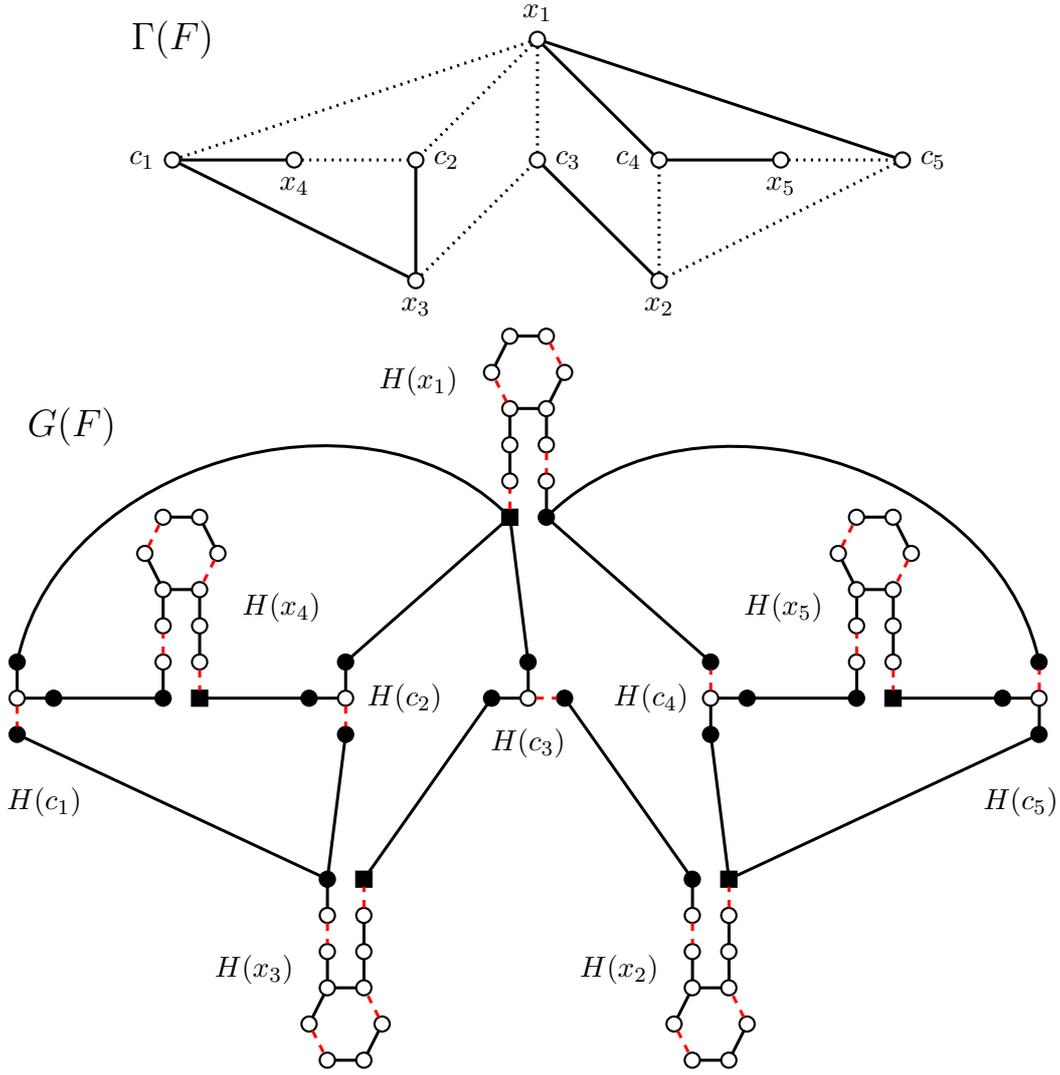

This completes the definition of the graph $G=G(F)$.
It is easy to see that this graph is planar.
Specifically, a planar embedding can be obtained from the given planar embedding of $\Gamma(F)$ by replacing variable vertices $x_i$ in $\Gamma(F)$ by the variable gadgets $H(x_i)$ in $G$, and by replacing clause vertices $c_j$ by the clause gadgets $H(c_j)$.
Using that for each variable vertex $x_i$ in $\Gamma(F)$ the positive and negative edges appear in cyclically contiguous intervals around $x_i$, the connection edges in $G$ (that connect the variable and clause gadgets) can also be drawn in a planar fashion.

Moreover, it is easy to check that $G$ has girth~6 and no degree-1 vertices.

Now consider the instance $\cI:=(G,\ell,\varphi)$ of the problem \WContraction{} with $\ell=1$ (unit length edges) and the tolerance function $\varphi(x)=x/2$.

Lemma~\ref{lem:weak-planar} implies that any feasible solution of $\cI$ is a matching, as $G$ has no vertices of degree 1. As $G$ contains no cycles of length 3 or 4, it cannot contain two edges between vertex sets of two different components of $(V, C)$ for any such feasible solution $C$.
This implies that our objective function satisfies $\Phi(C)=|C|$.

We proceed to show that $F$ is satisfiable if and only if $\cI$ has an optimal solution of cardinality (and thus of value) $4n+m$.
Specifically, a satisfying assignment of $F$ corresponds to a solution that contains exactly all edges of either $T_i$ or $F_i$ in $H(x_i)$ for every variable $i\in[n]$ (corresponding to the value \ttrue{} or \tfalse{} assigned to this variable, respectively) and exactly one edge in $H(c_j)$ for each clause $j\in[m]$ (corresponding to a literal that satisfies this clause).

Formally, for any variable assignment $\tau\colon \{x_1,x_2,\ldots,x_n\}\to\{\ttrue,\tfalse\}$, we define the set of edges $C(\tau) \subseteq E(G)$ as follows:
$C(\tau)$ contains all edges of $T_i$ for any variable $x_i$, $i\in[n]$, that $\tau$ sets to \ttrue{}, and it contains all edges of $F_i$ for any variable $x_i$ that $\tau$ sets to \tfalse{}.
Moreover, for every clause $c_j$, $j\in[m]$, that is satisfied by $\tau$, we choose an index $k\in [3]$ of a literal in $c_j$ that is satisfied by $\tau$ and add the edge $e_j^k$ to $C(\tau)$.

The following claim is an immediate consequence of Lemma~\ref{lem:weak-planar}.

\underline{Claim~1:}
Any subset $C\subseteq E(G)$ is a feasible solution if and only if every path containing two edges from $C$ also contains at least two edges not in $C$.

By Claim~1, for every variable assignment $\tau$ of $F$, the set $C(\tau)$ is a feasible solution of $\cI$.
In particular, if $\tau$ satisfies $F$, then $C(\tau)$ is a feasible solution of size $4n+m$.
The remainder of the proof is devoted to showing the converse, i.e., if $C\subseteq E(G)$ is a feasible solution of size $4n+m$, then $F$ is satisfiable.

For all $i\in[n]$ we let $H(x_i)^+$ denote the subgraph of $G$ induced by all edges of $H(x_i)$ and all connection edges incident to either $u_i$ or $\ol{u}_i$.

\underline{Claim~2:}
For any $j\in[m]$, $C$ contains at most one edge from $H(c_j)$.
For any $i\in[n]$, $C$ contains at most four edges from $H(x_i)^+$.
Moreover, if $C$ contains one of the connection edges incident to $u_i$ or $\ol{u}_i$ for some $i\in[n]$, it does not contain any edges from the gadget $H(c_j)$ that is connected to $H(x_i)$ via this edge.

\underline{Proof of Claim~2:}
The first and last statement are immediate consequences of Claim~1.
The argument for the second statement is as follows:
For all $i\in[n]$ we let $E_i$ denote the set of edges $\{h_i,t_i,f_i\}$ plus the connection edges incident to $u_i$, and we let $\ol{E}_i$ denote the set of edges $\{\ol{h}_i,\ol{t}_i,\ol{f}_i\}$ plus the connection edges incident to $\ol{u}_i$.
By Claim~1, $C$ contains at most two edges from $E_i$, and if the intersection size is two, then $C$ must contain the edge $h_i$.
Similarly, $C$ contains at most two edges from $\ol{E}_i$, and if the intersection size is two, then $C$ must contain the edge $\ol{h}_i$.
As $C$ cannot contain $h_i$ and $\ol{h}_i$ simultaneously, $C$ contains at most three edges from $E_i\cup\ol{E}_i$, and if the intersection size is three, then $C$ must contain either $h_i$ or $\ol{h}_i$.
Again by Claim~1, $C$ contains at most two edges from the 6-cycle $\{f_i',h_i',t_i',f_i'',h_i'',t_i''\}$.
However, if $C$ contains one of the edges $h_i$ or $\ol{h}_i$, it contains at most one edge from this 6-cycle.
This proves that $C$ indeed contains at most four edges from $H(x_i)^+$.
\qedclaim

Note that every edge of $G$ belongs to exactly one subgraph $H(x_i)^+$ or $H(c_j)$.
So if $|C|=4n+m$, we know by Claim~2 that $C$ contains exactly four edges from $H(x_i)$ for all $i\in[n]$ and exactly one edge from $H(c_j)$ for all $j\in[m]$, and none of the connection edges in $G$.

\underline{Claim~3:}
For any $i\in[n]$, if $C$ contains four edges from $H(x_i)$ and if $f_i$ is not among them, then those edges must be $T_i$.
On the other hand, if $\ol{t}_i$ is not among them, those edges must be $F_i$.
In particular, these two cases cannot occur simultaneously.

\underline{Proof of Claim~3:}
If $C$ contains four edges from $H(x_i)$ and $f_i$ is not among them, Claim~1 enforces taking first the edge $t_i$, then $t_i'$ and $t_i''$, and eventually $\ol{t}_i$.
This proves the first part of the statement.
The argument for the second part is symmetric.
The third part of the statement is a consequence of the first two.
\qedclaim

So given a solution $C$ of $\cI$ of size $4n+m$, we can derive from it a satisfying assignment $\tau$ of $F$ as follows:
For every clause $c_j$, $j\in[m]$, we consider the unique edge $e_j^k$ from $H(c_j)$ that belongs to $C$.
We follow the attachment edge incident to $e_j^k$, leading to the corresponding variable gadget $H(x_i)$, and connecting to either $u_i$ or $\ol{u}_i$.
If the attachment edge connects to $u_i$, then by Claim~1, $f_i\notin C$, so by Claim~3, the four edges of $H(x_i)$ contained in $C$ must be $T_i$, so we define $\tau(x_i):=\ttrue{}$.
If the attachment edge connects to $\ol{u}_i$, then by Claim~1, $\ol{t}_i\notin C$, so by Claim~3, the four edges of $H(x_i)$ contained in $C$ must be $F_i$, so we define $\tau(x_i):=\tfalse{}$.
This process does not lead to any contradicting variable assignments by the last statement of Claim~3.
However, this process may leave some variables $x_i$ undefined, and we can set them arbitrarily, e.g., $\tau(x_i):=\ttrue{}$.
By construction, each clause receives a satisfying literal, so the assignment $\tau$ is indeed a satisfying assignment of $F$.

This proves that $F$ is satisfiable if and only if $\cI$ has a feasible solution of size $4n+m$ (which must be optimal by Claim~2), completing the proof of the theorem in the case $g=2$.

For values $g\geq 2$, the construction of the gadgets $H(x_i)$ and $H(c_j)$ can be generalized as follows:
We subdivide each of the edges $h_i,\ol{h}_i$ and $h_i''$, and each of the edges $e_j^1$, $e_j^2$ and $e_j^3$ into $1+3(g-2)$ edges.
Then the resulting graph $G=G(F)$ clearly has girth $3g$, and the above arguments can be easily modified to show that any solution $C$ of $\cI$ contains at most $1+3(g-2)=3g-5$ edges from $H(c_j)$ for all $j\in[m]$, and at most $4+3(g-2)=3g-2$ edges from $H(x_i)^+$ for all $i\in[n]$, and that $F$ is satisfiable if and only if $\cI$ has an optimal solution of size $(3g-2)n+(3g-5)m$.
This completes the proof.
\end{proof}

\subsection{Inapproximability of \WContraction{}}
\label{sec:inapx-wcontraction}

We are able to further extend our hardness results for \WContraction{} as follows:

\begin{thm}
\label{thm:inapx-weak}
For any $\varepsilon > 0$, it is NP-hard to approximate the problem \WContraction{} with tolerance function $\varphi(x)=2x/3$ to within a factor of $n^{1-\varepsilon}$.
\end{thm}

Theorem~\ref{thm:inapx-weak} implies that \WContraction{} is hard to approximate for general multiplicative tolerance functions $\varphi(x)=x/\alpha$, $\alpha\geq 1$, but it leaves open the question whether this is true also for other fixed values of $\alpha$ other than $3/2$ (when $\alpha$ is not part of the input).
The arguments given in this section for $\alpha=3/2$ carry over straightforwardly to any fixed value $1<\alpha<2$, but not to 2 or larger values (for $\alpha=1$ the problem is trivial).

This time we reduce from the well-known \IndSet{} problem (which is equivalent to \Clique{} by considering the complement graph).
Recall that an \emph{independent set} in a graph $G$ is a subset of vertices of $G$ such that no two vertices in the subset are adjacent.

\begin{problem}{\IndSet{}}
Input: & A graph~$G$. \\
Output: & An independent set in $G$ of maximum size. \\
\end{problem}

We use again the fact that for any $\varepsilon > 0$, \IndSet{} is NP-hard to approximate to within a factor of $n^{1-\varepsilon}$ \cite{MR2403018}.

\begin{proof}[Proof of Theorem~\ref{thm:inapx-weak}]
Let $G=(V,E)$ be an instance of \IndSet{}.
We construct a graph $H=H(G)$ and a length function $\ell$ on the edges of $H$ as follows, see Figure~\ref{fig:indset}:
We start with a copy of $G$, and all edges of this copy receive length~2.
The vertices of this copy are also denoted by $V$.
We then add additional vertices and edges to $H$ as follows:
To every vertex $v\in V$ we attach two pending edges $\{v,(v,1)\}$ and $\{v,(v,2)\}$ of length~1 or 2, respectively.
We may assume $G$, and thus also $H$, to be connected.

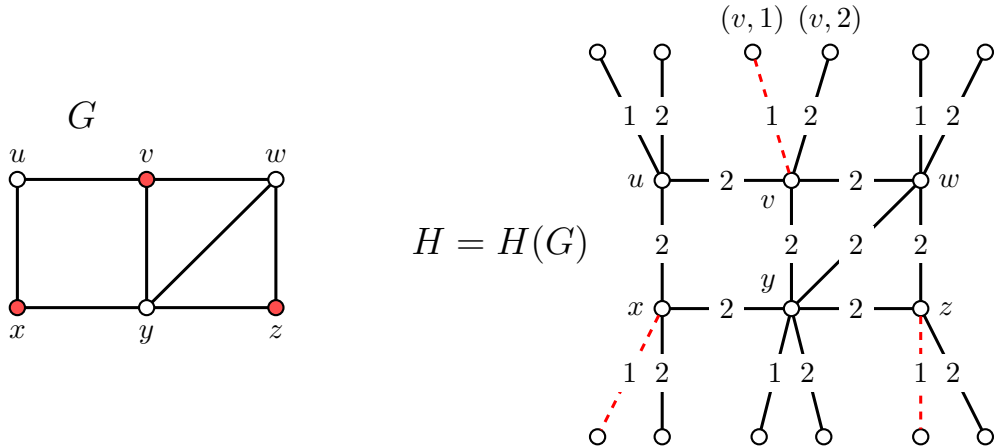
\begin{figure}[ht]
\centering
\input{Figures/indset_g}
\hspace{1cm}
\input{Figures/indset_h}
\caption{An instance $G$ of \IndSet{} (left) and the corresponding instance $H=H(G)$ of \WContraction{} (right) constructed in the proof of Theorem~\ref{thm:inapx-weak}.
The dashed edges are the set $C(U)$ for the independent set $U=\{v,x,z\}$.}
\label{fig:indset}
\end{figure}

Now consider the instance $\cI=(H,\ell,\varphi)$ of the problem \WContraction{} with the tolerance function $\varphi(x)=2x/3$.

We proceed to show that $\cI$ has a feasible solution of value $k$ if and only if $G$ has an independent set of size $k$.
This is an immediate consequence of Claim~3 below.
To prove Claim~3 we need the following two auxiliary claims.

\underline{Claim~1:}
For any induced subgraph of $H$ that is a path on two edges, a feasible solution $C$ of $\cI$ does not contain only the longer of the two edges (either it contains none of the two, the shorter of the two if there is one, or both).

\underline{Proof of Claim~1:}
Consider a path on two edges $\{u,v\}$, $\{v,w\}$ of length~2 in $H$ such that $\{u,w\} \notin E(H)$, and suppose for the sake of contradiction that $\{u,v\}\in C$, but $\{v,w\}\notin C$.
Then we have $\dist_\ell(u,w)=4$ and $\dist_{\ell_C}(u,w)=2$, violating the condition \eqref{eq:contr-cond} for the given tolerance function.
A similar contradiction arises if one of the edges has length~1 and the other length~2, and only the edge of length~2 is contracted.
This proves the claim.
\qedclaim

\underline{Claim~2:}
No feasible solution of $\cI$ contains an edge of length~2.

\underline{Proof of Claim~2:}
Assume for the sake of contradiction that a feasible solution $C$ contains an edge $e$ of length~2. Note that any edge $f$ of $H$ may be reached from $e$ via a walk $e_1\dots e_k$ where $e_1 = e$ and $e_k = f$, and for all $i<k-1$  we have $\ell(e_i) = 2$ and the edges $e_i$ and $e_{i+1}$ induce a path in $H$. Now successively applying Claim~1 to the subgraphs induced by $e_i$ and $e_{i+1}$ for $i < k$ shows that $C$ contracts $f$. Thus $C$ violates the condition that a weak contraction must not contract every edge.
\qedclaim

Claim~2 implies that our objective functions satisfies $\Phi(C)=|C|$ for every feasible solution $C$ of $\cI$, because $H$ never contains two edges between two different connected components of $(V,C)$.

For any set of vertices $U \subseteq V(G)$ we define $C(U) := \{\{u,(u,1)\}: u \in U\}$.

\underline{Claim~3:}
A set of edges $C\subseteq E(H)$ is a feasible solution of $\cI$ if and only if $C=C(U)$ for an independent set $U$ in $G$.

\underline{Proof of Claim~3:}
Let $C$ be a feasible solution of $\cI$.
By Claim~2, $C$ contains only edges of length~1, so we have $C=C(U)$ for some set of vertices $U$ in $G$.
Suppose that two such vertices $u,v\in U$ are connected by an edge, then we would have $\dist_\ell(u,v)=4$ and $\dist_{\ell_C}(u,v)=2$, violating the condition \eqref{eq:contr-cond} for the given tolerance function.
It follows that $U$ is an independent set.

To prove the other direction of the equivalence, let $U$ be an independent set in $G$ and consider the set of edges $C(U)$ in $H$.
To verify that $C$ is a weak $(3/2,0)$-contraction, it suffices to check condition \eqref{eq:contr-cond} between the end vertices of paths on two edges, one of length~1 from $C$ and the other of length~2, and for paths on $k$ edges that start and end with an edge of length~1 from $C$.
In the first case the contraction $C(U)$ changes the distance from 3 to 2, which is compatible with \eqref{eq:contr-cond}.
In the second case the contraction $C(U)$ changes the distance from $2k-2$ to $2k-4$, which is also compatible with \eqref{eq:contr-cond}, where we use that $k\geq 4$ because of the assumption that $U$ is an independent set.
\qedclaim

Claim~3 implies that $\cI$ has a feasible solution with $k$ edges if and only if $G$ has an independent set of size $k$.
As $n(H)=3n(G)=\cO(n(G))$, the theorem follows from the \cite{MR2403018} result.
\end{proof}

\section{Asymptotic bounds}
\label{sec:asymp}

In this section we show how to compute contractions for graphs that are not optimal, but can be computed efficiently despite our hardness results from the previous section.
In this vein, the main results of this section are Theorem~\ref{thm:asymp-mult} and the corresponding (not tight) lower bound (Theorem~\ref{thm:asymp-girth}) for the case of tolerance functions of the form $\varphi(x)=x/\alpha-1$.
Further we consider purely additive tolerance functions (Section~\ref{sec:asymp-add}) and the factor by which a contraction can reduce the number of vertices (Section~\ref{sec:asymp-vert}).
Throughout this section, we assume all graphs to have unit length edges $\ell=1$.

\subsection{Almost multiplicative contractions}
\label{sec:asymp-mult}

As mentioned in the introduction, a purely multiplicative tolerance function ($\beta=0$) forbids decreasing any distances.
In this section we thus consider an `almost' purely multiplicative tolerance function of the form $\varphi(x)=x/\alpha-1$.

\begin{thm}
\label{thm:asymp-mult}
Let $k\geq 1$ be a real number.
Any graph $G$ has a $(2k-1,1)$-contraction $C$ such that the contracted graph $G/C$ has at most $n^{1+1/k}$ edges, and such a contraction can be computed in time $\cO(m)$.
\end{thm}

Recall that here and throughout, $n$ and $m$ denote the number of vertices and edges of the input graph $G$, not of the contracted graph $G/C$.
Setting $k:=\log_2 n$ in Theorem~\ref{thm:asymp-mult} yields the following corollary.

\begin{cor}
\label{cor:asymp-mult}
Any graph $G$ has a $(2\log_2 n-1,1)$-contraction $C$ such that the contracted graph $G/C$ has at most $2n$ edges, and such a contraction can be computed in time $\cO(m)$.
\end{cor}

To prove Theorem~\ref{thm:asymp-mult}, we use a clustering approach as presented in \cite{MR810338}, yielding the next lemma.
Specifically, the following crucial lemma appears in a slightly weaker form in that paper.
For any real number $r\geq 1$, we define an \emph{$r$-partition} of a graph $G=(V,E)$ as a set of \emph{clusters} $P_i \subseteq V$, $i\in[l]$, with corresponding \emph{cluster centers $p_i \in P_i$}, where the sets $P_i$ are required to form a partition of the vertex set $V$ and where $\dist_\ell(p_i,u)\leq r-1$ for all $u\in P_i$ and $i\in[l]$.
We denote the resulting $r$-partition by $P:=\{(p_i,P_i) : i\in[l]\}$.
We write $\rho(P)$ for the number of pairs $1\leq i<j\leq l$ for which $P_i$ and $P_j$ are connected by at least one edge, and we refer to this quantity as the \emph{density of $P$}.

\begin{lem}
\label{lem:awerbuch}
Let $r\geq 1$ be a real number.
Any graph $G$ with unit length edges has an $r$-partition $P$ with density $\rho(P)\leq n^{1+1/r}$, and such a partition can be computed in time $\cO(m)$.
\end{lem}

\begin{proof}
The idea of the algorithm is to build an $r$-partition $P$ of $G$ iteratively in rounds.
In each round, we build a new cluster and remove all vertices from that cluster from the graph, processing the subgraph on the remaining vertices in the next round.
The algorithm proceeds until all vertices are assigned to a cluster.
In round $i$, we choose an arbitrary vertex $p_i$ as a cluster center, and define layers $L_{i,0},L_{i,1},\ldots$ around the vertex $p_i$, where the layer $L_{i,j}$ consists of all vertices at distance exactly $j$ from $p_i$ (this distance is measured in the subgraph of $G$ under consideration in this round).
We continue computing these layers as long as the number of vertices in the new layer is at least the number of vertices in all previous layers times the factor $n^{1/r}$.
The cluster $P_i$ is defined as the union of all layers around $p_i$ satisfying this expansion condition.
We refer to the first layer violating this condition (which is \emph{not} added to $P_i$ anymore) as the \emph{rejected layer}.
We let $P$ denote the partition of the vertices of $G$ computed in this fashion.

To verify that $P$ is indeed an $r$-partition, we proceed to show that each vertex within a cluster has distance at most $r-1$ from the center vertex of that cluster, and that the density $\rho(P)$ of the partition is at most $n^{1+1/r}$.
Intuitively, the expansion condition in the definition of the layers ensures that a cluster has few layers and that the number of edges that go to unclustered vertices is small.

Consider a cluster $P_i$ with center vertex $p_i$ and the layers $L_{i,0},L_{i,1},\ldots,L_{i,d}$.
Suppose for the sake of contradiction that $d\geq r$.
By the definition of the layers in the algorithm we know that $|L_{i,j}|\geq n^{1/r}\sum_{k=0}^{j-1} |L_{i,k}|$ holds for all $j\in [d]$, implying that $|L_{i,j}|\geq n^{j/r}$.
Consequently, the size of the cluster satisfies $|P_i|=\sum_{j=0}^d |L_{i,j}|\geq 1+n^{r/r}=n+1$, a contradiction.

We now show that $\rho(P)\leq n^{1+1/r}$.
The key idea is that the number of vertices in the rejected layer of a cluster $P_i$ is at most $n^{1/r}|P_i|$.
Thus the number of edges from $P_i$ to clusters that are created later is at most $n^{1/r}|P_i|$.
For every edge between two clusters we let the cluster that is created first account for that edge.
Summing over all these edges between clusters yields the desired upper bound of $\rho(P)\leq n\cdot n^{1/r}=n^{1+1/r}$.

Using breadth-first search, the partitioning algorithm described above runs in time $\cO(m)$ (recall that $G$ is assumed to be connected).
This completes the proof of the lemma.
\end{proof}

With Lemma~\ref{lem:awerbuch} in hand, we are now ready to prove Theorem~\ref{thm:asymp-mult}.

\begin{proof}[Proof of Theorem~\ref{thm:asymp-mult}]
Given $G=(V,E)$, we first compute a $k$-partition $P$ into $l$ clusters as described by Lemma~\ref{lem:awerbuch}.
We define the set $C$ of contracted edges as the union of all edges within the clusters, $C := \{\{u,v\} \in E : u,v \in P_i \text{ for some } i \in [l]\}$.
We thus contract each cluster into a single vertex and remove from every set of resulting parallel edges all but a single edge.

We proceed to show that $C$ is a $(2k-1,1)$-contraction, i.e., we show that $\dist_{\ell_C}(u,v) \geq \dist_{\ell}(u,v)/(2k-1)-1$ for all $u,v\in V$.
Consider two vertices $u \in P_i$ and $v \in P_j$, where $i$ and $j$ might be equal.
Let $Q_{u,v}$ be the shortest path from $u$ to $v$ in $G$ with edge lengths $\ell_C$ (all edges from $C$ receive length zero).
The length $d$ of $Q_{u,v}$ is the number of edges on that path that connect different clusters.
Note that $Q_{u,v}$ enters and leaves each of the $d+1$ visited clusters at most once, using at most $2k-2$ edges in every cluster, so in $G$ (where all edges have unit lengths) we get $\dist_{\ell}(u,v) \leq d + (d+1)(2k-2)$.

Combining these observations we obtain
\begin{equation*}
\dist_{\ell_C}(u,v) = d \geq d - \frac{1}{2k -1} = \frac{d + (d+1)(2k-2)}{2k -1} -1 \geq \frac{\dist_{\ell}(u,v)}{2k -1} -1,
\end{equation*}
proving the claim.
It remains to show that the contracted graph $G/C$ has at most $n^{1+1/k}$ edges, which is an immediate consequence of the upper bound $m(G/C)=\rho(P)\leq n^{1+1/k}$ given by Lemma~\ref{lem:awerbuch}.
This completes the proof of the theorem.
\end{proof}

Erd\H{o}s' girth conjecture \cite{MR0180500} asserts that there exist graphs with $\Omega(n^{1+1/k})$ edges and girth $2k+1$.
It has been verified for $k=1,2,3,5$ \cite{MR1109426} and the strongest spanner lower bounds depend on it.
We derive from the conjecture the following (not tight) lower bound.

\begin{thm}
\label{thm:asymp-girth}
Assuming Erd\H{o}s' girth conjecture, there exists for any integer $k\geq 2$ a graph $G$ such that any $(k-1,1)$-contraction $C$ results in a graph $G/C$ with $\Omega(n^{1+1/k})$ edges.
\end{thm}

\begin{proof}
For a given integer $k\geq 2$ let $G$ be a graph that is guaranteed by Erd\H{o}s' girth conjecture, i.e., $G$ has girth $2k+1$ and $\Omega(n^{1+1/k})$ edges.
Consider any $(k-1,1)$-contraction $C$ on $G$, and consider a connected component of the graph $(V,C)$.
Applying \eqref{eq:contr-cond} shows that $\dist_\ell(u,v)\leq k-1$ holds for any two vertices $u$ and $v$ in that component.
Using that the girth of $G$ is $2k+1$, it follows that for any cycle in $G$, the connected component of $(V,C)$ does not contain a contiguous segment of cycle edges of length at least half of the cycle.
This implies that all connected components of the graph $(V,C)$ are trees with diameter at most $k-1$.
Therefore, the total number of edges within all connected components of $(V,C)$ is at most $n$.
We will further argue that there is at most one edge between any two connected components.
Suppose for the sake of contradiction that there are two components of $(V,C)$ with two different edges connecting them, say $\{u,v\}$ and $\{u',v'\}$, where $u$ and $u^\prime$ lie in the same connected component and $v$ and $v'$ in the other.
As the diameter of each component is at most $k-1$, it follows that in $G$ there is a path from $u$ to $u'$ of length at most $k-1$, and a path from $v$ to $v'$ of length at most $k-1$.
Together with the two edges connecting the components we obtain a cycle of length at most $2(k-1)+2=2k$, contradicting the assumption that $G$ has girth $2k+1$.

Therefore, the resulting graph after the contraction has $\Omega(m)=\Omega(n^{1+1/k})$ edges.
\end{proof}

\subsection{Additive contractions}
\label{sec:asymp-add}

Turning to the case of a purely additive error, we obtain the following two results.

\begin{thm}
\label{thm:asymp-add}
Let $G$ be a graph with unit length edges.
\begin{enumerate}[label=(\roman*),leftmargin=7mm]
\item For any even integer $0 \leq k \leq n$, the set of edges incident to the $k/2$ vertices of highest degrees is a $(1,k)$-contraction $C$ in $G$ with $\Phi(C) \geq km/(2n)$.\label{thm:asymp-add-first}
\item For any real number $0 < k \leq n$, the set of edges incident to two vertices of degree at least $n/k$ is a $(1,k)$-contraction $C$ in $G$ such that $G/C$ has $\cO(n^2/k)$ edges.\label{thm:asymp-add-high}
\end{enumerate}
These contractions can be computed in time $\cO(m)$.
\end{thm}

As mentioned in the introduction, Bernstein and Chechik analyzed the contraction of Theorem~\hyperref[thm:asymp-add-high]{\ref*{thm:asymp-add}~\ref*{thm:asymp-add-high}} in~\cite{MR3536582} and used it in their dynamic shortest paths algorithm, so this part is already proved.

\begin{proof}[Proof of Theorem~{\hyperref[thm:asymp-add-first]{\ref*{thm:asymp-add}~\ref*{thm:asymp-add-first}}}]
Let $U$ be the set of $k$ vertices in $G$ of highest degree.
Then we have
\begin{equation*}
\sum_{u \in U} \deg(u) \geq k/n \sum_{v \in V} \deg(v) = k/n\cdot 2m = 2km/n.   
\end{equation*}
Let $C$ be the set of edges incident to any vertex in $U$.
As each edge is incident to at most two vertices in $U$, we get $|C| \geq 1/2 \sum_{u \in U} \deg(u) \geq km/n$ from the previous inequality.
As no shortest path visits a vertex in $U$ twice, $C$ is indeed a $(1,2k)$-contraction.
The set $C$ can be computed as follows: We first compute the degrees of all vertices in time $\cO(m)$, then find the $k$-th largest element in this list in time $\cO(n)$, and by another linear time sweep over this list we select $k$ vertices of highest degree.
Overall, the required time is $\cO(m)$.
\end{proof}

This result implies that the number of edges in $G/C$ is at most $m-k m/n$.
If $G$ is a path, no $(1,2k)$-contraction has an objective value greater than $2k$, and $k m/n=k(1-1/n)$, showing that the objective value in Theorem~\hyperref[thm:asymp-add-first]{\ref*{thm:asymp-add}~\ref*{thm:asymp-add-first}} can be improved by at most a factor of two.

The information theoretic lower bound in~\cite{MR3536579} implies that for all $\varepsilon > 0$, any contraction $C$ such that $G/C$ has $\cO(n^{4/3 - \varepsilon})$ edges does not admit a constant additive error.

\subsection{Vertex reduction}
\label{sec:asymp-vert}

All of the results above show that contractions can be effectively used to reduce the number of edges in a dense graph.
But one possible advantage of using a contraction instead of a spanner is that it also has the potential to reduce the number of \emph{vertices} in the graph.
Unfortunately, for constant approximation errors, it is not possible to guarantee more than a constant-factor reduction in general graphs: it is not hard to see that given a path on $n$ vertices, any $(k,1)$-contraction will still result in at least $n/(k+1)$ vertices.
The same problem applies to general dense graphs, since they could still contain a long path within them.
That being said, it seems likely that in practice contraction can lead to significant vertex reduction in many dense graphs.
We ground this practical intuition with the following theoretical result for the special case of graphs with large minimum degree.

\begin{thm}
\label{thm:min-degree}
Let $D$ be an integer.
Any graph $G$ with minimum degree at least $D$ has a $(5,1)$-contraction $C$ such that the contracted graph $G/C$ has at most $n/D$ vertices, and such a contraction can be computed in time $\cO(m)$.
\end{thm}

\begin{proof}
Recall the definition of an $r$-partition.
For a cluster $P_i$ with center vertex $p_i$ we refer to $r$ as the \emph{radius} of that cluster.
This is the maximum distance of all cluster vertices from $p_i$.

We will show how to construct a $3$-partition in which the number of clusters $P_i$ is at most $n/D$.
Using the exact same argument as in the proof of Theorem~\ref{thm:asymp-mult}, such a $3$-partition yields the desired $(5,1)$-contraction.
Our construction first builds clusters of radius~1, and then extends them to clusters of radius~2.
The clustering with radius~1 proceeds very similarly as in the proof of Lemma~\ref{lem:awerbuch} before with $r=1$.
The crucial difference is that we choose as center vertices only vertices with degree at least $D$.
If no such vertices are left, the clustering process terminates, and the remaining unclustered vertices have degree strictly less than $D$.
It is easy to see that since those vertices have degree at least $D$ in the original graph, they must be adjacent to a vertex in a radius~1 cluster.
We can thus assign each of those vertices to such a cluster arbitrarily, yielding a clustering of all vertices of $G$ with radius~2.

The number of clusters is at most $n/D$ because by construction every cluster contains at least $D$ vertices.
This shows that the number of vertices in the contracted graph is at most $n/D$.

This algorithm can be implemented in time $\cO(m)$ by using an adjacency list representation where we keep track of degree information after removing an edge from the graph.
\end{proof}

To see that we cannot guarantee less than $n/D$ vertices, even with larger approximation error, consider the graph $G$ that consists of $n/D$ isolated $D$-cliques.
We now show that even if $G$ is connected, we cannot guarantee $o(n/D)$ vertices in the contracted graph, even if we allow a larger (constant) approximation error.

\begin{thm}
\label{thm:asymp-degree}
Let $D$ and $k$ be integers.
There exists an infinite family of $n$-vertex graphs $G$ with minimum degree $D$ such that any $(k,1)$-contraction $C$ results in a graph $G/C$ with $n/((k+1)D)$ vertices.
\end{thm}

\begin{proof}
Assume for simplicity that $n$ is divisible by $D$.
We construct the graph $G$ as follows.
We partition the $n$ vertices into $n/D$ layers, with each layer containing exactly $D$ vertices.
For $1 \leq i < n/D$, all vertices in layer~$i$ receive an edge to all vertices in layer~$i+1$.
Clearly all vertices in the resulting graph have degree at least $D$.
Let $u$ and $v$ be two vertices in layers~$i$ and $j$, respectively.
Then clearly we have $\dist_{\ell}(u,v) \geq |j-i|$.
Now let $C$ be any $(k,1)$-contraction on $G$, and consider the connected components of the graph $(V,C)$.
Applying \eqref{eq:contr-cond} shows that $\dist_\ell(u,v)\leq k$ holds for any two vertices $u$ and $v$ in the same component.
Combining these two inequalities shows that every connected component contains vertices from at most $k+1$ layers.
As there are $n/D$ layers, the contracted graph has at least $n/((k+1)D)$ vertices.
\end{proof}

\section*{Acknowledgements}

We thank Martin Skutella for stimulating discussions about the problems treated in this paper.
We also thank the anonymous referees for their valuable suggestions that helped improving the presentation of results.

\bibliographystyle{alpha}
\bibliography{refs}{}

\end{document}

%% file: Figures/tikz_styles.tex
\usepackage{tikz}
\usetikzlibrary{calc,%
    intersections,%
    decorations.pathreplacing,%
    arrows,%
    petri,%
    topaths}%
     
\tikzset{node1/.style={circle, draw, fill=white, inner sep=0pt, text width=0pt, text height=0pt, text depth=0pt, minimum size = 5pt,}, edge1/.style={draw, thick},
}
\tikzset{node2/.style={shape=circle, draw, fill=white, inner sep=0pt, text width=0pt, text height=0pt, text depth=0pt, minimum size = 3pt}, edge2/.style={},
}

\usepackage{tkz-berge}

\newcommand\vertexsize{2} % adjust vertex diameter (mm)
\newcommand\edgewidth{0.4}  % adjust width of edges (mm)

\tikzset{vertex_default/.style={circle,draw,fill=white,thick, inner sep=0pt,minimum size=\vertexsize mm}}
\tikzset{vertex_bullet/.style={circle,draw,fill=black,thick, inner sep=0pt,minimum size=\vertexsize mm}}
\tikzset{vertex_square/.style={shape=rectangle,draw,fill=black,thick, inner sep=0pt,minimum size=\vertexsize mm}}
\tikzset{vertex_red/.style={circle,draw,fill=red!70,thick, inner sep=0pt,minimum size=\vertexsize mm}}

\tikzset{edge_default/.style={line width=\edgewidth mm}}
\tikzset{edge_in_c/.style={dashed,color=red,line width=\edgewidth mm}}
\tikzset{edge_dotted/.style={dotted,line width=\edgewidth mm}}

\definecolor{c-color}{rgb}{1,0,0}

%% file: Figures/ex_tree1.tex
\begin{tikzpicture}
\node[] (label) at (1.5,1.2) {$T$};
%\path[draw] (-1,0) to (1,0);
%\path[draw] (0,-1) to (0,1);
\path (-2,1.5) rectangle (2,-1.5);

\begin{pgfonlayer}{background}
\path[draw=black!35!white, fill=black!65!white , line width=6pt, line cap=round, line join=round] (0.38959,0.083034) -- (0.14247,-0.082004) -- (0.29681,-0.16921) -- cycle;
\path[draw=black!60!white, fill=black!65!white , line width=6pt, line cap=round, line join=round] (0.30535,-0.51941) -- (0.45764,-0.37445) -- (0.31316,-0.34433) -- cycle;
\end{pgfonlayer}

\node[node2, name path=39, label=center:{}] (39) at (-0.13084,0.31605) {};
\node[node2, name path=38, label=center:{}] (38) at (-0.086614,0.68401) {};

\node[node2, name path=37, label=center:{}] (37) at (0.38959,0.083034) {};
\node[node2, name path=36, label=center:{}] (36) at (-1.8879,1.2603) {};
\node[node2, name path=35, label=center:{}] (35) at (-0.70822,0.090325) {};
\node[node2, name path=34, label=center:{}] (34) at (-0.60014,0.45068) {};
\node[node2, name path=33, label=center:{}] (33) at (-0.54061,0.12515) {};
\node[node2, name path=32, label=center:{}] (32) at (-1.8864,0.56368) {};

\node[node2, name path=31, label=center:{}] (31) at (0.45764,-0.37445) {};

\node[node2, name path=30, label=center:{}] (30) at (0.31316,-0.34433) {};

\node[node2, name path=29, label=center:{}] (29) at (0.30535,-0.51941) {};
\node[node2, name path=28, label=center:{}] (28) at (-0.81116,-0.038312) {};
\node[node2, name path=27, label=center:{}] (27) at (1.5592,0.78609) {};
\node[node2, name path=26, label=center:{}] (26) at (1.3121,0.58556) {};
\node[node2, name path=25, label=center:{}] (25) at (-1.2507,-0.23896) {};
\node[node2, name path=24, label=center:{}] (24) at (-0.075087,0.84544) {};
\node[node2, name path=23, label=center:{}] (23) at (-0.065542,1.0037) {};
\node[node2, name path=22, label=center:{}] (22) at (1.2097,0.46729) {};
\node[node2, name path=21, label=center:{}] (21) at (0.67192,-0.0042096) {};
\node[node2, name path=20, label=center:{}] (20) at (1.4976,-0.063371) {};
\node[node2, name path=19, label=center:{}] (19) at (1.35,-0.013614) {};
\node[node2, name path=18, label=center:{}] (18) at (-1.6879,1.0146) {};
\node[node2, name path=17, label=center:{}] (17) at (-0.92481,-0.15578) {};
\node[node2, name path=16, label=center:{}] (16) at (-1.0695,-0.46611) {};
\node[node2, name path=15, label=center:{}] (15) at (0.98166,0.28823) {};
\node[node2, name path=14, label=center:{}] (14) at (-0.16221,0.15043) {};
\node[node2, name path=13, label=center:{}] (13) at (-0.99005,0.34453) {};
\node[node2, name path=12, label=center:{}] (12) at (0.11133,-1.2652) {};
\node[node2, name path=11, label=center:{}] (11) at (0.19902,-1.1318) {};
\node[node2, name path=10, label=center:{}] (10) at (1.0895,0.041077) {};
\node[node2, name path=9, label=center:{}] (9) at (1.0127,0.15224) {};
\node[node2, name path=8, label=center:{}] (8) at (0.14247,-0.082004) {};

\node[node2, name path=7, label=center:{}] (7) at (0.29681,-0.16921) {};
\node[node2, name path=6, label=center:{}] (6) at (-1.5679,0.69911) {};
\node[node2, name path=5, label=center:{}] (5) at (-1.2674,0.54518) {};
\node[node2, name path=4, label=center:{}] (4) at (0.4392,-1.3616) {};
\node[node2, name path=3, label=center:{}] (3) at (0.31723,-1.024) {};
\node[node2, name path=2, label=center:{}] (2) at (0.39411,-0.87865) {};
\node[node2, name path=1, label=center:{}] (1) at (0.72149,-1.0522) {};
\node[node2, name path=0, label=center:{}] (0) at (0.94124,-1.3137) {};

\path[edge2, draw] (39) to[] node[] {} (38);
\path[edge2, draw] (39) to[] node[] {} (14);
\path[edge2, draw] (38) to[] node[] {} (24);
\path[edge2, draw] (37) to[] node[] {} (7);
\path[edge2, draw] (36) to[] node[] {} (18);
\path[edge2, draw] (35) to[] node[] {} (33);
\path[edge2, draw] (35) to[] node[] {} (28);
\path[edge2, draw] (35) to[] node[] {} (13);
\path[edge2, draw] (33) to[] node[] {} (34);
\path[edge2, draw] (33) to[] node[] {} (14);
\path[edge2, draw] (32) to[] node[] {} (6);
\path[edge2, draw] (30) to[] node[] {} (31);
\path[edge2, draw] (30) to[] node[] {} (29);
\path[edge2, draw] (30) to[] node[] {} (7);
\path[edge2, draw] (29) to[] node[] {} (2);
\path[edge2, draw] (28) to[] node[] {} (17);
\path[edge2, draw] (26) to[] node[] {} (27);
\path[edge2, draw] (26) to[] node[] {} (22);
\path[edge2, draw] (25) to[] node[] {} (17);
\path[edge2, draw] (23) to[] node[] {} (24);
\path[edge2, draw] (22) to[] node[] {} (9);
\path[edge2, draw] (21) to[] node[] {} (9);
\path[edge2, draw] (21) to[] node[] {} (7);
\path[edge2, draw] (19) to[] node[] {} (20);
\path[edge2, draw] (19) to[] node[] {} (9);
\path[edge2, draw] (18) to[] node[] {} (6);
\path[edge2, draw] (16) to[] node[] {} (17);
\path[edge2, draw] (15) to[] node[] {} (9);
\path[edge2, draw] (14) to[] node[] {} (8);
\path[edge2, draw] (13) to[] node[] {} (5);
\path[edge2, draw] (11) to[] node[] {} (12);
\path[edge2, draw] (11) to[] node[] {} (3);
\path[edge2, draw] (9) to[] node[] {} (10);
\path[edge2, draw] (7) to[] node[] {} (8);
\path[edge2, draw] (5) to[] node[] {} (6);
\path[edge2, draw] (4) to[] node[] {} (3);
\path[edge2, draw] (2) to[] node[] {} (3);
\path[edge2, draw] (2) to[] node[] {} (1);
\path[edge2, draw] (0) to[] node[] {} (1);
\end{tikzpicture}%

%% file: Figures/ex_tree2.tex
\begin{tikzpicture}
\node[] (label) at (1.5,1.2) {$T'$};
%\path[draw] (-1,0) to (1,0);
%\path[draw] (0,-1) to (0,1);

\path (-2,1.5) rectangle (2,-1.5);

\begin{scope}[shift={(-.3,.4)}]

\begin{pgfonlayer}{background}
\path[draw=black!75!white, fill=black!65!white , line width=10pt, line cap=round, line join=round] (0.5493,-0.498462) -- (0.56565,-0.673582) -- cycle;
\path[draw=black!35!white, fill=black!65!white , line width=6pt, line cap=round, line join=round] (0.5493,-0.498462) -- cycle;
\path[draw=black!60!white, fill=black!65!white , line width=6pt, line cap=round, line join=round] (0.56565,-0.673582) -- cycle;
\end{pgfonlayer}
\node[node2, name path=38, label=center:{}] (38) at (0.288846,0.101932) {};
\node[node2, name path=34, label=center:{}] (34) at (-0.19331,0.034222) {};
\node[node2, name path=33, label=center:{}] (33) at (-0.13378,-0.291308) {};
\node[node2, name path=29, label=center:{}] (29) at (0.56565,-0.673582) {};
\node[node2, name path=25, label=center:{}] (25) at (-0.56261,-0.503125) {};
\node[node2, name path=24, label=center:{}] (24) at (0.300373,0.263362) {};
\node[node2, name path=22, label=center:{}] (22) at (1.46219,0.138038) {};
\node[node2, name path=21, label=center:{}] (21) at (0.92441,-0.333462) {};
\node[node2, name path=20, label=center:{}] (20) at (1.60249,-0.342866) {};
\node[node2, name path=18, label=center:{}] (18) at (-1.11346,0.632967) {};
\node[node2, name path=17, label=center:{}] (17) at (-0.23672,-0.419945) {};
\node[node2, name path=16, label=center:{}] (16) at (-0.38141,-0.730275) {};
\node[node2, name path=14, label=center:{}] (14) at (0.24462,-0.266028) {};
\node[node2, name path=13, label=center:{}] (13) at (-0.41561,-0.037103) {};
\node[node2, name path=12, label=center:{}] (12) at (0.45932,-1.28597) {};
\node[node2, name path=10, label=center:{}] (10) at (1.26519,-0.177012) {};
\node[node2, name path=8, label=center:{}] (8) at (0.5493,-0.498462) {};
\node[node2, name path=6, label=center:{}] (6) at (-0.99346,0.317477) {};
\node[node2, name path=5, label=center:{}] (5) at (-0.69296,0.163547) {};
\node[node2, name path=4, label=center:{}] (4) at (0.6995,-1.51577) {};
\node[node2, name path=3, label=center:{}] (3) at (0.57753,-1.17817) {};
\node[node2, name path=2, label=center:{}] (2) at (0.65441,-1.03282) {};
\node[node2, name path=1, label=center:{}] (1) at (0.98179,-1.20637) {};

\path[edge2, draw] (38) to[] node[] {} (24);
\path[edge2, draw] (33) to[] node[] {} (13);
\path[edge2, draw] (33) to[] node[] {} (17);
\path[edge2, draw] (33) to[] node[] {} (34);
\path[edge2, draw] (33) to[] node[] {} (14);
\path[edge2, draw] (29) to[] node[] {} (8);
\path[edge2, draw] (29) to[] node[] {} (2);
\path[edge2, draw] (25) to[] node[] {} (17);
\path[edge2, draw] (22) to[] node[] {} (10);
\path[edge2, draw] (21) to[] node[] {} (10);
\path[edge2, draw] (21) to[] node[] {} (8);
\path[edge2, draw] (20) to[] node[] {} (10);
\path[edge2, draw] (18) to[] node[] {} (6);
\path[edge2, draw] (16) to[] node[] {} (17);
\path[edge2, draw] (14) to[] node[] {} (38);
\path[edge2, draw] (14) to[] node[] {} (8);
\path[edge2, draw] (13) to[] node[] {} (5);
\path[edge2, draw] (12) to[] node[] {} (3);
\path[edge2, draw] (5) to[] node[] {} (6);
\path[edge2, draw] (4) to[] node[] {} (3);
\path[edge2, draw] (2) to[] node[] {} (3);
\path[edge2, draw] (2) to[] node[] {} (1);
\end{scope}
\end{tikzpicture}%

%% file: Figures/ex_tree3.tex
\begin{tikzpicture}
\node[] (label) at (1.5,1.2) {$T''$};
%\path[draw] (-1,0) to (1,0);
%\path[draw] (0,-1) to (0,1);

\path (-2,1.5) rectangle (2,-1.5);

\begin{scope}[shift={(-.5,.1)}]
\node[fill=black!75!white, inner sep=0pt, text width=0pt, text height=0pt, text depth=0pt, minimum size = 10pt, circle, label=center:{}] (8a) at (0.725737,-0.313689) {};
\node[node2, name path=34, label=center:{}] (34) at (0.042657,-0.106535) {};
\node[node2, name path=25, label=center:{}] (25) at (-0.283233,-0.189715) {};
\node[node2, name path=24, label=center:{}] (24) at (0.465283,0.286705) {};
\node[node2, name path=21, label=center:{}] (21) at (1.10085,-0.148689) {};
\node[node2, name path=20, label=center:{}] (20) at (1.77893,-0.158093) {};
\node[node2, name path=16, label=center:{}] (16) at (-0.102033,-0.416865) {};
\node[node2, name path=14, label=center:{}] (14) at (0.421057,-0.081255) {};
\node[node2, name path=13, label=center:{}] (13) at (-0.239173,0.14767) {};
\node[node2, name path=10, label=center:{}] (10) at (1.44163,0.007761) {};
\node[node2, name path=8, label=center:{}] (8) at (0.725737,-0.313689) {};
\node[node2, name path=6, label=center:{}] (6) at (-0.817023,0.50225) {};
\node[node2, name path=5, label=center:{}] (5) at (-0.516523,0.34832) {};
\node[node2, name path=3, label=center:{}] (3) at (0.737617,-0.818279) {};
\node[node2, name path=1, label=center:{}] (1) at (0.814497,-0.672929) {};

\path[edge2, draw] (34) to[] node[] {} (13);
\path[edge2, draw] (34) to[] node[] {} (14);
\path[edge2, draw] (25) to[] node[] {} (34);
\path[edge2, draw] (21) to[] node[] {} (10);
\path[edge2, draw] (21) to[] node[] {} (8);
\path[edge2, draw] (20) to[] node[] {} (10);
\path[edge2, draw] (16) to[] node[] {} (34);
\path[edge2, draw] (14) to[] node[] {} (24);
\path[edge2, draw] (14) to[] node[] {} (8);
\path[edge2, draw] (13) to[] node[] {} (5);
\path[edge2, draw] (8) to[] node[] {} (1);
\path[edge2, draw] (5) to[] node[] {} (6);
\path[edge2, draw] (1) to[] node[] {} (3);
\end{scope}
\end{tikzpicture}%

%% file: Figures/ex_graph1.tex
\begin{tikzpicture}[scale=.095]

\node[] (label) at (-18,12) {$G$};
%\path[draw] (-1,0) to (1,0);
%\path[draw] (0,-1) to (0,1);
\path (-23,15) rectangle (23,-15);

\begin{pgfonlayer}{background}
\path[draw=black!65!white, fill=black!65!white , line width=12pt, line cap=round, line join=round] (12.7254,8.26723) -- (13.1385,3.56592) -- (19.8412,6.25391) -- (17.8811,9.07593) -- cycle;
\path[draw=black!35!white, fill=black!35!white , line width=8pt, line cap=round, line join=round] (19.8412,6.25391) -- (17.8811,9.07593) -- (16.98,6.41921) -- cycle;
\end{pgfonlayer}
\node[node2, name path=29, label=center:{}] (29) at (-7.65365,7.70803) {};
\node[node2, name path=28, label=center:{}] (28) at (11.9184,0.338397) {};
\node[node2, name path=27, label=center:{}] (27) at (-13.972,-2.34916) {};
\node[node2, name path=26, label=center:{}] (26) at (-16.5319,-4.71761) {};
\node[node2, name path=25, label=center:{}] (25) at (-5.20989,7.93541) {};
\node[node2, name path=24, label=center:{}] (24) at (18.1616,1.04605) {};
\node[node2, name path=23, label=center:{}] (23) at (6.16376,8.797) {};
\node[node2, name path=22, label=center:{}] (22) at (17.8811,9.07593) {};
\node[node2, name path=21, label=center:{}] (21) at (11.2636,-2.81485) {};
\node[node2, name path=20, label=center:{}] (20) at (-17.1935,-0.596147) {};
\node[node2, name path=19, label=center:{}] (19) at (-2.71112,-8.85681) {};
\node[node2, name path=18, label=center:{}] (18) at (14.8657,0.594226) {};
\node[node2, name path=17, label=center:{}] (17) at (0.591991,2.37893) {};
\node[node2, name path=16, label=center:{}] (16) at (19.8412,6.25391) {};
\node[node2, name path=15, label=center:{}] (15) at (2.66842,-6.36021) {};
\node[node2, name path=14, label=center:{}] (14) at (-2.31258,-3.60638) {};
\node[node2, name path=13, label=center:{}] (13) at (-2.30594,-6.31937) {};
\node[node2, name path=12, label=center:{}] (12) at (-1.19764,6.94926) {};
\node[node2, name path=11, label=center:{}] (11) at (12.7254,8.26723) {};
\node[node2, name path=10, label=center:{}] (10) at (-11.9707,-5.27209) {};
\node[node2, name path=9, label=center:{}] (9) at (-6.83178,-2.64196) {};
\node[node2, name path=8, label=center:{}] (8) at (-7.82619,-7.01996) {};
\node[node2, name path=7, label=center:{}] (7) at (8.1442,6.3882) {};
\node[node2, name path=6, label=center:{}] (6) at (-13.3975,1.44172) {};
\node[node2, name path=5, label=center:{}] (5) at (7.72834,-0.141162) {};
\node[node2, name path=4, label=center:{}] (4) at (16.98,6.41921) {};
\node[node2, name path=3, label=center:{}] (3) at (13.1385,3.56592) {};
\node[node2, name path=2, label=center:{}] (2) at (2.3747,-2.27007) {};
\node[node2, name path=1, label=center:{}] (1) at (-6.27559,4.26573) {};
\node[node2, name path=0, label=center:{}] (0) at (-19.2212,-3.01526) {};

\path[edge2, draw] (29) to[] node[] {} (25);
\path[edge2, draw] (28) to[] node[] {} (21);
\path[edge2, draw] (28) to[] node[] {} (18);
\path[edge2, draw] (28) to[] node[] {} (3);
\path[edge2, draw] (27) to[] node[] {} (6);
\path[edge2, draw] (26) to[] node[] {} (10);
\path[edge2, draw] (26) to[] node[] {} (27);
\path[edge2, draw] (26) to[] node[] {} (0);
\path[edge2, draw] (25) to[] node[] {} (12);
\path[edge2, draw] (24) to[] node[] {} (16);
\path[edge2, draw] (23) to[] node[] {} (12);
\path[edge2, draw] (22) to[] node[] {} (4);
\path[edge2, draw] (21) to[] node[] {} (24);
\path[edge2, draw] (21) to[] node[] {} (5);
\path[edge2, draw] (21) to[] node[] {} (18);
\path[edge2, draw] (20) to[] node[] {} (6);
\path[edge2, draw] (20) to[] node[] {} (27);
\path[edge2, draw] (19) to[] node[] {} (13);
\path[edge2, draw] (19) to[] node[] {} (8);
\path[edge2, draw] (19) to[] node[] {} (15);
\path[edge2, draw] (18) to[] node[] {} (24);
\path[edge2, draw] (18) to[] node[] {} (3);
\path[edge2, draw] (17) to[] node[] {} (2);
\path[edge2, draw] (17) to[] node[] {} (12);
\path[edge2, draw] (16) to[] node[] {} (22);
\path[edge2, draw] (15) to[] node[] {} (21);
\path[edge2, draw] (15) to[] node[] {} (2);
\path[edge2, draw] (14) to[] node[] {} (2);
\path[edge2, draw] (14) to[] node[] {} (9);
\path[edge2, draw] (13) to[] node[] {} (14);
\path[edge2, draw] (13) to[] node[] {} (9);
\path[edge2, draw] (13) to[] node[] {} (15);
\path[edge2, draw] (12) to[] node[] {} (1);
\path[edge2, draw] (11) to[] node[] {} (22);
\path[edge2, draw] (11) to[] node[] {} (3);
\path[edge2, draw] (11) to[] node[] {} (4);
\path[edge2, draw] (11) to[] node[] {} (23);
\path[edge2, draw] (10) to[] node[] {} (27);
\path[edge2, draw] (10) to[] node[] {} (9);
\path[edge2, draw] (9) to[] node[] {} (1);
\path[edge2, draw] (9) to[] node[] {} (17);
\path[edge2, draw] (9) to[] node[] {} (27);
\path[edge2, draw] (8) to[] node[] {} (10);
\path[edge2, draw] (8) to[] node[] {} (9);
\path[edge2, draw] (7) to[] node[] {} (11);
\path[edge2, draw] (7) to[] node[] {} (23);
\path[edge2, draw] (7) to[] node[] {} (3);
\path[edge2, draw] (7) to[] node[] {} (17);
\path[edge2, draw] (6) to[] node[] {} (1);
\path[edge2, draw] (5) to[] node[] {} (3);
\path[edge2, draw] (5) to[] node[] {} (28);
\path[edge2, draw] (5) to[] node[] {} (17);
\path[edge2, draw] (4) to[] node[] {} (24);
\path[edge2, draw] (4) to[] node[] {} (16);
\path[edge2, draw] (3) to[] node[] {} (4);
\path[edge2, draw] (3) to[] node[] {} (24);
\path[edge2, draw] (2) to[] node[] {} (5);
\path[edge2, draw] (1) to[] node[] {} (25);
\path[edge2, draw] (1) to[] node[] {} (29);
\path[edge2, draw] (1) to[] node[] {} (17);
\path[edge2, draw] (0) to[] node[] {} (20);
\end{tikzpicture}%

%% file: Figures/ex_graph2.tex
\begin{tikzpicture}[scale=.095]

\node[] (label) at (-18,12) {$G'$};
%\path[draw] (-1,0) to (1,0);
%\path[draw] (0,-1) to (0,1);
\path (-23,15) rectangle (23,-15);

\node[draw=black!35!white, fill=black!35!white, inner sep=0pt, text width=0pt, text height=0pt, text depth=0pt, minimum size = 8pt, circle, label=center:{}] (8a) at (18.0332,7.37415) {};
\begin{pgfonlayer}{background}
\path[draw=black!65!white, fill=black!65!white , line width=12pt, line cap=round, line join=round] (12.7254,8.26723) -- (18.0332,7.37415) -- (13.1385,3.56592) -- cycle;
\end{pgfonlayer}
\node[node2, name path=27, label=center:{}] (27) at (-14.1815,-4.33343) {};
\node[node2, name path=25, label=center:{}] (25) at (-6.35368,6.04373) {};
\node[node2, name path=24, label=center:{}] (24) at (18.1616,1.04605) {};
\node[node2, name path=23, label=center:{}] (23) at (7.15398,7.5926) {};
\node[node2, name path=20, label=center:{}] (20) at (-18.2073,-1.8057) {};
\node[node2, name path=18, label=center:{}] (18) at (12.4097,-0.780114) {};
\node[node2, name path=17, label=center:{}] (17) at (0.591991,2.37893) {};
\node[node2, name path=14, label=center:{}] (14) at (-2.41055,-5.59724) {};
\node[node2, name path=12, label=center:{}] (12) at (-1.19764,6.94926) {};
\node[node2, name path=11, label=center:{}] (11) at (12.7254,8.26723) {};
\node[node2, name path=9, label=center:{}] (9) at (-6.83178,-2.64196) {};
\node[node2, name path=8, label=center:{}] (8) at (-7.82619,-7.01996) {};
\node[node2, name path=6, label=center:{}] (6) at (-13.3975,1.44172) {};
\node[node2, name path=5, label=center:{}] (5) at (7.72834,-0.141162) {};
\node[node2, name path=4, label=center:{}] (4) at (18.0332,7.37415) {};
\node[node2, name path=3, label=center:{}] (3) at (13.1385,3.56592) {};
\node[node2, name path=2, label=center:{}] (2) at (2.52156,-4.31514) {};

\path[edge2, draw] (27) to[] node[] {} (20);
\path[edge2, draw] (27) to[] node[] {} (6);
\path[edge2, draw] (25) to[] node[] {} (17);
\path[edge2, draw] (24) to[] node[] {} (4);
\path[edge2, draw] (23) to[] node[] {} (17);
\path[edge2, draw] (23) to[] node[] {} (3);
\path[edge2, draw] (23) to[] node[] {} (12);
\path[edge2, draw] (20) to[] node[] {} (6);
\path[edge2, draw] (18) to[] node[] {} (24);
\path[edge2, draw] (18) to[] node[] {} (3);
\path[edge2, draw] (17) to[] node[] {} (2);
\path[edge2, draw] (17) to[] node[] {} (12);
\path[edge2, draw] (14) to[] node[] {} (8);
\path[edge2, draw] (14) to[] node[] {} (2);
\path[edge2, draw] (14) to[] node[] {} (9);
\path[edge2, draw] (12) to[] node[] {} (25);
\path[edge2, draw] (11) to[] node[] {} (3);
\path[edge2, draw] (11) to[] node[] {} (4);
\path[edge2, draw] (11) to[] node[] {} (23);
\path[edge2, draw] (9) to[] node[] {} (25);
\path[edge2, draw] (9) to[] node[] {} (17);
\path[edge2, draw] (9) to[] node[] {} (27);
\path[edge2, draw] (8) to[] node[] {} (27);
\path[edge2, draw] (8) to[] node[] {} (9);
\path[edge2, draw] (6) to[] node[] {} (25);
\path[edge2, draw] (5) to[] node[] {} (3);
\path[edge2, draw] (5) to[] node[] {} (18);
\path[edge2, draw] (5) to[] node[] {} (17);
\path[edge2, draw] (3) to[] node[] {} (4);
\path[edge2, draw] (3) to[] node[] {} (24);
\path[edge2, draw] (2) to[] node[] {} (18);
\path[edge2, draw] (2) to[] node[] {} (5);
\end{tikzpicture}%

%% file: Figures/ex_graph3.tex
\begin{tikzpicture}[scale=.095]
\node[] (label) at (-18,12) {$G''$};
%\path[draw] (-1,0) to (1,0);
%\path[draw] (0,-1) to (0,1);
\path (-23,15) rectangle (23,-15);
\node[fill=black!65!white, inner sep=0pt, text width=0pt, text height=0pt, text depth=0pt, minimum size = 12pt, circle, label=center:{}] (8a) at (14.2072,6.28097) {};
\node[draw=black!35!white, fill=black!35!white, inner sep=0pt, text width=0pt, text height=0pt, text depth=0pt, minimum size = 8pt, circle, label=center:{}] (8a) at (14.2072,6.28097) {};
\node[node2, name path=25, label=center:{}] (25) at (-6.35368,6.04373) {};
\node[node2, name path=24, label=center:{}] (24) at (18.1616,1.04605) {};
\node[node2, name path=23, label=center:{}] (23) at (7.15398,7.5926) {};
\node[node2, name path=18, label=center:{}] (18) at (12.4097,-0.780114) {};
\node[node2, name path=12, label=center:{}] (12) at (-1.19764,6.94926) {};
\node[node2, name path=9, label=center:{}] (9) at (-4.62116,-4.1196) {};
\node[node2, name path=8, label=center:{}] (8) at (-7.82619,-7.01996) {};
\node[node2, name path=6, label=center:{}] (6) at (-15.4952,-1.94174) {};
\node[node2, name path=5, label=center:{}] (5) at (7.72834,-0.141162) {};
\node[node2, name path=4, label=center:{}] (4) at (14.2072,6.28097) {};
\node[node2, name path=2, label=center:{}] (2) at (1.55678,-0.968104) {};

\path[edge2, draw] (25) to[] node[] {} (2);
\path[edge2, draw] (23) to[] node[] {} (2);
\path[edge2, draw] (23) to[] node[] {} (4);
\path[edge2, draw] (23) to[] node[] {} (12);
\path[edge2, draw] (18) to[] node[] {} (24);
\path[edge2, draw] (18) to[] node[] {} (4);
\path[edge2, draw] (12) to[] node[] {} (25);
\path[edge2, draw] (9) to[] node[] {} (2);
\path[edge2, draw] (9) to[] node[] {} (8);
\path[edge2, draw] (9) to[] node[] {} (25);
\path[edge2, draw] (9) to[] node[] {} (6);
\path[edge2, draw] (8) to[] node[] {} (6);
\path[edge2, draw] (6) to[] node[] {} (25);
\path[edge2, draw] (5) to[] node[] {} (4);
\path[edge2, draw] (5) to[] node[] {} (18);
\path[edge2, draw] (5) to[] node[] {} (2);
\path[edge2, draw] (4) to[] node[] {} (24);
\path[edge2, draw] (2) to[] node[] {} (12);
\path[edge2, draw] (2) to[] node[] {} (18);
\end{tikzpicture}%

%% file: Figures/leafs.tex
\begin{tikzpicture}[scale=0.55]
\tikzstyle{VertexStyle}=[vertex_default]
\tikzstyle{EdgeStyle}=[edge_default]
\SetVertexNoLabel

\Vertex[x=0,y=4]{v0}
\Vertex[x=2,y=4]{v1}
\Vertex[x=4,y=4]{v2}
\Vertex[x=2,y=2]{v3}
\Vertex[x=0,y=0]{v4}

\Vertex[x=6,y=4]{v5}
\Vertex[x=8,y=4]{v6}
\Vertex[x=10,y=4]{v7}
\Vertex[x=12,y=4]{v8}
\Vertex[x=14,y=4]{v9}
\Vertex[x=16,y=4]{v10}
\Vertex[x=18,y=4]{v13}

\Vertex[x=10,y=2]{v11}
\Vertex[x=10,y=0]{v12}

\Edge(v0)(v1)
\Edge(v1)(v2)
\Edge(v2)(v3)
\Edge(v3)(v4)

\Edge(v2)(v5)
\Edge(v5)(v6)
\Edge(v6)(v7)
\Edge(v7)(v8)

\Edge(v8)(v9)
\Edge(v9)(v10)
\Edge(v10)(v13)

\Edge(v7)(v11)
\Edge(v11)(v12)

\begin{pgfonlayer}{background}
\path[draw=black!55!white, fill=black!55!white, line width=8pt, line cap=round, line join=round] (2,2) -- (16,2) -- (16,6) -- (2,6) -- cycle;
\path[draw=black!35!white, fill=black!35!white, line width=8pt, line cap=round, line join=round] (4,3.5) -- (10,4) -- (14,3.5) -- (14,5.8) -- (4,5.8) -- cycle;
\path[draw=black!15!white, fill=black!15!white, line width=8pt, line cap=round, line join=round] (6,3.8) -- (10,4) -- (12,3.8) -- (12,5) -- (6,5) -- cycle;
\end{pgfonlayer}

\node [anchor=west] at (4,1){$L(T,1)$};
\node [anchor=west] at (4,2.5){$L(T,2)\setminus L(T,1)$};
\node [anchor=west] at (4.1,5.55){$L(T,3)\setminus L(T,2)$};
\node [anchor=west] at (6.2,4.7){$T\setminus L(T,3)$};

\end{tikzpicture}

%% file: Figures/exchange.tex
\begin{tikzpicture}[scale=0.7]
\tikzstyle{VertexStyle}=[vertex_default]
\SetVertexNoLabel

\Vertex[style={label={left:$v$}},x=-2,y=4]{v0}
\Vertex[x=0,y=4]{v1}
\Vertex[x=2,y=4]{v2}
\Vertex[x=4,y=4]{v3}
\Vertex[x=6,y=4]{v4}
\Vertex[x=8,y=4]{v5}
\Vertex[x=10,y=4]{v6}
\Vertex[x=12,y=4]{v7}
\Vertex[x=14,y=4]{v8}

\Vertex[x=4,y=2]{v9}
\Vertex[x=4,y=0]{v10}

\Vertex[x=10,y=2]{v11}
\Vertex[x=8,y=2]{v12}

\tikzstyle{EdgeStyle}=[edge_default]
\Edge[label={$e' \in E^* \setminus D$},labelstyle={above,xshift=10,yshift=3}](v0)(v1)
%\node [style={label={above:}}] (12) at (0,4) {};
\Edge(v1)(v2)
\Edge(v2)(v3)
\Edge(v3)(v4)
\Edge(v5)(v6)
\Edge(v6)(v7)
\Edge(v7)(v8)
\Edge(v3)(v9)
\Edge(v9)(v10)
\Edge(v11)(v12)

\tikzstyle{EdgeStyle}=[edge_in_c]
\Edge[label={$f \in D \setminus E^*$},labelstyle={right,xshift=3}](v6)(v11)
\Edge[label={$e \in D \setminus E^*$},labelstyle={below,xshift=20,yshift=-3}](v4)(v5)

\path [draw,decoration={brace,amplitude=5pt,raise=45pt},decorate] (v0) to node[above=50pt] {$Q$} (v8);
\path [draw,decoration={brace,amplitude=5pt,raise=25pt},decorate] (v0) to node[above=28pt] {$P'\cap Q$} (v3);
\path [draw,decoration={brace,amplitude=5pt,raise=25pt},decorate] (v3) to node[above=28pt] {$P\cap Q=Q\setminus P'$} (v8);

\path [draw,decoration={brace,amplitude=5pt,mirror,raise=8pt},decorate] (v3) to node[left=15pt] {$P'\setminus Q=P\setminus Q$} (v10);

\path [draw,decoration={brace,amplitude=5pt,mirror,raise=20pt},decorate] (v0) to node[below left=25pt] {$P'$} (v10);
\path [draw,decoration={brace,amplitude=5pt,mirror,raise=15pt},decorate] (v10) to node[below right=20pt] {$P$} (v8);

\draw[dashed,color=blue,<->] (-1,3.9) to[out=-25,in=-155] (7,3.9);

\end{tikzpicture}

%% file: Figures/load.tex
\begin{tikzpicture}[scale=0.8]
\tikzstyle{VertexStyle}=[vertex_default]
\SetVertexNoLabel

\Vertex[style={label={right:$r$}},x=8,y=12.5]{v0}
\Vertex[style={label={above left:$v$}},x=8,y=11]{v1}

\Vertex[style={label={above:$u_1$}},x=4,y=10]{v2}
\Vertex[style={label={below right:$u_2$}},x=6.5,y=10]{v3}
\Vertex[style={label={right:$u_3$}},x=9.5,y=10]{v4}
\Vertex[style={label={above:$u_4$}},x=12,y=10]{v5}

\Vertex[x=2,y=8.5]{v6}
\Vertex[style={label={right:$w$}},x=2,y=7]{v7}

\Vertex[x=4,y=8.5]{v8}
\Vertex[x=4,y=7]{v9}
\Vertex[x=6,y=8.5]{v10}
\Vertex[x=6,y=7]{v11}
\Vertex[x=14,y=9]{v12}
\Vertex[x=8,y=9]{v13}
\Vertex[x=8,y=7.5]{v14}
\Vertex[x=8,y=6]{v15}
\Vertex[x=11,y=9]{v16}
\Vertex[x=11,y=7.5]{v17}
\Vertex[style={label={right:$w'$}},x=11,y=6]{v18}

\tikzstyle{LabelStyle}=[inner sep=1pt, circle]
\tikzstyle{EdgeStyle}=[edge_default]
\Edge[label=2](v0)(v1)
\Edge[label=\hbox{3}](v1)(v5)
\Edge[label=\hbox{3}](v5)(v12)

\tikzstyle{LabelStyle}=[fill=black!45!white, inner sep=1pt, circle]
\Edge[label=\hbox{2}](v2)(v6)
\Edge[label=\hbox{6}](v3)(v10)
\Edge[label=\hbox{5}](v8)(v9)

\tikzstyle{LabelStyle}=[fill=black!25!white, inner sep=1pt, circle]
\Edge[label=\hbox{3}](v4)(v13)
\Edge[label=\hbox{2}](v4)(v16)
\Edge[label=\hbox{2}](v13)(v14)

\tikzstyle{EdgeStyle}=[edge_in_c]
\tikzstyle{LabelStyle}=[fill=black!45!white, inner sep=1pt, circle]
\Edge[label=\hbox{1}](v1)(v2)
\Edge[label=\hbox{1}](v1)(v3)
\Edge[label=\hbox{1}](v3)(v8)
\Edge[label=\hbox{4}](v6)(v7)
\Edge[label=\hbox{3}](v10)(v11)

\tikzstyle{LabelStyle}=[fill=black!25!white, inner sep=1pt, circle]
\Edge[label=\hbox{2}](v1)(v4)
\Edge[label=\hbox{5}](v14)(v15)
\Edge[label=\hbox{1}](v16)(v17)
\Edge[label=\hbox{4}](v17)(v18)

\begin{pgfonlayer}{background}
\path[draw=black!45!white, fill=black!45!white, line width=12pt, line cap=round, line join=round] (v1.center) -- (v2.center) -- (v6.center) -- (v7.center) -- (v11.center) -- (v10.center) -- (v3.center) -- cycle;
\path[draw=black!25!white, fill=black!25!white, line width=12pt, line cap=round, line join=round] (v1.center) -- (v4.center) -- (v13.center) -- (v15.center) -- (v18.center) -- (v16.center) --(v4.center) -- cycle;
\end{pgfonlayer}

\node at (2,10){$T^+_{v,2}$};
\node at (12.5,7.5){$T_{v,3}$};
\path [draw,decoration={brace,mirror,amplitude=5pt,raise=10pt},decorate] (1.8,6) to node[below=20pt] {$T^+_{v,3} = T^+_{v,2} \cup T_{v,3}$} (11.2,6);

\end{tikzpicture}

%% file: Figures/wload_1.tex
\begin{tikzpicture}[scale=0.6]
\tikzstyle{VertexStyle}=[vertex_default]
\SetVertexNoLabel

\Vertex[x=0,y=0]{v0}
\Vertex[x=2,y=0]{v1}
\Vertex[x=1,y=2]{v2}
\Vertex[style={label={right:$v$}},x=1,y=4]{v3}

\tikzstyle{EdgeStyle}=[edge_default]
\Edge[label=\hbox{3}](v2)(v3)

\tikzstyle{EdgeStyle}=[edge_in_c]
\Edge[label=\hbox{4}](v0)(v2)
\Edge[label=\hbox{4}](v1)(v2)

\node at (0,4){$T_v$};
\node at (-1,1){$C$};

\end{tikzpicture}

%% file: Figures/wload_2.tex
\begin{tikzpicture}[scale=0.6]
\tikzstyle{VertexStyle}=[vertex_default]
\SetVertexNoLabel

\Vertex[x=0,y=0]{v0}
\Vertex[x=2,y=0]{v1}
\Vertex[x=1,y=2]{v2}
\Vertex[style={label={right:$v$}},x=1,y=4]{v3}

\tikzstyle{EdgeStyle}=[edge_default]
\Edge[label=\hbox{4}](v1)(v2)

\tikzstyle{EdgeStyle}=[edge_in_c]
\Edge[label=\hbox{4}](v0)(v2)
\Edge[label=\hbox{3}](v2)(v3)

\node at (0,4){$T_v$};
\node at (-1,1){$C'$};

\end{tikzpicture}

%% file: Figures/wload_3.tex
\begin{tikzpicture}[scale=0.6]
\tikzstyle{VertexStyle}=[vertex_default]
\SetVertexNoLabel

\Vertex[x=0,y=0]{v0}
\Vertex[x=2,y=0]{v1}
\Vertex[x=1,y=2]{v2}
\Vertex[style={label={right:$v$}},x=1,y=4]{v3}
\Vertex[x=1,y=6]{v4}

\tikzstyle{EdgeStyle}=[edge_default]
\Edge[label=\hbox{3}](v2)(v3)
\Edge[label=\hbox{3}](v3)(v4)

\tikzstyle{EdgeStyle}=[edge_in_c]
\Edge[label=\hbox{4}](v0)(v2)
\Edge[label=\hbox{4}](v1)(v2)

\node at (0,6){$T$};

\end{tikzpicture}

%% file: Figures/wload_4.tex
\begin{tikzpicture}[scale=0.6]
\tikzstyle{VertexStyle}=[vertex_default]
\SetVertexNoLabel

\Vertex[x=0,y=0]{v0}
\Vertex[x=2,y=0]{v1}
\Vertex[x=1,y=2]{v2}
\Vertex[style={label={right:$v$}},x=1,y=4]{v3}
\Vertex[x=1,y=6]{v4}

\tikzstyle{EdgeStyle}=[edge_default]
\Edge[label=\hbox{4}](v1)(v2)

\tikzstyle{EdgeStyle}=[edge_in_c]
\Edge[label=\hbox{4}](v0)(v2)
\Edge[label=\hbox{3}](v2)(v3)
\Edge[label=\hbox{1}](v3)(v4)

\node at (0,6){$T$};

\end{tikzpicture}

%% file: Figures/pareto_1.tex
\begin{tikzpicture}[scale=0.9]
\tikzstyle{VertexStyle}=[vertex_default]
\tikzstyle{EdgeStyle}=[edge_default]
\SetVertexNoLabel

\Vertex[style={label={right:$v$}},x=3,y=4]{v0}
\Vertex[x=1,y=2]{v1}
\Vertex[x=5,y=2]{v2}
\Vertex[x=0,y=0]{v3}
\Vertex[x=2,y=0]{v4}
\Vertex[x=4,y=0]{v5}
\Vertex[x=6,y=0]{v6}

\Edge[label=\hbox{3}](v0)(v1)
\Edge[label=\hbox{1}](v0)(v2)
\Edge[label=\hbox{5}](v1)(v3)
\Edge[label=\hbox{1}](v1)(v4)
\Edge[label=\hbox{5}](v2)(v5)
\Edge[label=\hbox{6}](v2)(v6)

\node at (2,4){$T_v$};

\end{tikzpicture}

%% file: Figures/pareto_2.tex
\begin{tikzpicture}
\begin{scope}[xscale=1.2, yscale=1.2]
% draw axes
\draw[->,very thick] (0,0) -- (5,0) node[below] {$\load$};
\draw[->,very thick] (0,-2) -- (0,3.5) node[left] {$\wload$};
\draw[dashed,thick] (0,0) -- (3.5,3.5) node[right] {$\wload=\load$};
% draw axes ticks
\foreach \x in {0,2,3,4}
\draw[very thick] (\x,-0.2) -- (\x,0) node {};
\draw[very thick] (1,-0.2) -- (1,0) node[below=5pt] {1};
\foreach \y in {0,2,3}
\draw[very thick] (-0.2,\y) -- (0,\y) node {};
\draw[very thick] (-0.2,1) -- (0,1) node[left=5pt] {1};
\draw[very thick] (-0.2,-1) -- (0,-1) node[left=5pt] {-1};

% draw red lines
\draw[dashed, red, very thick] (0,-2) -- (0,3);
\draw[dashed, red, very thick] (0.5,-2) -- (0.5,3);
\draw[dashed, red, very thick] (1.5,-2) -- (1.5,3);
\draw[dashed, red, very thick] (2,-2) -- (2,3);
\draw[dashed, red, very thick] (3,-2) -- (3,3);
\draw[dashed, red, very thick] (3.5,-2) -- (3.5,3);
\draw[dashed, red, very thick] (4,-2) -- (4,3);

% draw points
\draw plot [only marks,mark=*,mark options={scale=1.0,color=blue}] coordinates{(4,-0.5)(4,-1)(4,2)(4,2.5)(2,-1)(2,2)(2.5,2.5)(3,1)(3.5, 1)(1,1)(3,-1)(3.5,-1)(3.5,-1.5)};
\draw plot [only marks,mark=o,mark options={scale=1.8,color=blue,thick}] coordinates{(1, 1)(2,-1)(3.5,-1.5)};

% Pareto curve
\begin{pgfonlayer}{background}
\path[draw=black!35!white, line width=8pt, line cap=round, line join=round] (1,2) -- (1,1) -- (2,1) -- (2,-1) -- (3.5,-1) -- (3.5,-1.5) -- (4.5,-1.5);
\end{pgfonlayer}

\end{scope}
\end{tikzpicture}

%% file: Figures/cycle.tex
\begin{tikzpicture}[xscale=0.7,yscale=0.6]
\tikzstyle{VertexStyle}=[vertex_default]
\SetVertexNoLabel

\Vertex[style={label={left:$u_0$}},x=0,y=6]{0}
\Vertex[style={label={below:$u_1$}},x=3,y=6]{1}
\Vertex[style={label={below:$u_2$}},x=6,y=6]{2}
\Vertex[style={label={below:$u_3$}},x=9,y=6]{3}
\Vertex[style={label={below:$u_4$}},x=12,y=6]{4}
\Vertex[style={label={below:$u_5$}},x=15,y=6]{5}
\Vertex[style={label={right:$u_6$}},x=18,y=6]{6}

\Vertex[style={label={right:$v_1$}},x=18,y=4]{7}
\Vertex[style={label={right:$v_2$}},x=18,y=2]{8}

\Vertex[style={label={right:$w_0$}},x=18,y=0]{9}
\Vertex[style={label={above:$w_1$}},x=15,y=0]{10}
\Vertex[style={label={above:$w_2$}},x=12,y=0]{11}
\Vertex[style={label={above:$w_3$}},x=9,y=0]{12}
\Vertex[style={label={above:$w_4$}},x=6,y=0]{13}
\Vertex[style={label={above:$w_5$}},x=3,y=0]{14}
\Vertex[style={label={left:$w_6$}},x=0,y=0]{15}

\tikzstyle{EdgeStyle}=[edge_default]
\Edge[label={$a_3$}](2)(3)
\Edge[label={$a_4$}](3)(4)
\Edge[label={$a_6$}](5)(6)
\Edge[label={$\beta'$}](7)(8)
\Edge[label={$2-a_3$}](11)(12)
\Edge[label={$2-a_4$}](12)(13)
\Edge[label={$2-a_6$}](14)(15)
\Edge[label={$\beta'+2\varepsilon$}](15)(0)

\tikzstyle{EdgeStyle}=[edge_in_c]
\Edge[label={$a_1$}](0)(1)
\Edge[label={$a_2$}](1)(2)
\Edge[label={$a_5$}](4)(5)
\Edge[label={$\varepsilon$}](6)(7)
\Edge[label={$\varepsilon$}](8)(9)
\Edge[label={$2-a_1$}](9)(10)
\Edge[label={$2-a_2$}](10)(11)
\Edge[label={$2-a_5$}](13)(14)

\path [draw,decoration={brace,amplitude=5pt,raise=10pt},decorate] (0) to node[above=15pt] {$P_u$} (6);
\path [draw,decoration={brace,mirror,amplitude=5pt,raise=10pt},decorate] (15) to node[below=15pt] {$P_w$} (9);

\end{tikzpicture}

%% file: Figures/lollipop_g_1.tex
\begin{tikzpicture}[scale=0.7]
\SetVertexNoLabel
\tikzstyle{EdgeStyle}=[edge_default]

\tikzstyle{VertexStyle}=[vertex_red]
\Vertex[style={label={left:$u$}},x=0,y=-4]{v0}
\Vertex[style={label={right:$v$}},x=2,y=-4]{v1}
\Vertex[style={label={left:$w$}},x=1,y=-2]{v2}

\tikzstyle{VertexStyle}=[vertex_default]
\Vertex[style={label={left:$t$}},x=1,y=0]{v3}

\Edge(v0)(v1)
\Edge(v1)(v2)
\Edge(v2)(v0)
\Edge(v2)(v3)

\node at (0,1) {\Large $G$};
\end{tikzpicture}

%% file: Figures/lollipop_h_1.tex
\begin{tikzpicture}[scale=0.85]
\tikzstyle{VertexStyle}=[vertex_default]
\SetVertexNoLabel

\Vertex[style={label={above:$(u,2)$}},x=0,y=6]{v2}
\Vertex[style={label={below:$(u,1)$}},x=0,y=4]{v3}

\Vertex[style={label={left:$(v,2)$}},x=4,y=6]{v4}
\Vertex[style={label={below:$(v,1)$}},x=4,y=4]{v5}

\Vertex[style={label={right:$(w,2)$}},x=8,y=6]{v6}
\Vertex[style={label={below:$(w,1)$}},x=8,y=4]{v7}

\Vertex[style={label={above:$(t,2)$}},x=12,y=6]{v8}
\Vertex[style={label={below:$(t,1)$}},x=12,y=4]{v9}

\Vertex[style={label={above:$s$}},x=6,y=8]{v10}

\tikzstyle{EdgeStyle}=[edge_in_c]
\Edge[label=\hbox{$\beta$}](v2)(v3)
\Edge[label=\hbox{$\beta$}](v4)(v5)
\Edge[label=\hbox{$\beta$}](v6)(v7)

\tikzstyle{EdgeStyle}=[edge_default]
\Edge[label=\hbox{$\beta$}](v8)(v9)

\Edge[label=\hbox{$\beta+1$}](v10)(v2)
\Edge[label=\hbox{$\beta+1$}](v10)(v4)
\Edge[label=\hbox{$\beta+1$}](v10)(v6)
\Edge[label=\hbox{$\beta+1$}](v10)(v8)

\Edge[label=\hbox{$2\beta+2$}](v3)(v5)
\Edge[label=\hbox{$2\beta+2$}](v5)(v7)
\Edge[label=\hbox{$2\beta+2$}](v7)(v9)

\Edge[style={bend right},label=\hbox{$2\beta+2$}](v3)(v7)

\node at (0,8) {\Large $H=H(G)$};

\end{tikzpicture}

%% file: Figures/lollipop_g_2.tex
\begin{tikzpicture}[scale=0.7]
\SetVertexNoLabel
\tikzstyle{EdgeStyle}=[edge_default]

\tikzstyle{VertexStyle}=[vertex_red]
\Vertex[style={label={below:$u$}},x=3,y=-2]{v0}
\Vertex[style={label={below:$v$}},x=5,y=-2]{v1}
\Vertex[style={label={left:$w$}},x=4,y=0]{v2}

\tikzstyle{VertexStyle}=[vertex_default]
\Vertex[style={label={left:$t$}},x=4,y=2]{v3}

\Edge[label={$e$},labelstyle={below}](v0)(v1)
\Edge(v1)(v2)
\Edge(v2)(v0)
\Edge(v2)(v3)

\node at (3,3){\Large $G$};

\end{tikzpicture}

%% file: Figures/lollipop_h_2.tex
\begin{tikzpicture}[scale=0.85]
\tikzstyle{VertexStyle}=[vertex_default]
\SetVertexNoLabel

\path[use as bounding box] (4,-3.5) rectangle (14,8.5);

\Vertex[style={label={right:$(u,2)$}},x=8,y=6]{v2}
\Vertex[style={label={right:$(u,1)$}},x=8,y=4]{v3}

\Vertex[style={label={right:$(v,2)$}},x=10,y=6]{v4}
\Vertex[style={label={right:$(v,1)$}},x=10,y=4]{v5}

\Vertex[style={label={right:$(w,2)$}},x=12,y=6]{v6}
\Vertex[style={label={right:$(w,1)$}},x=12,y=4]{v7}

\Vertex[style={label={right:$(t,2)$}},x=14,y=6]{v8}
\Vertex[style={label={right:$(t,1)$}},x=14,y=4]{v9}

\Vertex[style={label={above:$s$}},x=8,y=8]{v10}

\Vertex[x=8,y=2]{v11}
\Vertex[style={label={below:$x_e$}},x=9,y=0.7]{v12}
\Vertex[x=11,y=-0.6667]{v13}
\Vertex[x=13,y=-2]{v14}

\tikzstyle{EdgeStyle}=[edge_in_c]
\Edge[label={$f_u$},labelstyle={right}](v2)(v3)
\Edge[label={$f_v$},labelstyle={right}](v4)(v5)
\Edge[label={$f_w$},labelstyle={right}](v6)(v7)

\tikzstyle{EdgeStyle}=[edge_default]
\Edge[label={$f_t$},labelstyle={right}](v8)(v9)

\Edge(v10)(v2)
\Edge(v10)(v4)
\Edge(v10)(v6)
\Edge(v10)(v8)

\Edge(v3)(v11)
\Edge(v7)(v11)

\Edge[label={$f_{e,u}$},labelstyle={left,pos=0.8}](v3)(v12)
\Edge[label={$f_{e,v}$},labelstyle={right,pos=0.8}](v5)(v12)

\Edge(v5)(v13)
\Edge(v7)(v13)

\Edge(v7)(v14)
\Edge(v9)(v14)

\Edge[style={bend right=20}](v10)(v11)
\Edge[style={bend right=60,looseness=1.4}](v10)(v12)
\Edge[style={bend right=80,looseness=1.5}](v10)(v13)
\Edge[style={bend right=90,looseness=1.7}](v10)(v14)

\node at (5,8){\Large $H=H(G)$};

\end{tikzpicture}

%% file: Figures/clause_1.tex
\begin{tikzpicture}[scale=0.7]
\tikzstyle{EdgeStyle}=[edge_default]
\SetVertexNoLabel

\tikzstyle{VertexStyle}=[vertex_default]
\Vertex[x=0,y=1,style={label={left:$y$}}]{v0}
\Vertex[x=2,y=1,style={label={above:$c$}}]{v1}
\Vertex[x=4,y=1,style={label={right:$z$}}]{v2}

\Edge(v0)(v1)
\tikzstyle{EdgeStyle}=[edge_dotted]
\Edge(v1)(v2)

\node at (0,-0.65) {};

\end{tikzpicture}

%% file: Figures/clause_2.tex
\begin{tikzpicture}[scale=0.7]
\tikzstyle{EdgeStyle}=[edge_default]
\SetVertexNoLabel

\tikzstyle{VertexStyle}=[vertex_default]
\Vertex[x=0,y=1,style={label={left:$y$}}]{v0}
\Vertex[x=2,y=2,style={label={above:$c_1$}}]{v1}
\Vertex[x=2,y=1,style={label={right:$x$}}]{v2}
\Vertex[x=2,y=0,style={label={below:$c_2$}}]{v3}
\Vertex[x=4,y=1,style={label={right:$z$}}]{v4}

\tikzstyle{EdgeStyle}=[edge_default]
\Edge(v0)(v1)
\Edge(v0)(v3)
\Edge(v2)(v1)

\tikzstyle{EdgeStyle}=[edge_dotted]
\Edge(v1)(v4)
\Edge(v3)(v4)
\Edge(v2)(v3)

\end{tikzpicture}

%% file: Figures/var_gadget.tex
\begin{tikzpicture}[scale=0.7]
\tikzstyle{EdgeStyle}=[edge_default]
\SetVertexNoLabel

\tikzstyle{VertexStyle}=[vertex_default]
\Vertex[x=0,y=0]{v0}
\Vertex[x=0,y=2]{v1}
\Vertex[x=2,y=3]{v2}
\Vertex[x=4,y=2]{v3}
\Vertex[x=4,y=0]{v4}
\Vertex[x=2,y=-1]{v5}

\Vertex[x=6,y=2]{v6}
\Vertex[x=8,y=2]{v7}

\Vertex[x=6,y=0]{v9}
\Vertex[x=8,y=0]{v10}

\tikzstyle{VertexStyle}=[vertex_bullet]
\Vertex[x=10,y=2,style={label={right:$u_i$}}]{v8}
\tikzstyle{VertexStyle}=[vertex_square]
\Vertex[x=10,y=0,style={label={right:$\ol{u}_i$}}]{v11}

\Edge[label={$h_i''$},labelstyle={left}](v0)(v1)
\Edge[label={$t_i''$},labelstyle={above left}](v1)(v2)
\Edge[label={$f_i'$},labelstyle={above right}](v2)(v3)
\Edge[label={$h_i'$},labelstyle={right}](v3)(v4)
\Edge[label={$t_i'$},labelstyle={below right}](v4)(v5)
\Edge[label={$f_i''$},labelstyle={below left}](v5)(v0)

\Edge[label={$h_i$},labelstyle={above}](v3)(v6)
\Edge[label={$t_i$},labelstyle={above}](v6)(v7)
\Edge[label={$f_i$},labelstyle={above}](v7)(v8)

\Edge[label={$\ol{h}_i$},labelstyle={below}](v4)(v9)
\Edge[label={$\ol{f}_i$},labelstyle={below}](v9)(v10)
\Edge[label={$\ol{t}_i$},labelstyle={below}](v10)(v11)

\node at (-1,3){$H(x_i)$};

\end{tikzpicture}

%% file: Figures/clause_gadget.tex
\begin{tikzpicture}[scale=0.7]
\tikzstyle{EdgeStyle}=[edge_default]
\SetVertexNoLabel

\tikzstyle{VertexStyle}=[vertex_default]
\Vertex[x=0,y=0]{v0}
\tikzstyle{VertexStyle}=[vertex_bullet]
\Vertex[x=0,y=2,style={label={right:$v_j^1$}}]{v1}
\Vertex[x=2,y=0,style={label={right:$v_j^2$}}]{v2}
\Vertex[x=0,y=-2,style={label={right:$v_j^3$}}]{v3}

\Edge[label={$e_j^1$},labelstyle={right}](v0)(v1)
\Edge[label={$e_j^2$},labelstyle={above,xshift=3}](v0)(v2)
\Edge[label={$e_j^3$},labelstyle={right}](v0)(v3)

\node at (-1.5,2){$H(c_j)$};

\end{tikzpicture}

%% file: Figures/var_clause_graph.tex
\begin{tikzpicture}[scale=0.8]
\tikzstyle{VertexStyle}=[vertex_default]
\SetVertexNoLabel

\tikzstyle{VertexStyle}=[vertex_default]
\Vertex[style={label={above:$x_1$}},x=6,y=4]{x1}
\Vertex[style={label={left:$c_1$}},x=0,y=2]{c1}
\Vertex[style={label={below:$x_4$}},x=2,y=2]{x4}
\Vertex[style={label={right:$c_2$}},x=4,y=2]{c2}
\Vertex[style={label={right:$c_3$}},x=6,y=2]{c3}
\Vertex[style={label={left:$c_4$}},x=8,y=2]{c4}
\Vertex[style={label={below:$x_5$}},x=10,y=2]{x5}
\Vertex[style={label={right:$c_5$}},x=12,y=2]{c5}
\Vertex[style={label={below:$x_3$}},x=4,y=0]{x3}
\Vertex[style={label={below:$x_2$}},x=8,y=0]{x2}

\tikzstyle{EdgeStyle}=[edge_default]
\Edge(x1)(c4)
\Edge(x1)(c5)
\Edge(x2)(c3)
\Edge(x3)(c1)
\Edge(x3)(c2)
\Edge(x4)(c1)
\Edge(x5)(c4)

\tikzstyle{EdgeStyle}=[edge_dotted]
\Edge(x1)(c1)
\Edge(x1)(c2)
\Edge(x1)(c3)
\Edge(x2)(c4)
\Edge(x2)(c5)
\Edge(x3)(c3)
\Edge(x4)(c2)
\Edge(x5)(c5)

\node at (0,4) {\Large $\Gamma(F)$};

\end{tikzpicture}

%% file: Figures/sat.tex
\begin{tikzpicture}[scale=1.2]
\SetVertexNoLabel

% draw variable gadgets
\newcounter{i}\setcounter{i}{1}
\foreach \x/\y/\bt/\lr in {5.8/4/0/1, 2/2/0/0, 9.6/2/0/0, 3.8/0/1/0, 7.8/0/1/0}
{
  % \x / \y is the relative position
  % \bt = 1 draws the gadget mirrored bottom/top
  % \lr = 1 draws the gadget mirrored left/right
  \ifthenelse{\lr=0}{\tikzstyle{VertexStyle}=[vertex_bullet]}{\tikzstyle{VertexStyle}=[vertex_square]}
  \ifthenelse{\bt=0}{\def\m{1}}{\def\m{-1}}
  
  \Vertex[x=0+\x,y=0.0*\m+\y]{z\thei1}
  \tikzstyle{VertexStyle}=[vertex_default]
  \Vertex[x=0+\x,y=0.4*\m+\y]{x\thei2}
  \Vertex[x=0+\x,y=0.8*\m+\y]{x\thei3}
  \Vertex[x=0+\x,y=1.2*\m+\y]{x\thei4}
  \Vertex[x=-0.2+\x,y=1.6*\m+\y]{x\thei5}
  \Vertex[x=0+\x,y=2.0*\m+\y]{x\thei6}
  \Vertex[x=0.4+\x,y=2.0*\m+\y]{x\thei7}
  \Vertex[x=0.6+\x,y=1.6*\m+\y]{x\thei8}
  \Vertex[x=0.4+\x,y=1.2*\m+\y]{x\thei9}
  \Vertex[x=0.4+\x,y=0.8*\m+\y]{x\thei10}
  \Vertex[x=0.4+\x,y=0.4*\m+\y]{x\thei11}
  \ifthenelse{\lr=0}{\tikzstyle{VertexStyle}=[vertex_square]}{\tikzstyle{VertexStyle}=[vertex_bullet]}
  \Vertex[x=0.4+\x,y=0.0*\m+\y]{z\thei2}

  \ifthenelse{\lr=0}{\tikzstyle{EdgeStyle}=[edge_default]}{\tikzstyle{EdgeStyle}=[edge_in_c]}
  \Edge(z\thei1)(x\thei2)
  \Edge(x\thei4)(x\thei5)
  \Edge(x\thei7)(x\thei8)
  \Edge(x\thei10)(x\thei11)

  \ifthenelse{\lr=0}{\tikzstyle{EdgeStyle}=[edge_in_c]}{\tikzstyle{EdgeStyle}=[edge_default]}
  \Edge(x\thei2)(x\thei3)
  \Edge(x\thei5)(x\thei6)
  \Edge(x\thei8)(x\thei9)
  \Edge(x\thei11)(z\thei2)

  \tikzstyle{EdgeStyle}=[edge_default]
  \Edge(x\thei3)(x\thei4)
  \Edge(x\thei6)(x\thei7)
  \Edge(x\thei9)(x\thei10)
  \Edge(x\thei4)(x\thei9)

  \stepcounter{i}
}

% draw clause gadgets
\newcounter{j}\setcounter{j}{1}
\foreach \x/\y/\turn/\c in {0.4/2/0/3, 4/2/2/3, 6/2/3/2, 8/2/0/1, 11.6/2/2/1}
{
  % \x / \y is the relative position
  % \turn = 0,1,2,3 are the four possible orientations
  % \c = 0,1,2,3 are the four possible ways of marking an edge as contracted
  \tikzstyle{VertexStyle}=[vertex_default]
  \Vertex[x=0+\x,y=0+\y]{y\thej0}
  \tikzstyle{VertexStyle}=[vertex_bullet]
  \ifthenelse{\not \(\turn=1\)}{\Vertex[x=0+\x,y=0.4+\y]{y\thej1}}{}
  \ifthenelse{\not \(\turn=2\)}{\Vertex[x=0.4+\x,y=0+\y]{y\thej2}}{}
  \ifthenelse{\not \(\turn=3\)}{\Vertex[x=0+\x,y=-0.4+\y]{y\thej3}}{}
  \ifthenelse{\not \(\turn=0\)}{\Vertex[x=-0.4+\x,y=0+\y]{y\thej4}}{}

  \ifthenelse{\c=1}{\tikzstyle{EdgeStyle}=[edge_in_c]}{\tikzstyle{EdgeStyle}=[edge_default]}
  \ifthenelse{\not \(\turn=1\)}{\Edge(y\thej0)(y\thej1)}{}
  \ifthenelse{\c=2}{\tikzstyle{EdgeStyle}=[edge_in_c]}{\tikzstyle{EdgeStyle}=[edge_default]}
  \ifthenelse{\not \(\turn=2\)}{\Edge(y\thej0)(y\thej2)}{}
  \ifthenelse{\c=3}{\tikzstyle{EdgeStyle}=[edge_in_c]}{\tikzstyle{EdgeStyle}=[edge_default]}
  \ifthenelse{\not \(\turn=3\)}{\Edge(y\thej0)(y\thej3)}{}
  \ifthenelse{\c=0}{\tikzstyle{EdgeStyle}=[edge_in_c]}{\tikzstyle{EdgeStyle}=[edge_default]}
  \ifthenelse{\not \(\turn=0\)}{\Edge(y\thej0)(y\thej4)}{}
  
  \stepcounter{j}
}

% add connection edges
\tikzstyle{EdgeStyle}=[edge_default]
\Edge[style={bend right=60}](z11)(y11)
\Edge(z11)(y21)
\Edge(z11)(y31)
\Edge(z12)(y41)
\Edge[style={bend left=60}](z12)(y51)
\Edge(z21)(y12)
\Edge(z22)(y24)
\Edge(z31)(y42)
\Edge(z32)(y54)
\Edge(z41)(y13)
\Edge(z41)(y23)
\Edge(z42)(y34)
\Edge(z51)(y32)
\Edge(z52)(y43)
\Edge(z52)(y53)

\node at (4.8,5.5) {$H(x_1)$};
\node at (7,-1) {$H(x_2)$};
\node at (3,-1) {$H(x_3)$};
\node at (3.3,3) {$H(x_4)$};
\node at (8.8,3) {$H(x_5)$};
\node at (0.7,0.85) {$H(c_1)$};
\node at (4.65,2) {$H(c_2)$};
\node at (6,1.5) {$H(c_3)$};
\node at (7.35,2) {$H(c_4)$};
\node at (11.4,0.85) {$H(c_5)$};

\node at (1,5) {\Large $G(F)$};

\end{tikzpicture}

%% file: Figures/indset_g.tex
\begin{tikzpicture}[scale=0.85]
\SetVertexNoLabel
\tikzstyle{EdgeStyle}=[edge_default]

\tikzstyle{VertexStyle}=[vertex_default]
\Vertex[style={label={above:$u$}},x=0,y=4]{v0}
\Vertex[style={label={above:$w$}},x=4,y=4]{v2}
\Vertex[style={label={below:$y$}},x=2,y=2]{v4}

\tikzstyle{VertexStyle}=[vertex_red]
\Vertex[style={label={above:$v$}},x=2,y=4]{v1}
\Vertex[style={label={below:$x$}},x=0,y=2]{v3}
\Vertex[style={label={below:$z$}},x=4,y=2]{v5}

\Edge(v0)(v1)
\Edge(v0)(v3)
\Edge(v1)(v2)
\Edge(v1)(v4)
\Edge(v2)(v4)
\Edge(v2)(v5)
\Edge(v3)(v4)
\Edge(v4)(v5)

\node at (1,5){\Large $G$};
\node at (0,0){};

\end{tikzpicture}

%% file: Figures/indset_h.tex
\begin{tikzpicture}[scale=0.85]
\tikzstyle{VertexStyle}=[vertex_default]
\SetVertexNoLabel

\Vertex[style={label={left:$u$}},x=8,y=4]{v6}
\Vertex[x=7,y=6]{v7}
\Vertex[x=8,y=6]{v8}

\Vertex[style={label={below left:$v$}},x=10,y=4]{v10}
\Vertex[style={label={above:$(v,1)$}},x=9.4,y=6]{v11}
\Vertex[style={label={above:$(v,2)$}},x=10.6,y=6]{v12}

\Vertex[style={label={right:$w$}},x=12,y=4]{v20}
\Vertex[x=12,y=6]{v21}
\Vertex[x=13,y=6]{v22}

\Vertex[style={label={left:$x$}},x=8,y=2]{v30}
\Vertex[x=7,y=0]{v31}
\Vertex[x=8,y=0]{v32}

\Vertex[style={label={above left:$y$}},x=10,y=2]{v40}
\Vertex[x=9.5,y=0]{v41}
\Vertex[x=10.5,y=0]{v42}

\Vertex[style={label={right:$z$}},x=12,y=2]{v50}
\Vertex[x=12,y=0]{v51}
\Vertex[x=13,y=0]{v52}

\tikzstyle{EdgeStyle}=[edge_default]

\Edge[label=\hbox{$2$}](v6)(v10)
\Edge[label=\hbox{$2$}](v6)(v30)
\Edge[label=\hbox{$2$}](v10)(v20)
\Edge[label=\hbox{$2$}](v10)(v40)
\Edge[label=\hbox{$2$}](v20)(v40)
\Edge[label=\hbox{$2$}](v20)(v50)
\Edge[label=\hbox{$2$}](v30)(v40)
\Edge[label=\hbox{$2$}](v40)(v50)

\Edge[label=\hbox{$1$}](v6)(v7)
\Edge[label=\hbox{$2$}](v6)(v8)
\Edge[label=\hbox{$2$}](v10)(v12)
\Edge[label=\hbox{$1$}](v20)(v21)
\Edge[label=\hbox{$2$}](v20)(v22)
\Edge[label=\hbox{$2$}](v30)(v32)
\Edge[label=\hbox{$1$}](v40)(v41)
\Edge[label=\hbox{$2$}](v40)(v42)
\Edge[label=\hbox{$2$}](v50)(v52)

\tikzstyle{EdgeStyle}=[edge_in_c]
\Edge[label=\hbox{$1$}](v10)(v11)
\Edge[label=\hbox{$1$}](v30)(v31)
\Edge[label=\hbox{$1$}](v50)(v51)

\node at (5.5,3){\Large $H=H(G)$};

\end{tikzpicture}